\documentclass[format=acmsmall, review=false]{acmart}
\usepackage{acm-ec-26}
\usepackage{booktabs} 
\setcitestyle{authoryear}
\usepackage{cleveref}
\usepackage{tcolorbox}
\usepackage{threeparttable}
\usepackage{tabularx, booktabs} 
\usepackage{multirow}
\usepackage{color, colortbl}
\definecolor{Gray}{gray}{0.9}

\newcommand{\smg}{\ensuremath{\mathsf{SM}^\mathsf{g}}\xspace}

\newcommand{\sm}{\ensuremath{\mathsf{SM}}\xspace}
\newcommand{\kp}{\ensuremath{\mathsf{KP}}\xspace}
\newcommand{\kpb}{\ensuremath{\mathsf{KRTV}}\xspace}
\newcommand{\kpc}{\ensuremath{\mathsf{KNR}}\xspace}

\newcommand{\alga}{\ensuremath{\mathsf{SM}}\xspace}

\newcommand{\bam}{\widehat{\cM}}

\newcommand{\tm}{\widetilde{\cM}}
\newcommand{\tce}{\widetilde{\cE}}

\newcommand{\am}{{\cM^{a}}}
\newcommand{\bm}{{\cM^{b}}}

\newcommand{\sda}{\textbf{BODA}\xspace}
\newcommand{\om}{\ensuremath{\operatorname{EW-OM-RO}}\xspace}

\makeatletter
\newcommand{\thickhline}{%
    \noalign {\ifnum 0=`}\fi \hrule height 1pt
    \futurelet \reserved@a \@xhline
}
\newcolumntype{"}{@{\hskip\tabcolsep\vrule width 1pt\hskip\tabcolsep}}
\makeatother
\newcommand{\bta}{\bar{\theta}}
\newcommand{\tta}{\tilde{\theta}}

\newcommand{\revise}[1]{{\color{black} #1}}
\usepackage{macros_pan}
\title{Degree-Parameterized Analysis of Sampling-Based Online Matching}

\author{Pan Xu}
\authornote{\texttt{{\color{blue}tcs.enthusiast@gmail.com}}}

\begin{abstract}
We study edge-weighted online bipartite matching under random arrival order, parameterized by the maximum degree $d$ on the offline side.
Let $\theta$ denote the sampling fraction.
Classical sampling-based algorithms for this problem are typically analyzed in the unbounded-degree regime.
We give a degree-parameterized analysis of two sampling-based frameworks, showing how their guarantees vary with both $d$ and $\theta$.

The first framework, \textbf{Deterministic Greedy Sampling}, computes offline prices from a fixed-size initial sample and then applies a local threshold rule to each subsequent arrival.
We characterize its worst-case competitiveness by an explicit function of $d$ and $\theta$, and show that this characterization is tight within this policy family for every fixed $d\ge2$ and every $\theta\in[0,1]$.
The optimal tuning interpolates between greedy behavior in very sparse instances and sample-driven behavior in dense instances: for $d=1,2$, no sampling is optimal, while in the unbounded-degree limit the optimized guarantee is approximately $0.2562$.
This strictly improves over the classical $1/8$ guarantee for the corresponding greedy sampling analysis of Korula and Pál (ICALP 2009), while retaining linear per-arrival time.

We also analyze the variance of the total number of matched offline agents produced by Deterministic Greedy Sampling.
This gives a second-moment analysis of allocation-volume variability under random arrival, distinct from weighted-reward variance.
The proof controls aggregate covariance among matching indicators despite the absence of negative dependence.
We derive explicit worst-case upper and lower bounds, show that the upper bound is order-tight as $\theta\to1_-$, and establish that, for fixed $m$ and $d$, the worst-case match-count variance converges to zero at both $\theta\to0_+$ and $\theta\to1_-$; in particular, it is $O(\theta)$ as $\theta\to0_+$.

The second framework, \textbf{Black-Box Sampling--Matching}, applies a prefix-dependent reweighting and then invokes an arbitrary approximate offline matching algorithm on each evolving prefix.
Under a natural arrival-order-independence condition on the black-box solver, we prove a transfer theorem showing that the online competitiveness is the offline approximation ratio multiplied by an explicit factor depending on $d$ and $\theta$.
With exact matching, this framework recovers the classical $1/e$ guarantee in the unbounded-degree limit; with approximate solvers, it quantifies the tradeoff between solver quality and per-arrival computational cost.

Together, these results give a degree-parameterized account of how sampling, bounded offline degree, solver choice, and prefix reweighting shape the expected weighted performance, online efficiency, and allocation-volume variability of sampling-based online matching algorithms.
\end{abstract}

\begin{document}
\begin{titlepage}
\maketitle
\setcounter{tocdepth}{1} 
\tableofcontents

\end{titlepage}

\newpage

\setcounter{page}{1}

\section{Introduction}

\revise{Online allocation systems frequently make irrevocable assignment decisions under limited information about future arrivals.}
Examples include worker--task assignment in crowdsourcing platforms, driver--rider assignment in ride-hailing and delivery markets, and impression--advertiser assignment in online advertising.
\revise{In these settings, a platform observes requests, workers, or impressions sequentially, and must decide in real time whether and how to assign each arriving online agent to one of a limited set of offline resources.}

Online matching and its variants have emerged as a central abstraction for studying this challenge, capturing the inherent tension between irrevocable decisions and incomplete future information~\cite{zhou2021primal,dickerson2018assigning,nanda2020balancing,xu2020trade}.
Most prior work in this area focuses on optimizing expected performance, typically measured by the expected total matching weight or welfare.
\revise{This expectation-based benchmark is fundamental, but it does not describe how much the realized outcome fluctuates across arrival orders or internal random choices of the policy.}
\revise{For applications in which allocation volume itself is operationally meaningful---for instance, the number of served tasks, assigned workers, or utilized resources---it is also natural to ask whether a policy produces stable realized matching sizes.}
\revise{This motivates a complementary, second-moment perspective: in addition to expected weighted performance, we study the variability of the number of matches produced by an online policy.}

\revise{We study this question in the classical model of \emph{edge-weighted online matching under random arrival order}, with a particular focus on how the maximum offline degree affects the performance of sampling-based policies.}

\xhdr{Edge-Weighted Online Matching under Random Order} (\om).
Consider a bipartite graph $G = (I, J, E)$, where $I$ denotes a fixed set of \emph{offline agents} and $J$ denotes a set of \emph{online agents} arriving sequentially.
An edge $e = (i, j) \in E$ represents feasibility or mutual interest between offline agent $i \in I$ and online agent $j \in J$, and is associated with a non-negative weight $w(e)$.\footnote{For example, in online advertising, an edge between advertiser $i$ and impression $j$ indicates that advertiser $i$ is interested in showing an ad to impression $j$.}
For notational convenience, let $I = [m] := \{1, \ldots, m\}$ and $J = [n] := \{1, \ldots, n\}$, and use $i$ and $j$ to index offline and online agents, respectively.

At each round, a single online agent $j \in J$ arrives, revealing the incident edge set $E_j$ and the corresponding weights $\{ w(e) : e \in E_j \}$.
Upon arrival, the algorithm must irrevocably decide whether to match $j$ to at most one currently unmatched offline agent $i \in I$, receiving an immediate reward of $w(i,j)$ if a match is made.
The online agents arrive in a uniformly random order, while the offline agents and the total number of arrivals $n$ are known in advance.
The objective is to design an online matching policy that maximizes the expected total weight of the resulting matching.

\xhdr{\revise{Bounded Offline Degree.}}
\revise{The structural parameter studied in this paper is the maximum degree on the offline side.}
Formally, we assume that the degree of every offline agent is bounded by a parameter $d$, which is known to the algorithm.
\revise{Equivalently, each offline resource has only a limited local neighborhood within the full arrival population.}
\revise{This assumption should be interpreted as a sparsity condition on the realized compatibility graph, rather than as a claim that a platform artificially truncates an otherwise dense graph.}
\revise{It captures settings in which compatibility is intrinsically limited by geography, skills, eligibility, targeting constraints, or type-specific feasibility.}

Such bounded-degree structures arise naturally across a range of online marketplaces.
In crowdsourcing systems, for example, tasks are typically static while workers arrive dynamically~\cite{DBLP:journals/corr/abs-2002-10697,ho2012online,Manshadi2020OnlinePF}.
Compatibility between a worker and a task is determined by skill requirements, certifications, or interest profiles, so that each task is effectively compatible with only a small subset of worker types, even when the total number of arriving workers is large.

A similar phenomenon occurs in ride-hailing platforms~\cite{ma2021fairness,teac21}, where drivers (offline agents) are associated with specific locations and service regions, and riders (online agents) are characterized by origin--destination pairs.\footnote{This modeling choice is typically justified over short time horizons, such as peak hours, during which the pool of available drivers remains relatively stable while riders arrive dynamically.}
Geographic proximity and service-area constraints ensure that each driver is compatible with only a small number of rider types.
Related bounded-degree patterns also appear in online advertising and assortment optimization, where targeting criteria or preference constraints restrict each advertisement or product assortment to a limited segment of users~\cite{goyal2020online,gong2021online,feng2019linear}.
\revise{Across these domains, bounded offline degree is best viewed as a structural abstraction for local or type-restricted compatibility, rather than as the outcome of arbitrary pre-filtering.}

Motivated by these considerations, we formalize this market structure assumption as follows.
\begin{tcolorbox}
\textbf{Bounded-Offline-Degree Assumption} (\sda).
Every offline agent has degree at most $d$ in the bipartite graph $G = (I, J, E)$, where $n = |J|$ denotes the total number of online agents.
Both $d$ and $n$ are accessible to the algorithm.\footnotemark  
~{Unless stated otherwise, asymptotic guarantees are interpreted in the regime $n \to \infty$ with $d$ treated as a structural parameter independent of $n$.}
\end{tcolorbox}

\footnotetext{
Access to the total number of arrivals $n$ is essential in the edge-weighted online matching setting; without it, no online policy can achieve a nontrivial competitive ratio (see Appendix~\ref{app:15-a}).
\revise{The degree parameter $d$ can be viewed as known from the compatibility model or estimated from historical data; the explicit dependence of our guarantees on $d$ also permits sensitivity analysis when $d$ is misspecified.}
}

\subsection{Effectiveness--Efficiency Tradeoffs in Sampling--Matching Algorithms}

For edge-weighted online matching under random arrival order, a standard algorithmic paradigm is \emph{sampling followed by online matching}, inspired by the classical Secretary Problem.
Such algorithms use an initial \textbf{Sampling Phase} to learn threshold, price, or matching information, and then use this information in a subsequent \textbf{Matching Phase} to process each arriving online agent irrevocably.

\xhdr{Performance Metrics: Effectiveness and Online Efficiency.}
We evaluate sampling--matching policies using two standard criteria.
The first is \emph{competitiveness}: the ratio between the expected total weight achieved by the online policy and the total weight of an optimal offline matching on the same instance.
This measures effectiveness relative to a prophet benchmark with full future information.
The second is \emph{online time complexity}: the worst-case computational cost incurred per arrival during the matching phase.
As is standard, we separate this per-arrival cost from the one-time computation performed after the sampling phase.
These two criteria expose a basic tradeoff: policies that repeatedly solve global matching problems on evolving prefix graphs can obtain stronger worst-case guarantees, while local threshold- or price-based policies admit faster online decisions but typically weaker guarantees.

\xhdr{Two Canonical Frameworks as Extreme Points.}
This tradeoff is illustrated by two classical sampling--matching frameworks.

\xhdr{$\kpb$~\cite{kess13}: Prefix Optimization.}
\citet{kess13} proposed a sampling--matching framework, denoted by $\kpb$, that achieves the optimal $1/\sfe$ competitiveness guarantee for \om without structural assumptions on the graph.\footnote{In this setting, an offline agent may be adjacent to all $n$ online agents. The optimality of the $1/\sfe$ bound follows because \om generalizes the classical Secretary Problem.}
The algorithm samples the first $K=\lfloor n/\sfe \rfloor$ arrivals and, for each subsequent arrival $j_t$, recomputes an optimal matching on the prefix graph $(I,\mathcal{S}_t)$; it assigns $j_t$ if the offline vertex matched to $j_t$ in this prefix optimum remains available.
Thus $\kpb$ achieves optimal worst-case effectiveness in the unstructured regime, but its online rule is global: each arrival triggers a matching computation on the current prefix graph.
Under standard static implementations, this costs $O((|I|+|J|)^3)$ per arrival using the Hungarian method~\cite{kuhn1955hungarian}, or $O(|E|^{3/2}\log |E|)$ via cost-scaling flow algorithms~\cite{goldberg1990finding}.
Moreover, the classical guarantee is degree-oblivious; it does not describe how performance changes when the offline degree is bounded by a structural parameter $d$.

\xhdr{$\kp$~\cite{korula2009algorithms}: Local Price-Based Decisions.}
At the other end of the spectrum, \citet{korula2009algorithms} introduced a lightweight framework, denoted by $\kp$, that achieves a $1/8$ competitiveness guarantee with linear online time.
The algorithm draws a random sample size $K\sim\mathrm{Bin}(n,1/2)$, runs Greedy on the sampled subgraph $(I,\mathcal{S}_K)$, and uses the resulting matching to define a price $p(i)$ for each offline agent. In the online phase, for each arriving agent, the algorithm first selects a maximum-weight incident edge whose weight exceeds the corresponding offline price. If the offline endpoint of this edge is available, the agent is matched along this edge; otherwise, the agent is rejected. This rule inspects only the incident edges of the arriving agent, yielding online time $O(\deg(j_t))$, and hence at most $O(|I|)$, per arrival.\footnote{As usual, the one-time cost of computing the sample matching is excluded from the online time complexity.}
However, like $\kpb$, the $\kp$ analysis is degree-oblivious and does not adapt its guarantee to bounded offline degree.

Our results revisit this effectiveness--efficiency tradeoff through the lens of bounded offline degree.
The first framework preserves the local threshold-based rule of $\kp$, but tunes the sampling fraction to $d$ and uses a fixed sample size. The second retains the prefix-optimization structure of $\kpb$, but introduces
a prefix-dependent reweighting and replaces exact matching with a black-box
offline approximation procedure, yielding a degree-parameterized transfer theorem.

\subsection{Variance Analysis for Online Matching: A Metric for Algorithmic Robustness}
\label{sec:pre}

Much of the online matching literature evaluates algorithms through expectation-based guarantees. 
Such guarantees are fundamental, but they do not describe how much the realized outcome fluctuates across random arrival orders or internal random choices of the policy. 
This issue is especially relevant in online settings: unlike offline randomized algorithms, which can often be derandomized or repeated to average out fluctuations~\cite{raghavan1988probabilistic}, an online algorithm acts only once on each arriving instance. 
This motivates a second-moment perspective on online matching.

Recent work has begun to study robustness in online algorithms from this viewpoint. 
\citet{xu2022exploring,xu2024tight} analyzed variance of the total number of matches under Known Independent and Identical Distributions (KIID) and Known Heterogeneous Distributions (KHD), respectively. 
These models differ from the random-arrival model considered here. 
In a related direction, \citet{dinitz2024controlling} studied robustness in ski rental through tail-risk guarantees. 
Variance-based analysis is complementary to such tail guarantees: once a performance metric is fixed, its variance gives a parameter-free measure of dispersion, and standard inequalities such as Chebyshev's inequality can convert variance bounds into tail bounds.

\xhdr{Two Sources of Online Uncertainty.} 
In \om, the input graph and all edge weights are fixed. 
Randomness enters through the uniformly random arrival order and, when present, the algorithm's internal randomization. 
We refer to these two stochastic sources as \emph{online uncertainty}. 
Our variance analysis asks how this online uncertainty propagates to variability in the realized outcome.

\xhdr{Match-Count Variance vs.\ Weighted-Reward Variance.}
In an edge-weighted model, a realized matching can be evaluated either by its total reward or by its cardinality. 
We use total reward for expected-performance and competitiveness analysis. 
For variance analysis, however, weighted reward mixes the effect of online uncertainty with the deterministic weight profile of the selected edges. 
Thus weighted-reward variance may reflect not only variability in allocation volume, but also which deterministic weights are selected and how large they are.

This distinction is already visible in the classical Secretary Problem. 
Suppose one candidate has value $w>0$ and the remaining $n-1$ candidates have value $0$. 
The optimal deterministic policy selects the high-value candidate with probability asymptotically $1/\sfe$, so its reward variance is asymptotically $(1-1/\sfe)(1/\sfe)w^2$. 
As $w$ varies, the arrival-order uncertainty is unchanged, but the reward variance grows quadratically in $w$. 
Thus weighted-reward variance can be dominated by the deterministic magnitude profile of the instance.

The obstruction is not only the possibility of very large weights. 
It persists even when all weights lie in $[0,1]$. 
Consider one offline agent and two feasible online agents with edge weights $1$ and $0$. 
A policy that accepts the first feasible arrival always produces exactly one match, so the match-count variance is zero. 
Under random arrival order, however, the total reward is Bernoulli with mean $1/2$ and variance $1/4$. Hence weighted-reward variance can be positive while match-count variance is
zero, ruling out any finite multiplicative bound of the form
$\operatorname{Var}(\text{reward}) \le C \cdot
\operatorname{Var}(\text{match count})$.

Appendix~\ref{app:why} gives a complementary counterexample showing that this failure
persists even when the match count itself has positive variance: the variance
of the total reward can exceed the match-count variance multiplied by the
square of the maximum edge weight.

\xhdr{Variance of Matches as a Robustness Proxy.}
For this reason, we analyze the variance of the \emph{total number of matches}. 
This choice is methodological rather than merely technical: it projects the realized weighted matching onto its cardinality, thereby factoring out deterministic weight-composition effects and isolating the contribution of online uncertainty to allocation-volume variability. 
For a fixed realized matching, cardinality is insensitive to the magnitudes of the matched edge weights. 
We therefore view match-count variance as a principled robustness metric for sampling-based online matching, complementary to expected weighted performance. 
Extending the analysis to weighted-reward variance would require additional structural assumptions linking allocation size, edge-weight composition, and reward magnitude, and remains open.

\begin{table}[th!]
\center
\caption{Comparison of the two sampling--matching frameworks proposed in this paper, $\smg(\theta)$ and $\alga(\BB,\theta)$, with existing algorithms from~\cite{korula2009algorithms,kess13,kaplan2022online}. 
Here, $\theta\in[0,1]$ is the sampling fraction, $\BB$ denotes a black-box $\alpha$-approximation procedure for maximum-weight bipartite matching, $K$ is the number of sampled arrivals, $d$ is the maximum offline degree, and $\mathsf{T}_{\mathrm{online}}$ denotes the worst-case per-arrival runtime during the online phase.}

\medskip
\renewcommand{\arraystretch}{1.5}
\resizebox{\textwidth}{!}{
\begin{tabular}{c|c|c|c|c|c}
\hline
\hline
 & $K$ & Setting of $d$ & Subroutine & $\mathsf{T}_{\text{online}}$ & Competitive Ratio \\
\hline \thickhline

&&  & \textsc{TMW-Select}\textsuperscript{\bluee{$\ast$}}&  &  \\
\multirow{-2}{*}{$\kp$~\cite{korula2009algorithms}} & \multirow{-2}{*}{$\mathrm{Bin}(n, 1/2)$} & \multirow{-2}{*}{$\infty$} & On $(I, j_t)$ & \multirow{-2}{*}{$O(|I|)$\textsuperscript{\bluee{\S}}} & \multirow{-2}{*}{$1/8$}  \\ 
\hline

&& \multirow{2}{*}{$\infty$} & Greedy & \multirow{2}{*}{$O(|E| \, \log |E|)$} & \multirow{2}{*}{$0.1715$} \\
\multirow{-2}{*}{$\kpc$~\cite{kaplan2022online}} & \multirow{-2}{*}{$\mathrm{Bin}(n, \sqrt{2}-1)$} & & On $(I, \cS_t)$ & & \\ 
\hline

\multirow{2}{*}{$\kpb$~\cite{kess13}} & \multirow{2}{*}{$\lfloor n/\sfe \rfloor$} & \multirow{2}{*}{$\infty$} & Optimal & \multirow{2}{*}{$O\sbp{(|I|+|J|)^3}$\textsuperscript{\bluee{\dag}}} & \multirow{2}{*}{$1/\sfe \approx 0.3678$} \\
& & & On $(I, \cS_t)$ & & \\
\hline \thickhline

\rowcolor{Gray} && & \textsc{TMW-Select} & & $\sig(d,\theta)$~\bluee{\eqref{eqn:sig}} \\
\rowcolor{Gray} \multirow{-2}{*}{$\smg(\ta)$} & \multirow{-2}{*}{$\lfloor n \, \ta \rfloor$} &  \multirow{-2}{*}{$d$ (Param.)} & On $(I, j_t)$ & \multirow{-2}{*}{$O(|I|)$} & $\max_\ta \sig(\infty, \ta) \sim 0.2562$ \\
\hline

\rowcolor{Gray} &&  & Reweighted $\BB$ & & $\alp \, \kap(d,\theta)$~\bluee{\eqref{eqn:kap}} \\
\rowcolor{Gray} \multirow{-2}{*}{$\alga(\BB, \ta)$} & \multirow{-2}{*}{$\lfloor n \, \ta \rfloor$} &   \multirow{-2}{*}{$d$ (Param.)} & On $(I, \cS_t)$ &   \multirow{-2}{*}{$\BB$-Dependent} & $\max_\ta \kap(\infty, \ta) = 1/\sfe$ \\
\hline \thickhline
\end{tabular}
}

\vspace{1em}
\parbox{\textwidth}{\footnotesize%
\textsuperscript{\bluee{$\ast$}}
\textsc{TMW-Select} refers to the \emph{Thresholded Max-Weight Rule}, which is used by both $\kp$ and $\smg(\ta)$ as an online subroutine for making real-time assignment decisions. Given an arriving online agent $j$ and a price vector $p : I \to \mathbb{R}_{\ge 0}$ over offline agents, the rule selects a neighbor $i^* \in I$ satisfying $w(i^*,j) \ge p(i^*)$ that maximizes $w(i,j)$ among all such neighbors. If no such neighbor exists, then $j$ is rejected. Otherwise, if $i^*$ is available (i.e., unmatched), then $j$ is assigned to $i^*$; if $i^*$ is already matched, then $j$ is rejected.

\quad
\textsuperscript{\bluee{\S}}
The online time complexity of both $\kp$ and $\smg$ is
$O(\max_{j \in J} \deg(j))$, i.e., linear in the maximum degree among all online agents in $J$.

\quad
\textsuperscript{\bluee{\dag}}
Alternatively, the running time can be expressed in terms of the number of edges as
$O(|E|^{3/2} \cdot \log |E|)$, based on the classical cost-scaling min-cost flow algorithm~\cite{goldberg1990finding}.
}
\label{table:coma}
\end{table}


\section{Main Contributions}
\label{sec:main-con}

This paper gives a degree-parameterized analysis of sampling-based algorithms for edge-weighted online matching under random arrival order.
The central question is how a bounded offline degree $d$ changes the guarantees of classical sampling--matching policies as a function of the sampling fraction $\theta$.
Under the bounded-offline-degree assumption (\sda), we study two complementary frameworks.

The first framework, \emph{Deterministic Greedy Sampling} ($\smg$), computes offline prices from an initial sample and then applies a local threshold rule to each arriving online agent. For this framework, we obtain a tight degree-parameterized competitiveness
characterization together with explicit worst-case bounds on the variance of
the number of matched offline agents, including order-tight behavior near
full sampling and a linear-in-$\theta$ upper bound in the low-sampling regime. The second framework, \emph{Black-Box Sampling--Matching} ($\sm(\BB,\theta)$), invokes an $\alpha$-approximate offline matching procedure on each evolving prefix graph.
For this framework, we prove a transfer theorem from offline approximation to online competitiveness.

A comparative summary of $\smg$, $\sm(\BB,\theta)$, and the classical $\kp$ and $\kpb$ algorithms is provided in Table~\ref{table:coma}. The table highlights the two main algorithmic contrasts: $\smg$ preserves the
local price-based online decisions of KP while tuning the sampling fraction to
$d$ and using a fixed sample size, whereas $\sm(\BB,\theta)$ retains the
prefix-optimization structure of $\kpb$ while introducing prefix-dependent
reweighting and allowing approximate offline solvers.

We use the following functions to state the guarantees:
\begingroup
\allowdisplaybreaks
\begin{align}
\kappa(d, \theta)
&:= \int_{\theta}^1  
\Bigl(
(1 - \zeta)^{d-1}
+ \bigl(1 - (1 - \zeta)^{d-1}\bigr)\frac{\theta}{\zeta}
\Bigr) \, \sd \zeta,
\label{eqn:kap}
\\
\eta(d, \theta)
&:= \frac{(1 - \theta)\theta}{1 - (1 - \theta)^d}
= \frac{1-\theta}{1 + (1-\theta) + \cdots + (1-\theta)^{d-1}},
\label{eqn:eta}
\\
\sigma(d, \theta)
&:= \kappa(d, \theta) - \frac{\theta(1 - \theta)}{2}.
\label{eqn:sig}
\end{align}
\endgroup
Here $\sigma(d,\theta)$ gives the tight competitiveness of $\smg(\theta)$, $\kappa(d,\theta)$ governs the black-box transfer theorem, and $\eta(d,\theta)$ appears in an intermediate diagnostic bound.
Lemma~\ref{lem:sig} in Appendix~\ref{app:lem-sig} collects the basic properties of these functions used below, including comparison inequalities and uniqueness of the relevant sampling optimizers.

\subsection{Degree-Parameterized Effectiveness of $\smg$}

Our first result characterizes the worst-case competitiveness of $\smg(\theta)$ exactly within this policy family.

\begin{theorem}[Section~\ref{sec:smg}]
\label{thm:smg-cr-tig}
Fix $\theta\in[0,1]$ and consider the fixed-$d$ asymptotic regime $n\to\infty$.
The deterministic greedy sampling framework $\smg(\theta)$ has linear online time complexity and is $\sigma(d,\theta)$-competitive, where $d$ is the maximum offline degree and $\sigma(d,\theta)$ is defined in~\eqref{eqn:sig}.
Moreover, for every fixed integer $d\ge2$ and every $\theta\in[0,1]$, this guarantee is tight within the $\smg(\theta)$ policy family: over bipartite graphs with offline degree at most $d$, the worst-case competitiveness of $\smg(\theta)$ is exactly $\sigma(d,\theta)$.\footnote{Throughout the paper, statements of tightness for $\smg(\theta)$ are with respect to this fixed policy family. They do not imply optimality among all online algorithms for degree-$d$ instances. The same convention applies to the variance analysis in Section~\ref{sec:asy}.}
\end{theorem}

Let $\theta^*_{\sigma}(d)$ denote the unique maximizer of $\sigma(d,\theta)$ over $\theta\in[0,1]$, and let $\sigma^*(d):=\sigma(d,\theta^*_{\sigma}(d))$.
The optimizer describes how sampling should vary with the offline degree.
For $d=1,2$, the optimal choice is $\theta^*_{\sigma}(d)=0$, yielding $\sigma^*(1)=1$ and $\sigma^*(2)=1/2$; thus, in the lowest-degree regimes, the optimal member of the $\smg$ family reduces to a greedy rule.
As $d$ grows, sampling becomes necessary to resolve competition among offline agents.
In the unbounded-degree limit, the competitiveness of $\smg$ admits the following closed-form expressions:
\begin{align*}
\kappa(\infty,\theta)
&= -\theta \ln \theta,
\qquad
\sigma(\infty,\theta)
= -\theta \ln \theta - \frac{\theta(1-\theta)}{2},
\\
\theta^*_{\sigma}(\infty)
&= -W\!\left(-\sfe^{-3/2}\right) \approx 0.3017,
\\
\sigma^*(\infty)
&= \frac{1}{2} W\!\left(-\sfe^{-3/2}\right)
\Bigl(
1 + 2 \ln\bigl(-W(-\sfe^{-3/2})\bigr)
+ W\!\left(-\sfe^{-3/2}\right)
\Bigr)
\approx 0.2562,
\end{align*}
where $W(\cdot)$ denotes the Lambert $W$ function.\footnote{For any $x \in [-\sfe^{-1},0]$, $W(x)$ is the unique real number $z \in [-1,0]$ such that $z \cdot \sfe^{z} = x$.}
Thus the optimized dense-limit guarantee is approximately $0.2562$, strictly above the classical $1/8$ guarantee for the corresponding greedy sampling analysis of~\citet{korula2009algorithms}, while preserving linear per-arrival time.
This limiting regime also illustrates the structural transition in the optimal sampling choice: $\theta=0$ is optimal in very sparse instances, whereas a positive sampling fraction is optimal as local competition increases.

The proof proceeds in two steps.
First, we give a diagnostic generalization of the analysis of~\citet{korula2009algorithms}, yielding the intermediate bound $\eta(d,\theta)/2$.
This diagnostic analysis identifies two inequalities behind the classical argument and shows where they lose tightness.
Second, we show that the tight examples for these two inequalities have incompatible local structures, and replace the two-step comparison by a direct local worst-case analysis.
This yields the exact expression $\sigma(d,\theta)$ for all $\theta\in[0,1]$ and all $d\ge2$.
Unless stated otherwise, the displayed formulas are in the fixed-$d$ asymptotic regime $n\to\infty$; finite-$n$ expressions are obtained by retaining the corresponding finite sums and products in the proofs.

\subsection{Worst-Case Variance Analysis for $\smg$}
\label{sec:asy}

We next analyze the match-count variability of $\smg$.
Consistent with Section~\ref{sec:pre}, the robustness metric is the
variance of the total number of matched offline agents, with randomness
taken over the random arrival order and any internal randomization of the
policy.
This metric measures allocation-volume variability rather than
weighted-reward variability.

For fixed $m$, $d$, and $\theta$, let $\psi_{m,d}(\theta)$ denote the
worst-case variance of $\smg(\theta)$ over instances with $m$ offline
agents and maximum offline degree at most $d$.
Formally, if $\cI(m,d)$ denotes this class of instances and
$|\smg(\theta,G)|$ denotes the number of matches produced on instance
$G$, then
\begin{align}
\psi_{m,d}(\theta)
:=
\sup_{G\in\cI(m,d)}
\Var\bigl[\,|\smg(\theta,G)|\,\bigr].
\label{def:var-psi}
\end{align}

\begin{theorem}[Section~\ref{sec:smg-var} and Appendix~\ref{app:smg-var-tig}]
\label{thm:smg-var-new}
For any $m$, $d$, and $\theta\in[0,1]$,
\[
\underline{\psi}_{m,d}(\theta)
\;\le\;
\psi_{m,d}(\theta)
\;\le\;
\overline{\psi}_{m,d}(\theta),
\]
where
\begin{align}
\overline{\psi}_{m,d}(\theta)
&:=
m\Bigl(
\tilde{\theta}(1-\tilde{\theta})
+
2\min\Bigl\{
\frac{\bar{\theta}^{4}}{\theta^{2}},
\;
(d+2)\bar{\theta}^{7/2}
\Bigr\}
\Bigr),\\
\underline{\psi}_{m,d}(\theta)
&:=
m\,\theta(1-\theta)
\Bigl(
1-\frac{1}{d}+O(1/n)
\Bigr),
\end{align}
with
\[
\tilde{\theta}:=\max\{1/2,\theta\},
\qquad
\bar{\theta}:=1-\theta.
\]
The $O(1/n)$ term is interpreted in the fixed-$d$ asymptotic regime.
\end{theorem}

Theorem~\ref{thm:smg-var-new} gives an explicit worst-case upper bound
on the match-count variance of $\smg(\theta)$, together with a general
lower-bound construction.
The two endpoint regimes exhibit different behavior.

Near full sampling, the upper bound is order-tight.
Specifically, Lemma~\ref{lem:ws-va} in Appendix~\ref{app:ws-va} shows
that, for every fixed $d\ge2$,
\[
\psi_{m,d}(\theta)
=
\Theta\bigl(\overline{\psi}_{m,d}(\theta)\bigr)
\qquad
\text{as }\theta\to1_-.
\]
In particular,
\[
\psi_{m,d}(\theta)=\Theta\bigl(m(1-\theta)\bigr)
\qquad
\text{as }\theta\to1_-,
\]
in the fixed-$d$ asymptotic regime.

The low-sampling regime requires a different argument.
The general covariance bound in
$\overline{\psi}_{m,d}(\theta)$ does not vanish as
$\theta\to0_+$ and is therefore not tight in this regime.
Appendix~\ref{app:ws-va} establishes the additional bound
\[
\psi_{m,d}(\theta)
\le
m^3d\,\theta,
\]
which holds uniformly over all instances in $\cI(m,d)$.
Consequently, for every fixed $m$ and $d$,
\[
\lim_{\theta\to0_+}\psi_{m,d}(\theta)
=
\psi_{m,d}(0)
=
0.
\]
Moreover, for every fixed $m$ and $d\ge2$, combining this upper bound
with the lower bound in Theorem~\ref{thm:smg-var-new} gives
\[
\psi_{m,d}(\theta)
=
\Theta(\theta)
\qquad
\text{as }\theta\to0_+,
\]
where the implicit constants may depend on $m$ and $d$.

At the other endpoint,
\[
\lim_{\theta\to1_-}\psi_{m,d}(\theta)
=
\psi_{m,d}(1)
=
0.
\]
Thus $\psi_{m,d}(\theta)$ is one-sided continuous at both
$\theta=0$ and $\theta=1$.
The mechanisms at the two endpoints are nevertheless different.
When $\theta=1$, the sampling phase consumes all arrivals and
$\smg(1)$ deterministically outputs an empty matching.
When $\theta=0$, the sample is empty and all offline prices are zero;
under the fixed tie-breaking rule, the resulting set of matched offline
agents is deterministic, although the identities of their matched online
partners may depend on the arrival order.
For small positive $\theta$, randomness enters through the possibility
that one or more edge-bearing online agents are included in the sample,
and Appendix~\ref{app:ws-va} shows that the resulting worst-case
match-count variance vanishes at least linearly with $\theta$ for fixed
$m$ and $d$.

These statements should be interpreted as worst-case guarantees for the
fixed policy family $\smg(\theta)$, rather than as a minimax optimality
result over all online policies.
Variance minimization without an effectiveness requirement is trivial:
a policy that rejects every arrival has zero variance.
The relevant interpretation is therefore bi-criteria:
$\smg(\theta)$ achieves the competitiveness guarantee of
Theorem~\ref{thm:smg-cr-tig} while admitting explicit worst-case control
of match-count variability.

\xhdr{Deterministic Sampling as a Variance-Reduction Mechanism.}
To highlight the role of deterministic sampling as an algorithmic design
choice, let $\kp(\theta)$ denote the parameterized variant of the policy
proposed by~\citet{korula2009algorithms}, identical to $\smg(\theta)$
except that it selects a random sample size
$K\sim\mathrm{Bin}(n,\theta)$ rather than fixing
$K=\lfloor n\theta\rfloor$.

Consider an instance consisting of $|I|=|J|=m$ disjoint unit-weight
edges, and assume for simplicity that $m\theta$ is an integer.
Both $\smg(\theta)$ and $\kp(\theta)$ yield an expected number of
$m(1-\theta)$ matches.
However, their match-count variance differs sharply:
$\smg(\theta)$ always produces exactly $m(1-\theta)$ matches and hence
has zero match-count variance, whereas $\kp(\theta)$ produces
\[
m-K
=
m-\mathrm{Bin}(m,\theta),
\]
which has variance $m\theta(1-\theta)$.
Thus randomizing the sample size can introduce allocation-volume
variance even in an instance with no matching conflict, while fixing the
sample size removes this source of variability without changing the
expected number of matches in the example.

The proof of Theorem~\ref{thm:smg-var-new} must also handle dependence
among offline matching indicators.
These indicators need not be negatively dependent;
Appendix~\ref{app:neg-dep} gives explicit examples with positive
covariance.
The main technical step is therefore an aggregate covariance bound,
which controls the total pairwise covariance despite possible local
positive correlations.
\subsection{A Black-Box Sampling--Matching Framework}

Our second framework builds on the prefix-optimization paradigm of
$\kpb$ while introducing two modifications.
First, instead of requiring an exact maximum-weight matching computation
on every evolving prefix, we allow an arbitrary $\alpha$-approximation
procedure $\BB$ as a black-box offline solver.
Second, before invoking $\BB$, we apply a prefix-dependent reweighting to
the observed edges.
The reweighting reflects a conditional lower bound on the probability
that the offline endpoint of a tentative edge remains available when the
edge is proposed.
Together, these two ingredients separate the online sampling analysis
from the choice of offline solver: the resulting competitive guarantee
factors into the approximation ratio $\alpha$ of $\BB$ and the
degree-dependent sampling term $\kappa(d,\theta)$.

An additional feature of the black-box implementation requires particular
care.
For every fixed unordered prefix, the matching returned by $\BB$ must be
independent of the arrival order of the online agents within that prefix.
In particular, $\BB$ must not use the identity of the current arrival as
a distinguished vertex, the previous arrival order, or the current online
matching when selecting among candidate prefix matchings.
For a randomized black box, its internal random seeds are taken to be
independent of the arrival order.
\emph{This arrival-order-independence requirement is not specific to our
reweighting technique: it is also necessary for the classical
prefix-optimization framework underlying $\kpb$}.
Appendix~\ref{sec:amend} gives a concrete example showing that, without
this requirement, even an exact optimizer on every prefix can lead to
vanishing online competitiveness.

The framework is conceptually distinct from $\smg$.
While $\smg$ computes prices once from the initial sample and subsequently
makes local threshold-based decisions, $\sm(\BB,\theta)$ repeatedly
solves a reweighted global matching problem on the current prefix graph.
Even when $\BB$ is Greedy, $\sm(\BB,\theta)$ is not equivalent to
$\smg(\theta)$; Appendix~\ref{app:comp} gives explicit examples.

\begin{theorem}[A Black-Box Sampling--Matching Transfer Theorem;
Section~\ref{sec:sda}]
\label{thm:main-2}
Fix $\theta\in[0,1]$ and consider the fixed-$d$ asymptotic regime
$n\to\infty$.
Suppose that $\BB$ is an $\alpha$-approximation procedure for
maximum-weight matching on edge-weighted bipartite graphs and satisfies
the arrival-order-independence convention specified in
Section~\ref{sec:sda}.
Then $\sm(\BB,\theta)$ has asymptotic competitiveness guarantee
\[
    \alpha\,\kappa(d,\theta),
\]
where $\kappa(d,\theta)$ is defined in~\eqref{eqn:kap}.
\end{theorem}

By Lemma~\ref{lem:sig}, for each fixed $d\ge1$, the function
$\kappa(d,\theta)$ admits a unique maximizer over $\theta\in[0,1]$,
denoted by $\theta^*_{\kappa}(d)$, and we write
\[
\kappa^*(d):=
\kappa\bigl(d,\theta^*_{\kappa}(d)\bigr).
\]
Thus, the sampling fraction $\theta$ can be tuned to the degree bound
$d$, while $\alpha$ quantifies the performance loss incurred by the
chosen offline solver.

For exact matching and in the unbounded-degree limit, we have
\[
    \kappa(\infty,\theta)=-\theta\ln\theta,
    \qquad
    \theta^*_{\kappa}(\infty)=1/\sfe,
    \qquad
    \kappa^*(\infty)=1/\sfe.
\]
Hence, with an exact black box and $\theta=1/\sfe$, our framework
recovers the classical $1/\sfe$ competitiveness guarantee of $\kpb$.
The algorithms themselves are not identical in general, since our
framework additionally applies the prefix-dependent reweighting.
For finite $d$, the same expression quantifies the benefit of bounded
offline degree; for example,
\[
\kappa^*(5)\approx0.372>1/\sfe\approx0.368.
\]

For approximate solvers, the guarantee is multiplied by the approximation
ratio $\alpha$.
For example, Greedy gives $\alpha=1/2$, and in the unbounded-degree
limit with $\theta=1/\sfe$ yields competitiveness $1/(2\sfe)$.
This illustrates the effectiveness--runtime tradeoff exposed by the
framework: an exact solver maximizes $\alpha$ but may be expensive when
invoked on every prefix graph, whereas a faster approximate solver can
reduce online computation at the cost of its approximation factor.

Theorem~\ref{thm:main-2} is a transfer theorem for expected weighted
performance.
Unlike $\smg$, the online running time and variance behavior of
$\sm(\BB,\theta)$ are solver-dependent because $\BB$ is invoked on each
evolving prefix.
We therefore view $\sm(\BB,\theta)$ as a flexible meta-framework that
combines prefix reweighting with arbitrary offline approximation
procedures, transferring their approximation guarantees to the online
random-order setting while explicitly accounting for bounded offline
degree.

\section{Related Work}
\label{app:related}

Online matching and its variants have been studied extensively since the seminal work of~\citet{kvv}, spanning a wide range of arrival models, weight structures, and algorithmic paradigms. We refer the reader to the surveys of~\citet{mehta2012online} and~\citet{huang2024online} for comprehensive overviews. In this section, we focus on prior work most closely related to \emph{online matching under random arrival order}, which is the primary setting of this paper.

\xhdr{Random arrival order.} For the unweighted case under random arrival order, \citet{mahdian2011online} and \citet{karande2011online} independently proposed randomized algorithms that surpass the classical $1-1/\sfe \approx 0.632$ bound. For the vertex-weighted variant, \citet{huang2019online} introduced a generalized ranking algorithm achieving a competitive ratio of $0.6534$. The edge-weighted setting has also been studied under random order: \citet{kess-stoc} analyzed online packing LPs, which generalize edge-weighted online matching, though their focus was on settings with large offline capacities and resource constraints. \citet{ashlagi2019edge} proposed a windowed model in which online agents may be deferred for up to $d$ rounds before a decision is required.

\xhdr{Algorithmic frameworks beyond sampling--matching.}
While this paper focuses on sampling--matching algorithms, primal--dual techniques constitute another powerful approach for online matching under random order; see~\citet{buchbinder2009design} for a foundational treatment. These methods have been applied across a range of matching and allocation problems, often emphasizing dual feasibility and potential-based arguments rather than explicit sampling.

\xhdr{Strategic and incentive-aware models.}
Another related line of work studies edge-weighted online bipartite matching in strategic environments~\cite{reiffenhauser2019optimal,krysta2012online}, where online agents are modeled as bidders and offline agents as items. A central objective in these settings is to design mechanisms that ensure truthfulness, so that agents have no incentive to misreport private valuations.

\xhdr{Degree-bounded models.}
There has been growing interest in online matching models that incorporate degree constraints. \citet{buchbinder2007online} studied the AdWords problem under bounded offline degree, where each offline agent is adjacent to at most $d \ll n$ online agents under adversarial arrival order, and proposed a $1-(1-1/d)^d$-competitive primal--dual algorithm; this bound was later shown to be tight by~\citet{azar2017online}. \citet{naor2018near} introduced the $(k,d)$-bipartite model, in which offline agents have degree at least $k$ and online agents have degree at most $d$, a setting orthogonal to ours, which focuses exclusively on bounded offline degree. \citet{albers2022tight} studied unweighted and vertex-weighted matching on $(k,d)$-bipartite graphs under adversarial order and extended their analysis to the AdWords problem~\cite{albers2022online}, where they derived an optimal deterministic algorithm. \citet{cohen2018randomized} analyzed unweighted matching on $d$-regular bipartite graphs and achieved a $1-O(\ln d/d)$ competitive ratio.   \citet{aamand2022optimal} studied an unweighted model in which the algorithm is given predictive access to offline degrees, and \citet{bhangale2025optimal} investigated the extreme constant-degree regime by studying online bipartite matching in graphs with offline degree at most two, obtaining optimal guarantees in this highly structured setting.   \citet{feng2025degree} studied unweighted online bipartite matching under adversarial arrival order on $(d,d)$-bounded graphs, a generalization of $d$-regular graphs, comparing \textsc{Ranking} with Online Correlated Selection (OCS) and showing that OCS outperforms \textsc{Ranking} for every fixed $d\ge2$.

\xhdr{Alternative arrival models.}
Finally, several variants of edge-weighted online matching have been studied under alternative arrival assumptions, including adversarial order with free disposal~\cite{fahrbach2022edge} and stochastic arrivals drawn from known distributions~\cite{huang2021online,brubach2020online}.

\section{Diagnostic Analysis of the Sampling-Based Framework $\smg$}
\label{sec:smg-wm}

We begin the competitive analysis by formally presenting the framework $\smg(\theta)$ in Algorithm~\ref{alg:smg}. 
\begin{algorithm}[ht!]
\caption{Deterministic Greedy Sampling: $\smg(\theta)$ with $\theta \in [0,1]$}
\label{alg:smg}
\DontPrintSemicolon
\textbf{Sampling Phase}:\;
Collect the first $K = \lfloor n \, \theta \rfloor$ online arrivals and let $\cS_K$ denote the resulting sample set.\;
Run the Greedy algorithm (\gre) on the induced subgraph $\cG_K := (I, \cS_K)$,\footnotemark and let $\cM_K$ denote the resulting matching.\;
For each offline agent $i \in I$, set the price $p(i)$ equal to the weight of the edge incident to $i$ in $\cM_K$, or set $p(i) = 0$ if $i$ is unmatched.\;

\textbf{Matching Phase}:\;
\tcc{Each online round runs in $O(\deg(j_t))$ time.}
Initialize $\cM \leftarrow \emptyset$.\;
\For{$t = K+1$ \KwTo $n$}{
    Let $j_t \in J$ arrive at time $t$.\;
    Let $\cE_t \gets \{ e = (i, j_t) \in E_{j_t} : w(e) \ge p(i) \}$.
    \tcc*[r]{$O(\deg(j_t))$} \label{alg:s6}
  If $\cE_t \neq \emptyset$, let $e_t=(i_t,j_t)$ be a maximum-weight edge in $\cE_t$; add $e_t$ to $\cM$ if $i_t$ is unmatched, and reject $j_t$ otherwise. If $\cE_t=\emptyset$, reject $j_t$.
    \tcc*[r]{$O(\deg(j_t))$} \label{alg:s5}
}
\end{algorithm}
\footnotetext{The Greedy algorithm sorts all edges in $\cG_K := (I, \cS_K)$ in decreasing order of weight and sequentially adds each edge to $\cM_K$ if doing so preserves the matching constraint.}

This section revisits and generalizes the sampling-based analysis originally developed by \citet{korula2009algorithms}, extending it to a fully parameterized setting indexed by the maximum offline degree $d$ and the sampling parameter $\theta$.
\emph{Our objective is not to replicate the prior analysis, but to use this generalized formulation as a diagnostic lens: to identify where the classical argument becomes loose and to motivate the refinements that lead to a tight analysis in the next subsection.}

Our main result in this section is the following theorem.

\begin{theorem}[Section~\ref{app:thm-warm}]
\label{thm:smg-cr}
The policy $\smg$ achieves an effectiveness level of at least $\tfrac{1}{2}\,\eta(d,\theta)$, where $d$ is the maximum offline degree and $\eta(d,\theta)$ is defined in~\eqref{eqn:eta}.
\end{theorem}

\xhdr{Remark on Theorem~\ref{thm:smg-cr}.}
In the extremely dense setting ($d \to \infty$), the function $\eta(d,\theta)$ converges to $\theta(1-\theta)$, which is maximized at $\theta=1/2$.
This choice yields an effectiveness level of $1/8$, thereby recovering the guarantee established for the parameter-free version of \kp.

\subsection{Proof of Theorem~\ref{thm:smg-cr}}
\label{app:thm-warm}

To carry out this diagnostic analysis, we consider the virtual auxiliary algorithm $\vir$, formally stated in Algorithm~\ref{alg:vir}, which originates from prior work on \kp.
The algorithm $\vir$ is assumed to have full access to the input graph and the sample set $\cS_K$, and processes all edges offline in decreasing order of weight.
As shown below, $\vir$ produces a matching that is statistically identical to the output of $\smg$, while enabling a transparent separation between sampling effects and congestion effects.

\begin{algorithm}[ht!]
\caption{Auxiliary algorithm $\vir$ used for analyzing $\smg$}
\label{alg:vir}
\DontPrintSemicolon
Sort all edges in $E$ in decreasing order of weight.\;
Initialize $\widehat{\cM} = \cM^a = \cM^b = \emptyset$. Mark all $j \in J$ as \emph{unprocessed}.\;

\For{each edge $e = (i,j) \in E$ in sorted order \label{alg:v1}}{
    \If{$j$ is unprocessed \textbf{and} $\cM^a \cup \{e\}$ remains a matching \label{alg:v2}}{
        Mark $j$ as processed.\;
        Add $e$ to $\cM^a$ if $j \in \cS_K$; otherwise, add $e$ to $\cM^b$.\label{alg:v5}
    }
}

Let $\cR^b \subseteq I$ be the set of offline agents incident to at least one edge in $\cM^b$.\;
\For{each $i \in \cR^b$}{
    Let $\cE_i^b \subseteq \cM^b$ be the set of edges in $\cM^b$ incident to $i$.\;
    Uniformly sample one edge $e \in \cE_i^b$ and add $e$ to $\widehat{\cM}$.\label{alg:v3}  
}
\end{algorithm}

\xhdr{Remarks on $\vir$.}
The set $\cM^b$ may form a pseudo-matching, in which an offline agent can be incident to multiple edges.
For example, in a star graph with a single offline agent, $\vir$ will keep adding edges to $\cM^b$ until encountering an online agent in $\cS_K$, at which point the corresponding edge is added to $\cM^a$ and no further edges incident to that online agent are processed.

\begin{lemma}[Appendix~\ref{app:lem-eqv}]
\label{lem:eqv}
The matching $\cM$ produced by $\smg$ in Algorithm~\ref{alg:smg} is statistically identical to the matching $\widehat{\cM}$ produced by $\vir$ in Algorithm~\ref{alg:vir}.
\footnote{That is, the two random matchings share the same distribution.}
\end{lemma}

\subsubsection{Proof of Theorem~\ref{thm:smg-cr}}
For any set of edges $M \subseteq E$, let $w(M) := \sum_{e \in M} w(e)$.
For ease of notation, we use $\OPT$ to denote both an optimal matching and the corresponding total weight it achieves.
By Lemma~\ref{lem:eqv}, the matching $\cM$ produced by $\smg$ and the matching $\widehat{\cM}$ produced by $\vir$ are identically distributed.
Therefore, it suffices to lower bound $\E[w(\widehat{\cM})] / \OPT$.

\begin{lemma}
\label{lem:smg-cr-1} (1) $\E[w(\cM^a)] \ge (\ta/2) \, \OPT$;  (2) When $d = 1$, $\E[w(\cM^b)] = (1-\ta)\,\OPT$, and when $d \ge 2$,
\begin{align}
\E[w(\cM^b)] \ge \frac{1-\ta}{2} \, \OPT . \label{ineq:bm}
\end{align}
\end{lemma}

\begin{proof}
We first prove part~(1).
Let $\cM_K^*$ denote an optimal matching in the random subgraph $\cG_K := (I,\cS_K)$.
Since $\Pr[j \in \cS_K] = K/n = \ta$ for every $j \in J$, linearity of expectation implies
\[
\E[w(\cM_K^*)] \ge \ta \, \OPT .
\]
Because $\cM_K$ is returned by the greedy algorithm on $\cG_K$, we have
\[
\E[w(\cM_K)] \ge \frac{1}{2} \, \E[w(\cM_K^*)] \ge (\ta/2) \, \OPT .
\]
By construction of $\vir$, for any fixed sample set $\cS_K$, the matching $\cM^a$ coincides with $\cM_K$.
Therefore, $\E[w(\cM^a)] = \E[w(\cM_K)] \ge (\ta/2) \, \OPT$.

We now prove part~(2).
When $d = 1$, each offline agent has degree one, so every edge not selected into $\cM^a$ must belong to $\cM^b$, implying $\E[w(\cM^b)] = (1-\ta)\,\OPT$.
When $d \ge 2$, conditioning on the order in which edges are processed, each edge is assigned 
to $\cM^a$ or $\cM^b$ according to whether its online endpoint lies in $\cS_K$, 
which occurs with marginal probabilities $\ta$ and $1-\ta$, respectively. Consequently,
\[
\E[w(\cM^b)] = \E[w(\cM^a)] \, \frac{1-\ta}{\ta}
\ge \frac{1-\ta}{2} \, \OPT.
\]
\end{proof}

\begin{lemma}
\label{lem:deg}
Let $D$ denote the (random) degree of an offline agent in $\cM^b$.
Then
\begin{align}\label{ineq:2-4-a}
\E[D \mid D \ge 1] \le \sum_{\ell=1}^d (1-\ta)^{\ell-1},
\end{align}
where $d$ is the maximum offline degree.
\end{lemma}

\begin{proof}
Fix an offline agent $i \in I$ and assume, without loss of generality, that it has exactly $d$ online neighbors.
Under $\vir$, edges incident to $i$ are processed in decreasing order of weight and are added to $\cM^b$ until the first incident edge whose online endpoint lies in $\cS_K$ is encountered, at which point that edge is added to $\cM^a$.

Viewing the addition of an edge to $\cM^a$ as a \emph{success} and to $\cM^b$ as a \emph{failure}, the conditional random variable $(D \mid D \ge 1)$ equals the number of failures before the first success, truncated at $d$, given that at least one failure occurs.
Since sampling is performed without replacement, the events corresponding to different neighbors are not independent; nevertheless, at each step the conditional probability of observing another failure is at most $1-\ta$.\footnote{For example, when $\ell=2$, conditioning on $D \ge 1$ means that the highest-weight neighbor is not in $\cS_K$. The conditional probability that the second-highest-weight neighbor is also not in $\cS_K$ equals $(n-K-1)/(n-1)$, which is at most $1-K/n = 1-\ta$. More generally, after $\ell-1$ failures, the corresponding conditional probability is $(n-K-(\ell-1))/(n-(\ell-1)) \le 1-\ta$.}

Therefore, for each $\ell \in [d]$,
\[
\Pr[D \ge \ell \mid D \ge 1] \le (1-\ta)^{\ell-1}.
\]
Consequently,
\[
\E[D \mid D \ge 1]
= \sum_{\ell=1}^d \Pr[D \ge \ell \mid D \ge 1]
\le \sum_{\ell=1}^d (1-\ta)^{\ell-1}.
\]
\end{proof}

\begin{proof}[Proof of Theorem~\ref{thm:smg-cr}]
For each offline agent $i$, let $\cE_i^b$ denote the set of edges incident to $i$ in $\cM^b$, and let $D_i := |\cE_i^b|$.
By construction of $\vir$,
\[
\E[w(\widehat{\cM})]
= \sum_{i \in I} \E\!\left[ \sum_{e \in \cE_i^b} \frac{w(e)}{D_i} \,\middle|\, D_i \ge 1 \right] \Pr[D_i \ge 1].
\]

For a fixed offline agent $i$, the random variable
\[
\left( \sum_{e \in \cE_i^b} \frac{w(e)}{D_i} \,\middle|\, D_i \ge 1 \right)
\]
is non-increasing in $(D_i \mid D_i \ge 1)$, since $\vir$ processes incident edges in decreasing order of weight.
Applying the FKG inequality, we obtain
\begin{align*}
\E\!\left[ \sum_{e \in \cE_i^b} w(e) \right]
&= \E\!\left[
\left( \sum_{e \in \cE_i^b} \frac{w(e)}{D_i} \right) D_i \,\middle|\, D_i \ge 1
\right] \Pr[D_i \ge 1] \le
\E\!\left[
\sum_{e \in \cE_i^b} \frac{w(e)}{D_i}
\right]
\, \E[D_i \mid D_i \ge 1].
\end{align*}

Summing over all offline agents and invoking Lemma~\ref{lem:deg}, we obtain
\begin{align}
\E[w(\cM^b)]
&\le
\E[w(\widehat{\cM})]
\, \sum_{\ell=1}^d (1-\ta)^{\ell-1}.
\label{ineq:smg-1}
\end{align}

Finally, for $d \ge 2$, combining Lemma~\ref{lem:eqv}, inequality~\eqref{ineq:smg-1}, and Lemma~\ref{lem:smg-cr-1}, we obtain
\begin{align*}
\frac{\E[w(\cM)]}{\OPT}
&= \frac{\E[w(\widehat{\cM})]}{\OPT}
\quad \text{(by Lemma~\ref{lem:eqv})} \nonumber \\
&\ge
\frac{1}{\OPT}
\, \frac{\E[w(\cM^b)]}{\sum_{\ell=1}^d (1-\ta)^{\ell-1}}
\quad \text{(by inequality~\eqref{ineq:smg-1}, comparing $\widehat{\cM}$ to $\cM^b$)}\\
&\ge
\frac{1}{2}
\, \frac{1-\ta}{\sum_{\ell=1}^d (1-\ta)^{\ell-1}}
= \frac{1}{2} \, \eta(d,\ta)
\quad \text{(by inequality~\eqref{ineq:bm} and the definition of $\eta(d,\ta)$)}.
\end{align*}

The case $d = 1$ follows analogously. This completes the proof of Theorem~\ref{thm:smg-cr}, establishing the stated effectiveness guarantee.
\end{proof}

\section{A Tight Competitive Analysis of $\smg$}\label{sec:smg}
\subsection{Motivation for a Tight Analysis of $\smg$}
\label{sec:smg-t}

As shown in the proof of Theorem~\ref{thm:smg-cr}, the diagnostic analysis of $\smg$ relies on two inequalities:
inequality~\eqref{ineq:smg-1}, which relates $\E[w(\widehat{\cM})]$ to $\E[w(\cM^b)]$, and
inequality~\eqref{ineq:bm}, which bounds $\E[w(\cM^b)]$ in terms of the offline optimum $\OPT$.
Although each inequality can be tight for suitable instances, the mechanisms underlying their tightness are fundamentally different.

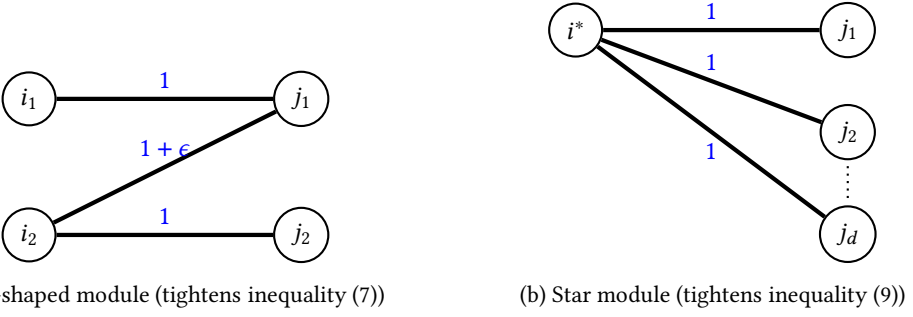
\begin{figure}[ht!]
\centering
\begin{minipage}[b]{0.48\linewidth}
\centering
\begin{tikzpicture}[scale=0.9]
    \node[circle, draw, thick, minimum size=7mm] (i1) at (0,0) {$i_1$};
    \node[circle, draw, thick, minimum size=7mm] (i2) at (0,-2) {$i_2$};

    \node[circle, draw, thick, minimum size=7mm] (j1) at (4,0) {$j_1$};
    \node[circle, draw, thick, minimum size=7mm] (j2) at (4,-2) {$j_2$};

    \draw[ultra thick] (i1) -- node[above, midway, blue] {$1$} (j1);
    \draw[ultra thick] (i2) -- node[above, midway, blue] {$1$} (j2);
    \draw[ultra thick] (i2) -- node[above, midway, blue] {$1+\ep$} (j1);
\end{tikzpicture}

\vspace{1mm}
{\small (a) $Z$-shaped module (tightens inequality~\eqref{ineq:bm})}
\end{minipage}
\hfill
\begin{minipage}[b]{0.48\linewidth}
\centering
\begin{tikzpicture}[scale=0.9]
    \node[circle, draw, thick, minimum size=7mm] (i) at (0,0) {$i^*$};

    \node[circle, draw, thick, minimum size=7mm] (j1) at (4,0) {$j_1$};
    \node[circle, draw, thick, minimum size=7mm] (j2) at (4,-1.5) {$j_2$};
    \node[circle, draw, thick, minimum size=7mm] (jd) at (4,-3) {$j_d$};

    \draw[dotted, thick] (4,-2) -- (4,-2.5);

    \draw[ultra thick] (i) -- node[above, midway, blue] {$1$} (j1);
    \draw[ultra thick] (i) -- node[above, midway, blue] {$1$} (j2);
    \draw[ultra thick] (i) -- node[below, midway, blue] {$1$} (jd);
\end{tikzpicture}

\vspace{1mm}
{\small (b) Star module (tightens inequality~\eqref{ineq:smg-1})}
\end{minipage}

\caption{
Two structurally distinct tightness mechanisms underlying the diagnostic analysis.
Panel (a) shows a $Z$-shaped weighted module that can make inequality~\eqref{ineq:bm} tight by combining an essentially unique optimal matching with a worst-case configuration for greedy.
Panel (b) shows a star-like module that can make inequality~\eqref{ineq:smg-1} tight by inducing a truncated geometric degree process in $\cM^b$.
Formal constructions and calculations are provided in Appendix~\ref{app:tight}.
}
\label{fig:tightness-mechanisms}
\end{figure}

\xhdr{Structural requirements for tightening inequality~\eqref{ineq:bm}.}
We first consider inequality~\eqref{ineq:bm},
\[
\E[w(\cM^b)] = \frac{1-\ta}{2} \, \OPT.
\]
As illustrated in Figure~\ref{fig:tightness-mechanisms}(a), the auxiliary set $\cM^b$ can be viewed as the outcome of applying the greedy algorithm to a random edge set $\cE$, where each edge $e \in E$ is included with marginal probability $1-\ta$.
For this inequality to be tight, two conditions must hold simultaneously.
First, the expected optimal value over $\cE$ must satisfy $\E[\OPT(\cE)] = (1-\ta)\,\OPT$, which requires the input graph to have an essentially unique optimal matching.
Second, greedy must attain its worst-case approximation ratio of $1/2$ on $\cE$ with high probability.
This combination is known to require highly structured instances, such as collections of $Z$-shaped modules, in which greedy consistently performs suboptimally.
As a result, tightness of inequality~\eqref{ineq:bm} arises only in regimes with little effective randomness and highly constrained optimal structure; see Appendix~\ref{app:z-tight} for concrete constructions.

\xhdr{Structural requirements for tightening inequality~\eqref{ineq:smg-1}.}
In contrast, the tightness of inequality~\eqref{ineq:smg-1},
\[
\E[w(\cM^b)] = \E[w(\widehat{\cM})] \, \sum_{\ell=1}^d (1-\ta)^{\ell-1},
\]
depends on the tightness of the degree bound in Lemma~\ref{lem:deg},
\[
\E[D \mid D \ge 1] = \sum_{\ell=1}^d (1-\ta)^{\ell-1}.
\]
As shown schematically in Figure~\ref{fig:tightness-mechanisms}(b), this equality requires a markedly different structural regime:
each offline agent must have exactly $d$ incident edges revealed in strictly decreasing order of weight, so that the degree process in $\cM^b$ follows a truncated geometric distribution.
Such behavior is realized, for example, in collections of disjoint star graphs; see Appendix~\ref{app:star-tight}.

Taken together, these two regimes are fundamentally incompatible.
Tightness of inequality~\eqref{ineq:bm} requires instances with little randomness and highly constrained optimal structure, whereas tightness of inequality~\eqref{ineq:smg-1} requires substantial randomness and star-like local neighborhoods.
Consequently, the two sources of looseness identified in the diagnostic analysis cannot, in general, be simultaneously realized by the same input instance.
This incompatibility motivates a refined analysis that directly compares $\E[w(\widehat{\cM})]$ with $\OPT$, without passing through the intermediate quantity $\cM^b$, which is the focus of the tight analysis developed next.



\subsection{Proof of Theorem~\ref{thm:smg-cr-tig}}
\label{sec:smg-cr-tig}

We sketch the proof of Theorem~\ref{thm:smg-cr-tig}, emphasizing the key ideas and worst-case structures.
All technical details are deferred to Appendix~\ref{app:smg-cr-tight}.

The analysis focuses on a local worst-case scenario around a fixed offline agent $i^*$.
Let $\beta$ denote the expected gain contributed by the edge incident to $i^*$ that is selected into $\bam$ by the virtual algorithm $\vir$, and let $\gamma$ denote the weight of the edge incident to $i^*$ in the offline optimal matching $M^*$.
Establishing a tight lower bound on $\beta/\gamma$ directly yields a lower bound on $\E[w(\bam)]/\OPT$.

\begin{lemma}[Appendix~\ref{app:smg-cr-tight}]
\label{lem:bam-b}
\begin{align}
\frac{\beta}{\gamma} \ge \kap(d,\ta) - \frac{\ta(1-\ta)}{2},
\label{ineq:bam-b}
\end{align}
where $\kap(d,\ta)$ is defined in~\eqref{eqn:kap}.
\end{lemma}

The proof of Lemma~\ref{lem:bam-b} performs a case analysis based on whether the highest-weight edge incident to $i^*$ belongs to the offline optimal matching $M^*$.
For $\ta \in (0,1)$, the worst-case ratio is uniquely attained in a specific configuration (denoted as \textbf{Case~2a} in Appendix~\ref{app:smg-cr-tight}), where competition between two offline agents leads to the maximal loss relative to the optimum.
All other cases, including boundary cases with $\ta \in \{0,1\}$, yield strictly better bounds.

Combining Lemma~\ref{lem:bam-b} with the structural incompatibility identified in Section~\ref{sec:smg-t} establishes the tight competitive guarantee claimed in Theorem~\ref{thm:smg-cr-tig}.


\section{Variance Analysis of the Deterministic Greedy Sampling Framework}
\label{sec:var}

\subsection{Technical Challenges in Variance Analysis}

Recall that $\cM$ and $\widehat{\cM}$ denote the final matchings output by $\smg$ (Algorithm~\ref{alg:smg}) and the virtual algorithm $\vir$ (Algorithm~\ref{alg:vir}), respectively.
By Lemma~\ref{lem:eqv}, these two matchings are statistically identical.
Consequently, it suffices to analyze the variance of the matching produced by $\vir$.

For each offline agent $i \in I$, let $X_i$ be the indicator variable that $i$ is matched in $\widehat{\cM}$, and define
\[
X := \sum_{i \in I} X_i .
\]
Our goal is to upper bound the variance $\Var[X]$.

\smallskip
\noindent
\emph{Throughout this section, we write $\bta := 1-\ta$ and $\tta := \max\{1/2,\ta\}$.}

\begin{lemma}
\label{lem:var-gen}
For each $i \in I$, we have $\E[X_i] \le \bta = 1-\ta$ for all $\ta \in [0,1]$.
\end{lemma}

\begin{proof}
An offline agent $i$ is matched by $\vir$ if and only if the first edge incident to $i$ that is processed in Step~\eqref{alg:v5} of $\vir$, say $e=(i,j)$, satisfies $j \notin \cS_K$.
Since $\Pr[j \notin \cS_K] = 1-\ta$, it follows that $\E[X_i] \le 1-\ta = \bta$.
\end{proof}

Lemma~\ref{lem:var-gen} immediately implies that when $\ta \in [1/2,1]$, we have
\[
\Var[X_i] \le \E[X_i](1-\E[X_i]) \le \ta(1-\ta) \le \ta .
\]
Thus, in this regime the variance of each individual indicator is well controlled.

The main difficulty arises in the low-sampling regime $\theta\in[0,1/2]$. 
In this regime, $\bar{\theta}\ge 1/2$, so an individual indicator $X_i$ can have variance as large as $1/4$. 
More importantly, the matching indicators $\{X_i\}_{i\in I}$ need not satisfy negative dependence. 
Appendix~\ref{app:neg-dep} gives two explicit examples: one in which a single matching indicator attains variance $1/4$ and a pair of indicators has joint variance exceeding the independent benchmark, and another in which two matching indicators have strictly positive covariance. 
Thus the variance analysis cannot rely on standard concentration arguments based on independence or negative association.

The key step is instead to control the aggregate covariance directly. 
The following lemma shows that, despite possible positive correlations among individual matching events, the total covariance is bounded by a quantity linear in the number of offline agents.

%

\begin{lemma}[Appendix~\ref{app:lem-var5}]
\label{lem:var-5}
Let $\bta := 1-\ta$.
For each $i \in I$, let $X_i=1$ indicate that $i$ is matched in $\vir$.
Then
\[
\sum_{1 \le i < i' \le m}
\bigl(\E[X_i X_{i'}] - \E[X_i]\E[X_{i'}]\bigr)
\;\le\;
m \, \min\Bigl\{
\frac{\bta^4}{\ta^2},
\; (d+2)\,\bta^{7/2}
\Bigr\}.
\]
\end{lemma}

\subsection{Proof of the Upper Bound on $\psi$ in Theorem~\ref{thm:smg-var-new}}
\label{sec:smg-var}

\begin{proof}
Let
$X= \sum_{i \in I} X_i$, where $X_i=1$ if offline agent $i$ is matched in $\vir$, and $X_i=0$ otherwise.
Define $\mu_i := \E[X_i] \le 1-\ta = \bta$ for each $i \in I$, and recall $\tta := \max\{1/2,\ta\}$.

Using the variance decomposition formula, we have
\begingroup
\allowdisplaybreaks
\begin{align*}
\Var[X]
&= \E[X^2] - \E[X]^2 = \sum_{i \in I} \bigl(\mu_i - \mu_i^2\bigr)
   + 2 \sum_{1 \le i < i' \le m}
     \bigl(\E[X_i X_{i'}] - \E[X_i]\E[X_{i'}]\bigr).
\end{align*}
\endgroup

By Lemma~\ref{lem:var-5}, the total covariance term is at most
$2m \cdot \min\{\bta^4/\ta^2,\,(d+2)\,\bta^{7/2}\}$.
Moreover, since $\mu_i \le \bta$ for all $i$ and $\tta=\max\{1/2,\ta\}$, $
\mu_i - \mu_i^2 \le \tta(1-\tta)$.

Combining these bounds yields
\[
\Var[X]
\le
m \cdot \Bigl(
\tta(1-\tta)
+ 2 \min\Bigl\{
\frac{\bta^4}{\ta^2},
\; (d+2)\,\bta^{7/2}
\Bigr\}
\Bigr),
\]
which completes the proof.
\end{proof}

\section{Competitive Analysis of the Black-Box Sampling--Matching Framework}
\label{sec:sda}

The framework $\sm(\BB,\ta)$ invokes an offline matching procedure on a
prefix-dependent reweighting of the observed graph. The reweighting
incorporates a conditional lower bound on the availability of each offline
endpoint. Accepted edges yield their original weights, rather than their
reweighted values.

\xhdr{Black-box convention: independence from the arrival order.}
Let $\BB$ be an $\alp$-approximation procedure for maximum-weight bipartite
matching with nonnegative weights. The approximation guarantee may hold
in expectation over the internal randomness of $\BB$. Each call to $\BB$
uses only the unordered, labeled prefix graph and its reweighted edge
weights; in particular, its output is independent of the arrival order of
agents within the prefix, and it does not use the identity of the current
arrival as a distinguished vertex or the current online matching.
If $\BB$ is randomized, let $\xi_t$ denote the random seed used at round
$t$, where the seeds $(\xi_t)_{t>K}$ are mutually independent and
independent of the arrival order. For a deterministic procedure, the seeds
may be taken to be constant. Appendix~\ref{sec:amend} gives a concrete
example showing that this arrival-order-independence requirement is
essential: even an exact prefix solver can yield vanishing online
competitiveness if its choice among optimal matchings is allowed to depend
on the arrival order.

\xhdr{Prefix reweighting.}
Fix $\ta\in[0,1]$ and set $K:=\lfloor n\ta\rfloor$.
For each matching-phase round $t>K$, let
\[
d_t(i):=
\bigl|\{j\in\cS_t:(i,j)\in E\}\bigr|
\]
denote the degree of offline agent $i$ in the observed prefix.
Define
\begin{equation}
b_t:=
\begin{cases}
1, & t=1,\\[1mm]
\dfrac{K}{t-1}, & t\ge2,
\end{cases}
\qquad
a_t(i):=
\begin{cases}
1, & d_t(i)\le1,\\
b_t, & d_t(i)\ge2.
\end{cases}
\label{eq:sda-factors}
\end{equation}
The case $t=1$ occurs only when $K=0$, and the value assigned to $a_t(i)$
when $d_t(i)=0$ is immaterial. For every edge $(i,j)$ in
$\cG_t:=(I,\cS_t)$, define the reweighted value
\begin{equation}
\widetilde w_t(i,j):=a_t(i)w(i,j).
\label{eq:sda-reweighting}
\end{equation}
For notational convenience, we further set
\[
a_t(\bot):=1.
\]
All reweighting factors lie in $[0,1]$ and depend only on the unordered
prefix, rather than on which offline agents are currently matched. In
particular, the reweighting requires no knowledge of unrevealed neighbors.

A formal description of the framework is given in
Algorithm~\ref{alg:sda}.

\begin{algorithm}[ht!]
\caption{A Black-Box Sampling--Matching Framework for \om:
$\sm(\BB,\ta)$}
\label{alg:sda}
\DontPrintSemicolon

\textbf{Sampling Phase}:\;
Collect the first $K:=\lfloor n\ta\rfloor$ online arrivals without
matching them, and initialize $\cM\gets\emptyset$.\;

\textbf{Matching Phase}:\;
\For{$t=K+1,\ldots,n$}{
    Let $j_t\in J$ arrive at time $t$, and let $\cS_t$ denote the set of
    online agents that have arrived by time $t$.\;

    Define the induced subgraph $\cG_t:=(I,\cS_t)$.\;

    Compute $d_t(i)$ and $a_t(i)$ according to~\eqref{eq:sda-factors},
    and assign weight
    $\widetilde w_t(i,j):=a_t(i)w(i,j)$ to every edge of $\cG_t$.\;

    Run $\BB$ on $\cG_t$ with weights $\widetilde w_t$, and let
    $\cM_t$ denote the returned matching.\;
    \label{step:call_bb}

    If $j_t$ is matched in $\cM_t$, let $e_t=(i_t,j_t)$ denote its
    incident edge; otherwise, set $e_t=i_t=\bot$.\;
    \label{alg:g3}

    \eIf{$e_t\neq\bot$ and $i_t$ is unmatched in $\cM$}{
        Add $e_t$ to $\cM$, receiving reward $w(e_t)$.\;
        \label{alg:g4}
    }{
        Reject $j_t$.\;
    }
}
\end{algorithm}

\subsection{Proof of Theorem~\ref{thm:main-2}}

Let $\OPT$ denote the total weight of a maximum-weight matching in the
full input graph $G$. For any matching $M$ in $\cG_t$, write
\[
w(M):=\sum_{e\in M}w(e),
\qquad
\widetilde w_t(M):=\sum_{e\in M}\widetilde w_t(e).
\]

For each $t>K$, let
\[
W_t:=w(e_t),
\qquad
\widetilde W_t:=\widetilde w_t(e_t)
\]
when $e_t\neq\bot$, and set
\[
W_t=\widetilde W_t:=0
\]
otherwise.

Let
\[
A_t:=
\{i_t\text{ is unmatched immediately before round }t\},
\qquad
R_t:=W_t\mathbf{1}_{A_t}.
\]
For notational convenience, when $i_t=\bot$, we regard $\bot$ as
unmatched at every round and recall that $a_t(\bot)=1$.
This convention has no effect on the algorithm or its reward since
$W_t=\widetilde W_t=0$ whenever $i_t=\bot$.
Thus,
\[
w(\cM)=\sum_{t=K+1}^nR_t.
\]

All expectations below are taken over the random arrival order and, when
applicable, the internal randomness of $\BB$.

Define
\begin{equation}
q_{n,d}(t):=
\frac{\binom{n-t}{d-1}}
     {\binom{n-1}{d-1}},
\qquad
p_t:=
q_{n,d}(t)+
\bigl(1-q_{n,d}(t)\bigr)b_t.
\label{eq:sda-pt}
\end{equation}
Here
\[
\binom{k_1}{k_2}=0
\qquad
\text{for integers }0\le k_1<k_2.
\]
In particular, $q_{n,1}(t)=1$.

We first establish two auxiliary lemmas.

\begin{lemma}[Section~\ref{sec:14-a}]
\label{lem:warm1}
For every $t>K$, the reweighted tentative reward satisfies
\[
\E[\widetilde W_t]
\ge
\frac{\alp\,\OPT}{n}\,p_t.
\]
\end{lemma}

The next lemma lower bounds the conditional probability that the offline
endpoint of the tentative edge remains available.

Write
\[
\mathcal H_t:=
\sigma(\cS_t,j_t,\xi_t),
\]
for the information that determines the tentative edge selected at
round $t$: the unordered prefix $\cS_t$, the identity of the current
arrival $j_t$, and the current black-box random seed $\xi_t$.
Conditioned on $\mathcal H_t$, the tentative edge $e_t$ and its weight
are fixed, while the earlier arrival order, and hence the availability
of $i_t$, remain random.

\begin{lemma}[Section~\ref{sec:14-b}]
\label{lem:ge1}
For every $t>K$,
\[
\Pr[A_t\mid\mathcal H_t]
\ge
a_t(i_t)
\qquad\text{almost surely}.
\]
Consequently,
\[
\E[R_t\mid\mathcal H_t]
\ge
\widetilde W_t
\qquad\text{almost surely}.
\]
\end{lemma}

\begin{proof}[\textbf{Proof of Theorem~\ref{thm:main-2}}]
By conditional expectation and
Lemmas~\ref{lem:ge1} and~\ref{lem:warm1},
\begin{align*}
\E[w(\cM)]
&=
\sum_{t=K+1}^n
\E\bigl[\E[R_t\mid\mathcal H_t]\bigr] 
\ge
\sum_{t=K+1}^n\E[\widetilde W_t] 
\ge
\frac{\alp\,\OPT}{n}
\sum_{t=K+1}^np_t.
\end{align*}
Thus, for every input instance,
\begin{equation}
\E[w(\cM)]
\ge
\frac{\alp\,\OPT}{n}
\sum_{t=K+1}^n
\Bigl[
q_{n,d}(t)+
\bigl(1-q_{n,d}(t)\bigr)b_t
\Bigr].
\label{eq:sda-finite-bound}
\end{equation}
Notice that this argument does not require the tentative weight and the
availability event to be independent.

We first consider $\ta\in(0,1)$ and keep $d$ fixed as $n\to\infty$.
Uniformly over $t>K$,
\[
q_{n,d}(t)
=
\left(1-\frac{t}{n}\right)^{d-1}
+
O_d(1/n),
\]
and
\[
b_t
=
\frac{\ta}{t/n}
+
O_{\ta}(1/n).
\]
The first estimate follows from the fixed-length product representation
of $q_{n,d}(t)$, while the second follows from
$K=\lfloor n\ta\rfloor$ and $t-1\ge K$.
Therefore, the normalized sum in~\eqref{eq:sda-finite-bound} converges
to
\begin{align*}
\lim_{n\to\infty}
\frac{1}{n}\sum_{t=K+1}^np_t
&=
\int_{\ta}^1
\Bigl[
(1-\zeta)^{d-1}
+
\bigl(1-(1-\zeta)^{d-1}\bigr)
\frac{\ta}{\zeta}
\Bigr]
\,\sd\zeta =
\kap(d,\ta),
\end{align*}
where $\kap(d,\ta)$ is defined in~\eqref{eqn:kap}.
Hence,
\[
\E[w(\cM)]
\ge
\bigl(\alp\,\kap(d,\ta)-o(1)\bigr)\OPT,
\]
where the $o(1)$ term is independent of the particular input instance.

When $\ta=0$, we have $K=0$, $b_t=0$ for $t\ge2$, and
$q_{n,d}(1)=1$. Hence $p_t=q_{n,d}(t)$ for every $t$, including
$t=1$. Using the binomial summation identity,
\begin{align*}
\frac{1}{n}\sum_{t=1}^np_t
&=
\frac{
\sum_{t=1}^n\binom{n-t}{d-1}
}{
n\binom{n-1}{d-1}
} =
\frac{\binom{n}{d}}
     {n\binom{n-1}{d-1}}
=
\frac{1}{d}
=
\kap(d,0).
\end{align*}
Thus, the guarantee is at least $\alp/d$ even for finite $n$.

When $\ta=1$, the matching phase is empty and the guarantee is
\[
0=\alp\,\kap(d,1).
\]
Together, these cases establish the theorem for all
$\ta\in[0,1]$.
\end{proof}

\paragraph{Unrestricted offline degree.}
For $K\ge1$, setting $d=n$ in~\eqref{eq:sda-finite-bound} gives
\[
q_{n,n}(t)=0
\qquad
\text{for every }t>K.
\]
Therefore,
\[
\E[w(\cM)]
\ge
\alp\,\OPT\,
\frac{K}{n}
\sum_{t=K+1}^n\frac{1}{t-1}.
\]
For every fixed $\ta\in(0,1)$, this yields asymptotic competitiveness
\[
-\alp\,\ta\ln\ta
\]
without any degree restriction. In particular, choosing
$\ta=1/\mathrm{e}$ gives the guarantee $\alp/\mathrm{e}$.

\subsection{Proof of Lemma~\ref{lem:warm1}}
\label{sec:14-a}

\begin{proof}
Fix $t>K$, and let $\widetilde{\OPT}_t$ denote the maximum matching
weight in $\cG_t$ under the reweighted objective $\widetilde w_t$.

Conditioned on $\cS_t$ and $\xi_t$, the matching $\cM_t$ returned by
$\BB$ is fixed, while $j_t$ is uniformly distributed over $\cS_t$.
Consequently,
\[
\E[\widetilde W_t\mid\cS_t,\xi_t]
=
\frac{\widetilde w_t(\cM_t)}{t}.
\]
The approximation guarantee of $\BB$ implies
\[
\E[\widetilde w_t(\cM_t)\mid\cS_t]
\ge
\alp\,\widetilde{\OPT}_t.
\]
Taking expectations gives
\begin{equation}
\E[\widetilde W_t]
\ge
\frac{\alp}{t}
\E[\widetilde{\OPT}_t].
\label{eq:sda-prefix-approx}
\end{equation}

Let $\cM^*$ be a fixed maximum-weight matching in the full graph under
the original weights. Restricting $\cM^*$ to those edges whose online
endpoints belong to $\cS_t$ gives a feasible matching in $\cG_t$.
Therefore,
\begin{equation}
\widetilde{\OPT}_t
\ge
\sum_{(i,j)\in\cM^*}
w(i,j)\,
\mathbf{1}_{\{j\in\cS_t\}}\,
a_t(i).
\label{eq:sda-fixed-benchmark}
\end{equation}

Fix an edge $(i,j)\in\cM^*$, and write
\[
D_i:=\deg_G(i)\le d.
\]
Conditioned on $j\in\cS_t$, the set
$\cS_t\setminus\{j\}$ is a uniformly random $(t-1)$-element subset of
$J\setminus\{j\}$. Hence,
\begin{align*}
\Pr[d_t(i)=1\mid j\in\cS_t]
&=
\frac{\binom{n-t}{D_i-1}}
     {\binom{n-1}{D_i-1}} 
\ge
\frac{\binom{n-t}{d-1}}
     {\binom{n-1}{d-1}} =
q_{n,d}(t).
\end{align*}
The inequality follows because enlarging the fixed set of other
neighbors from $D_i-1$ to $d-1$ can only decrease the probability of
avoiding all of them. Importantly, this calculation concerns a fixed
edge of $\cM^*$, rather than an endpoint selected adaptively by $\BB$.

Since
\[
\Pr[j\in\cS_t]=\frac{t}{n}
\]
and $d_t(i)\ge1$ whenever $j\in\cS_t$, we obtain
\begin{align*}
\E\Bigl[
\mathbf{1}_{\{j\in\cS_t\}}
a_t(i)
\Bigr]
&=
\frac{t}{n}
\Bigl[
b_t+
(1-b_t)
\Pr[d_t(i)=1\mid j\in\cS_t]
\Bigr] \\
&\ge
\frac{t}{n}
\Bigl[
b_t+
(1-b_t)q_{n,d}(t)
\Bigr] =
\frac{t}{n}\,p_t.
\end{align*}

Taking expectations in~\eqref{eq:sda-fixed-benchmark} and summing over
the edges of $\cM^*$ yields
\[
\E[\widetilde{\OPT}_t]
\ge
\frac{t}{n}\,p_t
\sum_{(i,j)\in\cM^*}w(i,j)
=
\frac{t}{n}\,p_t\,\OPT.
\]
Substituting this into~\eqref{eq:sda-prefix-approx} gives
\[
\E[\widetilde W_t]
\ge
\frac{\alp\,\OPT}{n}\,p_t,
\]
which completes the proof.
\end{proof}

\subsection{Proof of Lemma~\ref{lem:ge1}}
\label{sec:14-b}

\begin{proof}
Fix $t>K$ and condition on $\mathcal H_t$.

If $i_t=\bot$, then by convention
\[
\Pr[A_t\mid\mathcal H_t]
=
1
=
a_t(\bot),
\]
and the claimed probability bound is immediate.
Henceforth assume that $i_t\in I$, and write
\[
i:=i_t.
\]

\smallskip
\noindent
\textbf{Case 1: $d_t(i)=1$.}
The only neighbor of $i$ in $\cS_t$ is $j_t$.
Hence no earlier arrival can have matched $i$, and therefore
\[
\Pr[A_t\mid\mathcal H_t]
=
1
=
a_t(i).
\]
This case also covers $t=1$ whenever a tentative edge exists.

\smallskip
\noindent
\textbf{Case 2: $d_t(i)\ge2$.}
In this case,
\[
a_t(i)=b_t.
\]
If $K=0$, then $t\ge2$ and $b_t=0$, so the required lower bound is
immediate. Henceforth assume $K\ge1$.

Expose the order of
\[
\cS_t\setminus\{j_t\}
\]
backwards, beginning with round $t-1$ and proceeding down to round
$K+1$. For each
\[
s\in\{K+1,\ldots,t-1\},
\]
define
\[
\mathcal F_s:=
\sigma\bigl(
\cS_t;\,
j_t,j_{t-1},\ldots,j_{s+1};\,
\xi_t,\ldots,\xi_s
\bigr).
\]
Conditioned on $\mathcal F_s$, the set
\[
\cS_s
=
\cS_t
\setminus
\{j_t,j_{t-1},\ldots,j_{s+1}\}
\]
is fixed, while $j_s$ is uniformly distributed over the $s$ elements
of $\cS_s$.
Moreover, since $\xi_s$ is included in $\mathcal F_s$, the matching
$\cM_s$ returned by $\BB$ on the reweighted prefix graph $\cG_s$ is
fixed under this conditioning.

Since $\cM_s$ is a matching, it contains at most one edge incident to
$i$. Define
\[
B_s:=
\{(i,j_s)\notin\cM_s\},
\]
the event that round $s$ makes no tentative proposal to $i$.
It follows that
\begin{equation}
\Pr[B_s\mid\mathcal F_s]
\ge
1-\frac{1}{s}.
\label{eq:sda-backward-step}
\end{equation}

To combine these conditional bounds, define
\[
C_s:=
\bigcap_{r=s}^{t-1}B_r,
\qquad
C_t:=\Omega.
\]
The event $C_{s+1}$ is $\mathcal F_s$-measurable, since all tentative
proposals in the already exposed suffix are determined by
$\mathcal F_s$. Moreover,
\[
\mathcal H_t\subseteq\mathcal F_s.
\]
Therefore,
\begin{align*}
\Pr[C_s\mid\mathcal H_t]
&=
\E\Bigl[
\mathbf{1}_{C_{s+1}}
\Pr[B_s\mid\mathcal F_s]
\;\Big|\;
\mathcal H_t
\Bigr] \ge
\left(1-\frac{1}{s}\right)
\Pr[C_{s+1}\mid\mathcal H_t].
\end{align*}
Iterating this inequality from $s=t-1$ down to $s=K+1$ yields
\begin{align*}
\Pr[C_{K+1}\mid\mathcal H_t]
&\ge
\prod_{s=K+1}^{t-1}
\left(1-\frac{1}{s}\right) =
\frac{K}{t-1}
=
b_t,
\end{align*}
where an empty product is interpreted as $1$.

On the event $C_{K+1}$, no earlier matching-phase round makes a
tentative proposal to $i$. Hence $i$ is unmatched immediately before
round $t$. Consequently,
\[
\Pr[A_t\mid\mathcal H_t]
\ge
b_t
=
a_t(i).
\]

Combining the two cases, together with the dummy case $i_t=\bot$,
establishes
\[
\Pr[A_t\mid\mathcal H_t]
\ge
a_t(i_t).
\]

Finally, $W_t$ is $\mathcal H_t$-measurable. Therefore,
\begin{align*}
\E[R_t\mid\mathcal H_t]
&=
W_t\Pr[A_t\mid\mathcal H_t] \ge
W_t a_t(i_t).
\end{align*}
If $i_t\in I$, then by definition
\[
W_ta_t(i_t)
=
\widetilde W_t.
\]
If $i_t=\bot$, then both sides are zero by convention.
Thus,
\[
\E[R_t\mid\mathcal H_t]
\ge
\widetilde W_t,
\]
which completes the proof.
\end{proof}


\section{Conclusion and Future Work}

We developed a degree-parameterized analysis of sampling-based algorithms for edge-weighted online matching under random arrival order. For \emph{Deterministic Greedy Sampling}, we gave a tight competitiveness characterization as a function of the offline degree and sampling fraction, while preserving linear per-arrival time. We also derived worst-case bounds on the variance of the number of matched
offline agents, including an order-tight characterization near full sampling
and a linear-in-$\theta$ upper bound showing that the worst-case variance
vanishes as the sampling fraction approaches zero. For \emph{Black-Box Sampling--Matching}, we proved a transfer theorem that converts an offline approximation guarantee into online competitiveness through an explicit degree- and sampling-dependent factor.

These results show how bounded offline degree changes the behavior of classical sampling--matching frameworks and how the sampling fraction can be tuned to interpolate between sparse and dense compatibility structures. Several directions remain open. One is to identify structural assumptions under which weighted-reward variance can be controlled without being dominated by edge-weight composition. Another is to derive solver-specific variance guarantees for black-box frameworks, possibly under stability or monotonicity assumptions on the offline solver. It would also be useful to extend the degree-parameterized viewpoint to richer sparsity notions, degree regimes growing with $n$, and alternative arrival models.

\begin{acks}
This work was supported in part by the National Science Foundation under
CAREER Award No.~2538522.
A preliminary version of this work has been accepted to the 22nd Conference
on Web and Internet Economics (WINE 2026).
\end{acks}

\bibliographystyle{ACM-Reference-Format}
\bibliography{EC_21}

@inproceedings{feng2025degree,
  title={Degree-Bounded Online Bipartite Matching: OCS vs. Ranking},
  author={Feng, Yilong and Li, Haolong and Wu, Xiaowei and Zhou, Shengwei},
  booktitle={International Conference on Web and Internet Economics},
  pages={747--764},
  year={2025},
  organization={Springer}
}

@article{bhangale2025optimal,
  title={Optimal Online Bipartite Matching in Degree-2 Graphs},
  author={Bhangale, Amey and Chakraborty, Arghya and Harsha, Prahladh},
  journal={arXiv preprint arXiv:2511.16025},
  year={2025}
}

@inproceedings{azar2017online,
  title={Online lower bounds via duality},
  author={Azar, Yossi and Cohen, Ilan Reuven and Roytman, Alan},
  booktitle={Proceedings of the Twenty-Eighth Annual ACM-SIAM Symposium on Discrete Algorithms},
  pages={1038--1050},
  year={2017},
  organization={SIAM}
}

@inproceedings{buchbinder2007online,
  title={Online primal-dual algorithms for maximizing ad-auctions revenue},
  author={Buchbinder, Niv and Jain, Kamal and Naor, Joseph},
  booktitle={European Symposium on Algorithms},
  pages={253--264},
  year={2007},
  organization={Springer}
}

@article{huang2024online,
  title={Online matching: A brief survey},
  author={Huang, Zhiyi and Tang, Zhihao Gavin and Wajc, David},
  journal={ACM SIGecom Exchanges},
  volume={22},
  number={1},
  pages={135--158},
  year={2024},
  publisher={ACM New York, NY, USA}
}

@article{kuhn1955hungarian,
  title={The Hungarian method for the assignment problem},
  author={Kuhn, Harold W},
  journal={Naval research logistics quarterly},
  volume={2},
  number={1-2},
  pages={83--97},
  year={1955},
  publisher={Wiley Online Library}
}

@article{goldberg1990finding,
  title={Finding minimum-cost circulations by successive approximation},
  author={Goldberg, Andrew V and Tarjan, Robert E},
  journal={Mathematics of Operations Research},
  volume={15},
  number={3},
  pages={430--466},
  year={1990},
  publisher={INFORMS}
}

@article{abdel1982secretary,
  title={The secretary problem with an unknown number of candidates},
  author={Abdel-Hamid, AR and Bather, JA and Trustrum, GB},
  journal={Journal of Applied Probability},
  volume={19},
  number={3},
  pages={619--630},
  year={1982},
  publisher={Cambridge University Press}
}

@article{aamand2022optimal,
  title={(Optimal) Online Bipartite Matching with Degree Information},
  author={Aamand, Anders and Chen, Justin and Indyk, Piotr},
  journal={Advances in Neural Information Processing Systems},
  volume={35},
  pages={5724--5737},
  year={2022}
}

@inproceedings{krysta2012online,
  title={Online mechanism design (randomized rounding on the fly)},
  author={Krysta, Piotr and V{\"o}cking, Berthold},
  booktitle={International Colloquium on Automata, Languages, and Programming},
  pages={636--647},
  year={2012},
  organization={Springer}
}

@inproceedings{reiffenhauser2019optimal,
  title={An optimal truthful mechanism for the online weighted bipartite matching problem},
  author={Reiffenhauser, Rebecca},
  booktitle={Proceedings of the Thirtieth Annual ACM-SIAM Symposium on Discrete Algorithms},
  pages={1982--1993},
  year={2019},
  organization={SIAM}
}

@inproceedings{xu2024tight,
  title={Tight Competitive and Variance Analyses of Matching Policies in Gig Platforms},
  author={Xu, Pan},
  booktitle={Proceedings of the ACM Web Conference 2024},
  pages={5--13},
  year={2024}
}

@inproceedings{dinitz2024controlling,
  title={Controlling Tail Risk in Online Ski-Rental},
  author={Dinitz, Michael and Im, Sungjin and Lavastida, Thomas and Moseley, Benjamin and Vassilvitskii, Sergei},
  booktitle={Proceedings of the 2024 Annual ACM-SIAM Symposium on Discrete Algorithms (SODA)},
  pages={4247--4263},
  year={2024},
  organization={SIAM}
}

@article{naor2018near,
  title={Near-optimum online ad allocation for targeted advertising},
  author={Naor, Joseph and Wajc, David},
  journal={ACM Transactions on Economics and Computation (TEAC)},
  volume={6},
  number={3-4},
  pages={1--20},
  year={2018},
  publisher={ACM New York, NY, USA}
}

@article{raghavan1988probabilistic,
  title={Probabilistic construction of deterministic algorithms: approximating packing integer programs},
  author={Raghavan, Prabhakar},
  journal={Journal of Computer and System Sciences},
  volume={37},
  number={2},
  pages={130--143},
  year={1988},
  publisher={Elsevier}
}

@inproceedings{xu2022exploring,
  title={Exploring the Tradeoff between Competitive Ratio and Variance in Online-Matching Markets},
  author={Xu, Pan},
  booktitle={The 18th Conference on Web and Internet Economics},
  year={2023}
}

@inproceedings{albers2022online,
  title={Online Ad Allocation in Bounded-Degree Graphs},
  author={Albers, Susanne and Schubert, Sebastian},
  booktitle={International Conference on Web and Internet Economics},
  pages={60--77},
  year={2022},
  organization={Springer}
}

@article{albers2022tight,
  title={Tight bounds for online matching in bounded-degree graphs with vertex capacities},
  author={Albers, Susanne and Schubert, Sebastian},
  journal={arXiv preprint arXiv:2206.15336},
  year={2022}
}

@inproceedings{cohen2018randomized,
  title={Randomized online matching in regular graphs},
  author={Cohen, Ilan Reuven and Wajc, David},
  booktitle={Proceedings of the Twenty-Ninth Annual ACM-SIAM Symposium on Discrete Algorithms},
  pages={960--979},
  year={2018},
  organization={SIAM}
}

@article{fahrbach2022edge,
  title={Edge-weighted online bipartite matching},
  author={Fahrbach, Matthew and Huang, Zhiyi and Tao, Runzhou and Zadimoghaddam, Morteza},
  journal={Journal of the ACM},
  volume={69},
  number={6},
  pages={1--35},
  year={2022},
  publisher={ACM New York, NY}
}

@inproceedings{ashlagi2019edge,
  title={Edge weighted online windowed matching},
  author={Ashlagi, Itai and Burq, Maximilien and Dutta, Chinmoy and Jaillet, Patrick and Saberi, Amin and Sholley, Chris},
  booktitle={Proceedings of the 2019 ACM Conference on Economics and Computation},
  pages={729--742},
  year={2019}
}

@article{huang2019online,
  title={Online vertex-weighted bipartite matching: Beating $1-1/e$ with random arrivals},
  author={Huang, Zhiyi and Tang, Zhihao Gavin and Wu, Xiaowei and Zhang, Yuhao},
  journal={ACM Transactions on Algorithms (TALG)},
  volume={15},
  number={3},
  pages={1--15},
  year={2019},
  publisher={ACM New York, NY, USA}
}

@inproceedings{zhou2021primal,
  title={A Primal-Dual Online Algorithm for Online Matching Problem in Dynamic Environments},
  author={Zhou, Yu-Hang and Hu, Peng and Liang, Chen and Xu, Huan and Huzhang, Guangda and Feng, Yinfu and Da, Qing and Wang, Xinshang and Zeng, An-Xiang},
  booktitle={Proceedings of the AAAI Conference on Artificial Intelligence},
  volume={35},
  number={12},
  pages={11160--11167},
  year={2021}
}

@article{DBLP:journals/corr/abs-2002-10697,
  author    = {Faez Ahmed and
               John Dickerson and
               Mark Fuge},
  title     = {Forming Diverse Teams from Sequentially Arriving People},
  journal   = {CoRR},
  volume    = {abs/2002.10697},
  year      = {2020},
  archivePrefix = {arXiv}
}

@article{goyal2020online,
  title={Online allocation of reusable resources via algorithms guided by fluid approximations},
  author={Goyal, Vineet and Iyengar, Garud and Udwani, Rajan},
  journal={arXiv preprint arXiv:2010.03983},
  year={2020}
}

@article{gong2021online,
  title={Online assortment optimization with reusable resources},
  author={Gong, Xiao-Yue and Goyal, Vineet and Iyengar, Garud N and Simchi-Levi, David and Udwani, Rajan and Wang, Shuangyu},
  journal={Management Science},
  year={2021},
  publisher={INFORMS}
}

@article{feng2019linear,
  title={Linear programming based online policies for real-time assortment of reusable resources},
  author={Feng, Yiding and Niazadeh, Rad and Saberi, Amin},
  journal={Chicago Booth Research Paper},
  number={20-25},
  year={2019}
}

@article{ma2021fairness,
  title={Fairness Maximization among Offline Agents in Online-Matching Markets},
  author={Ma, Will and Xu, Pan and Xu, Yifan},
  journal={arXiv preprint arXiv:2109.08934},
  year={2021}
}

@article{teac21,
  title={Allocation problems in ride-sharing platforms: Online matching with offline reusable resources},
  author={Dickerson, John P and Sankararaman, Karthik A and Srinivasan, Aravind and Xu, Pan},
  journal={ACM Transactions on Economics and Computation (TEAC)},
  volume={9},
  number={3},
  pages={1--17},
  year={2021},
  publisher={ACM New York, NY, USA}
}

@inproceedings{kaplan2022online,
  title={Online Weighted Matching with a Sample},
  author={Kaplan, Haim and Naori, David and Raz, Danny},
  booktitle={Proceedings of the 2022 Annual ACM-SIAM Symposium on Discrete Algorithms (SODA)},
  pages={1247--1272},
  year={2022},
  organization={SIAM}
}

@inproceedings{korula2009algorithms,
  title={Algorithms for secretary problems on graphs and hypergraphs},
  author={Korula, Nitish and P{\'a}l, Martin},
  booktitle={International Colloquium on Automata, Languages, and Programming},
  pages={508--520},
  year={2009},
  organization={Springer}
}

@inproceedings{huang2021online,
  title={Online stochastic matching, poisson arrivals, and the natural linear program},
  author={Huang, Zhiyi and Shu, Xinkai},
  booktitle={Proceedings of the 53rd Annual ACM SIGACT Symposium on Theory of Computing},
  pages={682--693},
  year={2021}
}

@inproceedings{karande2011online,
  title={Online bipartite matching with unknown distributions},
  author={Karande, Chinmay and Mehta, Aranyak and Tripathi, Pushkar},
  booktitle={Proceedings of the forty-third annual ACM symposium on Theory of computing},
  pages={587--596},
  year={2011},
  organization={ACM}
}

@inproceedings{mahdian2011online,
  title={Online bipartite matching with random arrivals: an approach based on strongly factor-revealing lps},
  author={Mahdian, Mohammad and Yan, Qiqi},
  booktitle={Proceedings of the forty-third annual ACM symposium on Theory of computing},
  pages={597--606},
  year={2011},
  organization={ACM}
}

@inproceedings{xu2020trade,
  title={Trading the System Efficiency for the Income Equality of Drivers in Rideshare},
  author={Xu, Yifan and Xu, Pan},
  year={2020},
  pages     = {4199--4205},
  booktitle={Proceedings of the Twenty-Ninth International Joint Conference on Artificial Intelligence}
}

@article{Manshadi2020OnlinePF,
  title={Online Policies for Efficient Volunteer Crowdsourcing},
  author={Vahideh Manshadi and Scott Rodilitz},
  journal={Proceedings of the 21st ACM Conference on Economics and Computation},
  year={2020}
}

@article{brubach2020online,
  title={Online stochastic matching: New algorithms and bounds},
  author={Brubach, Brian and Sankararaman, Karthik Abinav and Srinivasan, Aravind and Xu, Pan},
  journal={Algorithmica},
  volume    = {82},
  number    = {10},
  pages     = {2737--2783},
  year={2020},
  publisher={Springer}
}

@inproceedings{ho2012online,
  title={Online task assignment in crowdsourcing markets},
  author={Ho, Chien-Ju and Vaughan, Jennifer Wortman},
  booktitle={Twenty-sixth AAAI conference on artificial intelligence},
  year={2012}
}

@inproceedings{dickerson2018assigning,
  title={Assigning tasks to workers based on historical data: Online task assignment with two-sided arrivals},
  author={Dickerson, John P and Sankararaman, Karthik Abinav and Srinivasan, Aravind and Xu, Pan},
  booktitle={Proceedings of the 17th International Conference on Autonomous Agents and MultiAgent Systems},
  pages={318--326},
  year={2018},
  organization={International Foundation for Autonomous Agents and Multiagent Systems}
}

@article{buchbinder2009design,
  title={The design of competitive online algorithms via a primal--dual approach},
  author={Buchbinder, Niv and Naor, Joseph Seffi and others},
  journal={Foundations and Trends{\textregistered} in Theoretical Computer Science},
  volume={3},
  number={2--3},
  pages={93--263},
  year={2009},
  publisher={Now Publishers, Inc.}
}

@inproceedings{kess-stoc,
author = {Kesselheim, Thomas and T\"{o}nnis, Andreas and Radke, Klaus and V\"{o}cking, Berthold},
title = {Primal Beats Dual on Online Packing LPs in the Random-Order Model},
year = {2014},
isbn = {9781450327107},
publisher = {Association for Computing Machinery},
address = {New York, NY, USA},
url = {https://doi.org/10.1145/2591796.2591810},
doi = {10.1145/2591796.2591810},
booktitle = {Proceedings of the Forty-Sixth Annual ACM Symposium on Theory of Computing},
pages = {303–312},
numpages = {10},
location = {New York, New York},
series = {STOC ’14}
}

@inproceedings{kess13,
  title={An optimal online algorithm for weighted bipartite matching and extensions to combinatorial auctions},
  author={Kesselheim, Thomas and Radke, Klaus and T{\"o}nnis, Andreas and V{\"o}cking, Berthold},
  booktitle={European Symposium on Algorithms},
  pages={589--600},
  year={2013},
  organization={Springer}
}

@inproceedings{nanda2020balancing,
  title={Balancing the tradeoff between profit and fairness in rideshare platforms during high-demand hours},
  author={Nanda, Vedant and Xu, Pan and Sankararaman, Karthik Abhinav and Dickerson, John and Srinivasan, Aravind},
  booktitle={Proceedings of the AAAI Conference on Artificial Intelligence},
  volume={34},
  number={02},
  pages={2210--2217},
  year={2020}
}

@article{mehta2012online,
  author    = {Aranyak Mehta},
  title     = {Online Matching and Ad Allocation},
  journal   = {Foundations and Trends in Theoretical Computer Science},
  volume    = {8},
  number    = {4},
  pages     = {265--368},
  year      = {2013},
}

@inproceedings{kvv,
  author    = {Richard M. Karp and
               Umesh V. Vazirani and
               Vijay V. Vazirani},
  title     = {An Optimal Algorithm for On-line Bipartite Matching},
  booktitle = {Proceedings of the 22nd Annual {ACM} Symposium on Theory of Computing},
  series = {STOC '90},
  pages     = {352--358},
  year      = {1990},
}

 \clearpage
\appendix

 \section{Why Knowing $n$ Is Essential in the Secretary Problem}\label{app:15-a}
 The result in the following claim is implicitly suggested by the work of~\citet{abdel1982secretary}. For completeness, we include a formal proof.
\begin{claim}\label{lem:8-1-1}
Every online policy is zero-competitive for the classical Secretary Problem if the total number $n$ of arriving agents is unknown.
\end{claim}

Observe that for the Secretary Problem, maximizing the value associated with the picked candidate is equivalent to maximizing the probability of picking the best candidate (with the highest value), considering the metric of competitiveness. Thus, we can focus on policies that consider only the best candidate arriving so far, ignoring any second-best or worse candidates.\footnote{This can be seen as follows: Consider an instance where any two distinct candidates have values differing by a multiplicative factor of at least $1/\epsilon$. We claim that for any policy, the scenario when it picks any non-best candidate is equivalent to it getting nothing with a loss of no more than $\ep$ in the resulting competitiveness.}

\begin{proof}
Consider a given policy $\alg(\mathbf{q})$ parameterized by a vector $\mathbf{q} = (q_t)_{t}$, where for each $t = 1, 2, \ldots$, $q_t \in [0, 1]$ denotes the probability that $\alg$ will accept the candidate arriving at time $t$ given it is the best (i.e., its value is the largest so far). We say $\alg$ enters state $t$ if (\tbf{E1}) there is a candidate arriving at time $t$ (i.e., $n \ge t$) and (\tbf{E2}) $\alg$ does not pick any candidate before $t$. Note that (\tbf{E1}) is not a random event but is determined by the value $n$ pre-arranged by an adversary, while (\tbf{E2})  is a random event governed by $\alg$ itself. Let $s_t$ be the resulting probability that $\alg$ enters state $t$ given $n \ge t$. Therefore, we have

\begin{align}\label{seq:8-1-1}
&s_1 = 1,~s_2 = 1 - q_1,~s_3 = (1 - q_1) \, (1 - q_2/2); ~~s_{k+1} = s_k \, \left(1 - q_{k}/k\right), \forall k = 1, 2, \ldots.
\end{align}

Suppose $n \ge t$. Observe that when $\alg$ enters state $t$ (with probability $s_t$) and finds that the candidate arriving at time $t$, denoted by $a_t$, is the best so far (with probability $1/t$) and picks $a_t$ (with probability $q_t$),  the agent $a_t$ being the best among the whole pool of $n$ candidates occurs with probability $t/n$. Therefore, we claim that given $n \ge t$, $\alg$ picks the best among the pool at time $t$ with probability equal to $s_t \cdot (1/t) \cdot (t/n) \cdot q_t=s_t q_t/n$. Let $P(n, \bfq)$ denote the probability that the policy $\alg(\bfq)$ picks the best candidate among an input of $n$ candidates. Thus,
\begin{align*}
P(n, \bfq)=\sum_{t=1}^n s_t \, q_t/n=\frac{1}{n} \sum_{t=1}^n s_t \, q_t,
\end{align*}
where the sequence $(s_k)_k$ is defined in~\eqref{seq:8-1-1}. Let $\mathbb{N}$ denote the set of all positive natural numbers. The problem of identifying an optimal policy for the Secretary Problem with unknown $n$ (the number of arriving agents) can be formulated as the below maxmin program:
\begin{align}\label{pro:8-1-1}
\max_{\bfq=(q_t): q_t \in [0,1] \forall t \in \mathbb{N}}~~{V(\bfq):= \inf_{n \in \mathbb{N}} \bB{P(n, \bfq):=\frac{1}{n} \sum_{t=1}^n s_t \cdot q_t}}.
\end{align}

Observe that any policy must specify the sequence $\mathbf{q} = (q_t)$ in the absence of any knowledge of $n$. By Lemma~\ref{lem:8-1-2}, we see that for any given sequence $\mathbf{q} = (q_t)$ with $q_t \in [0, 1]$ for all $t \in \mathbb{N}$, $P(n, \mathbf{q}) \to 0$ as $n \to \infty$. This implies that for any $\mathbf{q}$, $V(\mathbf{q}) = 0$, suggesting that any policy characterized by $\mathbf{q}$ can pick the best candidate with probability approaching zero when $n$ is sufficiently large.
\end{proof}

\begin{lemma}\label{lem:8-1-2}
For any given sequence $\bfq=(q_t)$ with $q_t \in [0,1]$ for all $t \in \mathbb{N}$, $P(n, \bfq) \to 0$ when $n \to \infty$. 
\end{lemma}
We give a direct proof based on the definition of limits.
\begin{proof}
Observe that the sequence $\bfs:=(s_t)$ is non-increasing over $t$ by definition in~\eqref{seq:8-1-1}. Thus, $\lim_{t \to \infty} s_t$ exists. Let $s^*=\lim_{t \to \infty} s_t \in [0,1]$.  Consider the following two cases:

\tbf{Case 1}: $s^*=0$. By the definition of limits, for any given $\ep>0$, there exists an $M>0$ such that for any $t \ge M$, $s_t \le \ep$. Thus, for any $n >M$, 
\begin{align*}
P(n, \bfq)=\frac{1}{n} \sum_{t=1}^n s_t \, q_t \le \frac{1}{n} \sum_{t=1}^n s_t=
 \frac{1}{n} \sum_{t=1}^M s_t+ \frac{1}{n} \sum_{t=M+1}^n s_t \le \frac{M}{n}+\frac{n-M}{n} \ep \le  \frac{M}{n}+\ep.
\end{align*}
Therefore, we establish the claim.

\tbf{Case 2}: $s^* >0$.\footnote{This is possible: Consider the case when $q_1=0$ and $q_k=1/k$ for all $k \ge 2$, for example.} By definition, for each $t \ge 1$, 
\begin{align*}
s_{t+1}=\prod_{k=1}^t \bp{1-q_k/k}=\exp\bp{\sum_{k=1}^t \ln (1-q_k/k)}.
\end{align*}
$s^*>0$ suggests that the infinity series of $\sum_{k=1}^\infty q_k/k$ converges. By definition, for any given $\ep>0$, there exists an $M>0$ such that $\sum_{k=M}^{\infty} q_k/k \le \ep$. For any $n>M$, 

\begin{align*}
P(n, \bfq)=\frac{1}{n} \sum_{t=1}^n s_t \, q_t \le \frac{1}{n} \sum_{t=1}^n q_t=
 \frac{1}{n} \sum_{t=1}^M q_t+ \frac{1}{n} \sum_{t=M+1}^n q_t \le \frac{M}{n}+\sum_{k=M+1}^n q_k/k \le  \frac{M}{n}+\ep.
\end{align*}
Therefore, we establish the claim.
\end{proof}

\xhdr{Remarks on Claim~\ref{lem:8-1-1} and the Proof}.

(1) The proof above is more than a proof. In fact, it offers a general framework from which we can solve an optimal policy for the Secretary Problem under various settings. For example, for the classical setting when $n$ is known, the optimal solution to Program~\ref{pro:8-1-1} recovers the well-known optimal $1/\sfe$-competitive policy. Specifically, in that case, Program~\ref{pro:8-1-1} is updated to 
\begin{align}\label{pro:8-1-2}
\max_{{\bfq=(q_t) \in [0,1]^n}}~~{V(\bfq)={\frac{1}{n} \sum_{t=1}^n s_t \, q_t}},
\end{align}
where the domain of $\bfq$ is updated to $[0,1]^n$. For any given $n$, we can verify that an optimal solution to the above program is as follows: $q_t=0$ for all $1 \le t \le \tau$ and $q_t=1$ for all $\tau < t \le n$, where $\tau/n \approx 1/\sfe$ when $n \to \infty$.

(2) For the setting where $n$ is unknown but follows a known or unknown distribution, Program~\ref{pro:8-1-1} can also be useful. For example, assuming $n \sim \text{Pois}(\lambda)$, a Poisson distribution with a known parameter $\lambda > 0$, finding an optimal policy is equivalent to solving the following updated program:
\begin{align}\label{pro:8-1-3}
\max_{\bfq=(q_t): q_t \in [0,1] \forall t \in \mathbb{N}}~~{V_\lam(\bfq):=\sum_{n=1}^{\infty} \left(\sfe^{-\lambda} \frac{\lambda^{n}}{n!}\right) \,\left(\frac{1}{n} \sum_{t=1}^n s_t \, q_t\right)}.
\end{align}
We can potentially apply the framework above to design robust policies that perform relatively well even when errors exist in the estimation of the parameters of the distribution of $n$ (e.g., errors in $\lambda$ given $n \sim \text{Pois}(\lambda)$). This suggests a possible direction for robust policies under misspecified arrival-count distributions.  
\section[Weighted-Variance Inflation Counterexample]
{A Counterexample Demonstrating Weighted Variance Inflation Beyond a Quadratic Factor of Edge-Weight Variation}
\label{app:why}

\begin{example}
Consider a generic algorithm $\ALG$ for the Edge-Weighted Online Matching problem under Random Order (\om).
Let $X$ denote the (random) number of matches produced by $\ALG$, let $Y$ denote the resulting total weight, and let $\bar{\omega}$ denote the maximum edge weight (or, more generally, a bound on edge-weight variation).
Although the inequality $Y \le \bar{\omega}\, X$ holds almost surely, the variance of $Y$ can nevertheless be much larger than $\bar{\omega}^2 \, \Var[X]$.
This shows that even when the number of matches exhibits limited variability, the total weight can fluctuate significantly due to heterogeneity in edge weights.

Consider the following simple distribution.
The random variable $X$ takes values $1$ and $k$, where $k \ge 2$, each with probability $1/2$.
Correspondingly, the random variable $Y$ takes values $1$ and $\bar{\omega}\, k$, each with probability $1/2$.
This models a situation in which $\ALG$ either matches a single edge of unit weight (with probability $1/2$), or matches $k$ edges, each with weight $\bar{\omega} > 1$ (with probability $1/2$).

We can verify that
\[
\Pr[Y \le \bar{\omega}\, X] = 1, \quad 
\Var[Y]
= \frac{(\bar{\omega}\, k - 1)^2}{4}
> \frac{(\bar{\omega}\, k - \bar{\omega})^2}{4}
= \bar{\omega}^2 \, \frac{(k-1)^2}{4}
= \bar{\omega}^2 \,\Var[X]. 
\]
\end{example}

This example highlights the additional technical challenges introduced by edge-weight heterogeneity when attempting to upper bound the variance of the total weight.
In this paper, our primary focus is on controlling the variance arising \emph{exclusively} from random factors associated with the input instance (namely, the random arrival order of online agents) and the algorithm itself. For this reason, in the variance analysis we restrict attention to the unweighted projection of the realized matching, namely its cardinality.

\section{Basic Properties of $\kappa(d,\theta)$, $\eta(d,\theta)$, and $\sigma(d,\theta)$}
\label{app:lem-sig}
The following lemma records basic properties of these functions used throughout the analysis.

\begin{lemma}
\label{lem:sig}
\begin{enumerate}
    \item For any integer $d \ge 1$ and any $\theta \in [0,1]$, we have
    \[
        \sigma(d, \theta) \;\ge\; \frac{1}{2}\kappa(d, \theta)
        \;\ge\; \frac{1}{2}\eta(d, \theta).
    \]
    \item For each $d \ge 1$, the function $\sigma(d, \theta)$ admits a unique maximizer over $\theta \in [0,1]$, denoted by $\theta^*_{\sigma}(d)$.
    \item For each $d \ge 1$, the function $\kappa(d, \theta)$ admits a unique maximizer over $\theta \in [0,1]$, denoted by $\theta^*_{\kappa}(d)$. Moreover, the sequence $\{\theta^*_{\kappa}(d)\}_{d \ge 1}$ is increasing in $d$, while
    \[
        \kappa^*(d) := \kappa\bigl(d, \theta^*_{\kappa}(d)\bigr)
    \]
    is decreasing in $d$.
\end{enumerate}
\end{lemma}

We decompose the proof into three auxiliary lemmas.

\begin{lemma}
\label{lem:sig-1}
For any integer $d \ge 1$ and any $\ta \in [0,1]$, we have
\[
\sig(d,\ta) \ge \frac{\kap(d,\ta)}{2} \ge \frac{\eta(d,\ta)}{2}.
\]
\end{lemma}

\begin{lemma}
\label{lem:sig-2}
For any integer $d \ge 1$, the function $\sig(d,\ta)$ admits a unique maximizer over $\ta \in [0,1]$, denoted by $\ta^*_\sig(d)$.
\end{lemma}

\begin{lemma}
\label{lem:sig-3}
For each integer $d \ge 1$, the function $\kap(d,\ta)$ admits a unique maximizer over $\ta \in [0,1]$, denoted by $\ta^*_\kap(d)$.
Moreover, the sequence $\{\ta^*_\kap(d)\}_{d \ge 1}$ is increasing in $d$, while
\[
\kap^*(d) := \kap\bigl(d,\ta^*_\kap(d)\bigr)
\]
is decreasing in $d$.
\end{lemma}

\begin{proof}[\tbf{Proof of Lemma~\ref{lem:sig-1}}]
We first show that $\sig(d,\ta) \ge \kap(d,\ta)/2$.
Recall that
\[
\sig(d,\ta) = \kap(d,\ta) - \frac{\ta(1-\ta)}{2}.
\]
Thus, it suffices to prove that $\kap(d,\ta) \ge \ta(1-\ta)$.

When $d=1$, we have $\kap(1,\ta)=1-\ta \ge \ta(1-\ta)$.
Hence, assume $d\ge 2$. Observe that
\begingroup
\allowdisplaybreaks
\begin{align*}
\kap(d,\ta)
&= \int_\ta^1 \sd z \, \Bigl( (1-z)^{d-1}
    + \bigl(1-(1-z)^{d-1}\bigr)\frac{\ta}{z} \Bigr) \\
&= \frac{(1-\ta)^d}{d}
  + \ta \int_\ta^1 \sd z \,\frac{1-(1-z)^{d-1}}{z} \\
&= \frac{(1-\ta)^d}{d}
  + \ta \int_\ta^1 \sd z \,
    \bigl(1 + (1-z) + \cdots + (1-z)^{d-2}\bigr) \\
&= \frac{(1-\ta)^d}{d}
  + \ta \int_0^{1-\ta} \sd z \,
    \bigl(1 + z + \cdots + z^{d-2}\bigr) \\
&= \frac{(1-\ta)^d}{d}
  + \ta \sum_{i=1}^{d-1} \frac{(1-\ta)^i}{i}.
\end{align*}
\endgroup

For $\ta=0$, the desired inequality is immediate. Now suppose
$\ta\in(0,1]$ and set $q:=1-\ta$. Since
\[
-\ln \ta = \sum_{i=1}^{\infty}\frac{q^i}{i},
\]
we have
\[
\ta \sum_{i=d}^{\infty}\frac{q^i}{i}
\le
\frac{\ta}{d}\sum_{i=d}^{\infty}q^i
=
\frac{\ta}{d}\frac{q^d}{1-q}
=
\frac{q^d}{d}.
\]
Therefore,
\begin{align*}
\kap(d,\ta)
&= \frac{q^d}{d}
   + \ta\sum_{i=1}^{d-1}\frac{q^i}{i} \ge
\ta\sum_{i=d}^{\infty}\frac{q^i}{i}
+\ta\sum_{i=1}^{d-1}\frac{q^i}{i} =
\ta\sum_{i=1}^{\infty}\frac{q^i}{i}
=
-\ta\ln\ta \ge \ta(1-\ta),
\end{align*}
where the last inequality follows from the elementary bound
$-\ln x \ge 1-x$ for $x\in(0,1]$.

This establishes $\sig(d,\ta) \ge \kap(d,\ta)/2$.
\smallskip
Next, we show that $\kap(d,\ta) \ge \eta(d,\ta)$.
Consider the following random experiment.
We toss a biased coin with success probability $\ta$ up to $d$ times, stopping as soon as the first success appears.
Let $D$ denote the number of tosses before stopping (or $d$ if no success occurs).

We have
\begin{align*}
\E[D]
&= \sum_{i=1}^d \Pr[D \ge i]
 = \sum_{i=1}^d (1-\ta)^{i-1}, \\
\E[1/D]
&= \sum_{i=1}^{d-1} \frac{1}{i}(1-\ta)^{i-1}\ta
  + \frac{1}{d}(1-\ta)^{d-1}.
\end{align*}
By Jensen's inequality,
\[
\E[1/D] \ge \frac{1}{\E[D]}.
\]
Substituting the expressions above gives
\[
\sum_{i=1}^{d-1} \frac{1}{i}(1-\ta)^{i-1}\ta
+ \frac{1}{d}(1-\ta)^{d-1}
\ge \frac{1}{\sum_{i=1}^d (1-\ta)^{i-1}}.
\]
Multiplying both sides by $1-\ta$ yields $\kap(d,\ta) \ge \eta(d,\ta)$, completing the proof.
\end{proof}

\begin{proof}[\tbf{Proof of Lemma~\ref{lem:sig-2}}]
We first verify the cases $d=1$ and $d=2$.
In both cases, the function $\sig(d,\ta)$ is monotonically decreasing over $\ta \in [0,1]$.
Therefore, $\sig(d,\ta)$ admits a unique maximizer on $[0,1]$, which must occur at the left endpoint.

We now focus on the case $d \ge 3$.
It suffices to show that, for any fixed $d \ge 3$, the first derivative $\partial \sig(d,\ta)/\partial \ta$ is decreasing over $\ta \in [0,1]$.
This implies that $\sig(d,\ta)$ is concave on $[0,1]$, and hence admits a unique maximizer.

Recall that
\[
\sig(d,\ta) = \kap(d,\ta) - \frac{\ta(1-\ta)}{2}.
\]
Differentiating with respect to $\ta$ yields
\begin{align*}
\frac{\partial \sig(d,\ta)}{\partial \ta}
&= \frac{\partial \kap(d,\ta)}{\partial \ta} + \ta - \frac{1}{2}.
\end{align*}

Using the integral representation of $\kap(d,\ta)$, we obtain
\begin{align*}
\frac{\partial \kap(d,\ta)}{\partial \ta}
&= \int_\ta^1 \sd \zeta \,\frac{1-(1-\zeta)^{d-1}}{\zeta} - 1,
\end{align*}
and hence
\begin{align*}
\frac{\partial \sig(d,\ta)}{\partial \ta}
&= \int_\ta^1 \sd \zeta \,\frac{1-(1-\zeta)^{d-1}}{\zeta}
   + \ta - \frac{3}{2} = \sum_{\ell=2}^{d-1} \frac{(1-\ta)^\ell}{\ell} - \frac{1}{2}.
\end{align*}

The final expression is clearly decreasing in $\ta$ over $[0,1]$, since each term $(1-\ta)^\ell/\ell$ is decreasing in $\ta$.
Therefore, $\sig(d,\ta)$ is concave on $[0,1]$ for all $d \ge 3$, which implies the existence and uniqueness of the maximizer $\ta^*_\sig(d)$.
\end{proof}

\begin{proof}[\tbf{Proof of Lemma~\ref{lem:sig-3}}]
We first establish the existence and uniqueness of the maximizer of $\kap(d,\ta)$.

When $d=1$, we have $\kap(1,\ta)=1-\ta$, which is strictly decreasing over $\ta \in [0,1]$.
Thus, the unique maximizer is $\ta^*_\kap(1)=0$.

Now consider $d \ge 2$.
Differentiating $\kap(d,\ta)$ with respect to $\ta$ yields
\[
\frac{\partial \kap(d,\ta)}{\partial \ta}
= \left(\int_\ta^1 \sd \zeta \,\frac{1-(1-\zeta)^{d-1}}{\zeta} \right) - 1
=: F(d,\ta).
\]
We verify the following properties.
\begin{enumerate}
\item For any fixed $d \ge 2$, the function $F(d,\ta)$ is strictly decreasing over $\ta \in [0,1]$.
\item We have $F(d,1) = -1$, and
\[
F(d,0)
= \left(\int_0^1 \sd \zeta \,\frac{1-(1-\zeta)^{d-1}}{\zeta} \right)- 1
\ge \left(\int_0^1 \sd \zeta \,\frac{1-(1-\zeta)}{\zeta}\right) - 1
= 0.
\]
\end{enumerate}
By continuity and monotonicity, there exists a unique $\ta^*_\kap(d) \in [0,1]$ satisfying $F(d,\ta^*_\kap(d))=0$.
This establishes that $\kap(d,\ta)$ admits a unique maximizer over $\ta \in [0,1]$.

We now prove the monotonicity statements.
For any fixed $\ta \in [0,1]$, the function $F(d,\ta)$ is increasing in $d$ for $d=1,2,\ldots$.
Consequently, the unique root $\ta^*_\kap(d)$ of $F(d,\ta)=0$ is increasing in $d$.

Finally, observe that for any fixed $\ta \in [0,1]$, the function $\kap(d,\ta)$ is decreasing in $d$.
Therefore, the maximal value
\[
\kap^*(d) := \max_{\ta \in [0,1]} \kap(d,\ta)
\]
is also decreasing in $d$, completing the proof.
\end{proof}



\section{Proof of Lemma~\ref{lem:eqv}}
\label{app:lem-eqv}

We begin with the following auxiliary result, which will be used to establish Lemma~\ref{lem:eqv}.

\begin{lemma}
\label{lem:eqv-1}
Fix a graph $G=(I,J,E)$ and a sample set $\cS_K \subseteq J$.
Then the matching $\cM_K$ produced by $\smg$ after the sampling phase coincides with the set $\am$ produced by $\vir$.
\end{lemma}

\begin{proof}
We prove the claim by induction on $K = |\cS_K|$.

\smallskip
\noindent
\textbf{Base case ($K=1$).}
Suppose $\cS_K=\{j\}$.
Let $e^*$ be the edge of maximum weight in $E_j$, the set of edges incident to $j$.
Since the induced graph $\cG_K=(I,\cS_K)$ is a star centered at $j$, the greedy algorithm used in $\smg$ selects exactly the edge $e^*$.
Hence, $\cM_K=\{e^*\}$.

For the virtual algorithm $\vir$, observe that prior to processing any edge incident to $j$, we have $\am=\emptyset$.
Among all edges in $E_j$, the first one encountered by $\vir$ is $e^*$, which is then added to $\am$.
No other edge incident to $j$ can be added thereafter.
Thus, $\am=\{e^*\}$, and we conclude that $\cM_K=\am$ in the base case.

\smallskip
\noindent
\textbf{Inductive step.}
Assume that the claim holds for every bipartite graph and every sample set of size at most $N$, and consider a sample set $\cS_K$ with $K=N+1$.
Let $e^*=(i,j)$ be the maximum-weight edge in the induced graph $\cG_K=(I,\cS_K)$.

Under $\smg$, the Greedy algorithm first selects $e^*$.
After this selection, no other edge incident to either $i$ or $j$ can be selected.
Hence, the remainder of the Greedy procedure is exactly the Greedy algorithm applied to the residual graph
\[
G' := G\bigl (I\setminus\{i\},\, J\setminus\{j\}\bigr)
\]
with sample set
\[
\cS' := \cS_K\setminus\{j\},
\qquad |\cS'|=N.
\]
Let $\cM'$ denote the matching produced by Greedy on the induced sampled subgraph of $G'$.
Then
\[
\cM_K=\{e^*\}\cup \cM'.
\]

Under $\vir$, the first edge added to $\am$ is also $e^*$.
After $e^*$ is added to $\am$, every subsequent edge incident to $i$ fails the matching condition in Step~(4), while every subsequent edge incident to $j$ is ignored because $j$ has already been marked as processed.
Consequently, the subsequent construction of $\am$ is identical to running $\vir$ on the residual graph $G'$ with sample set $\cS'$.
Thus,
\[
\am=\{e^*\}\cup \widetilde{\cM}^{a},
\]
where $\widetilde{\cM}^{a}$ denotes the set produced by $\vir$ on $G'$ with sample set $\cS'$.

Since $|\cS'|=N$, the inductive hypothesis gives
\[
\cM'=\widetilde{\cM}^{a}.
\]
Therefore,
\[
\cM_K=\am.
\]
This completes the induction.
\end{proof}

\begin{proof}[\tbf{Proof of Lemma~\ref{lem:eqv}}]
We begin by considering a slight modification of $\smg$ in Algorithm~\ref{alg:smg}.
For each arrival $j_t$, define
\[
\tce_t := \{\, e=(i,j_t) \in E_{j_t} : w(e) \ge p(i) \,\}.
\]
Replace Step~\eqref{alg:s5} of $\smg$ with the following rule:

\smallskip
\emph{Add the maximum-weight edge in $\tce_t$ to $\tm$ if $\tce_t \neq \emptyset$, and reject $j_t$ otherwise,}
\smallskip

\noindent
where $\tm$ is initialized as the empty set and represents the final output of the modified algorithm.
That is, whenever $\tce_t$ is nonempty, $\smg$ adds the edge $e_t=(i_t,j_t)$ of maximum weight in $\tce_t$ to $\tm$, regardless of whether $i_t$ is already matched.

We first show that for any fixed graph $G$ and sample set $\cS_K$, the resulting set $\tm$ coincides with $\bm$, the pseudo-matching produced by $\vir$.

Fix a round $t$ and the arriving online agent $j_t$.
By definition, $e_t=(i_t,j_t)$ is the edge of maximum weight in $\tce_t$.
Suppose that $(i_t,j') \in \cM_K$ for some $j' \in \cS_K$.
To avoid ambiguity, we say that an edge $e$ is \emph{considered} in $\vir$ if it reaches Step~\eqref{alg:v1}, and \emph{processed} if it reaches Step~\eqref{alg:v2}.

Note the following observations.
\begin{enumerate}
\item The edge $(i_t,j_t)$ is considered before $(i_t,j')$ in $\vir$, since it has strictly larger weight.
\item At the time $(i_t,j_t)$ is considered, the offline agent $i_t$ must be unmatched; otherwise $(i_t,j') \notin \am$, and hence $(i_t,j') \notin \cM_K$ by Lemma~\ref{lem:eqv-1}.
\end{enumerate}

Since $e_t$ is the maximum-weight edge in $\tce_t$, it follows that $e_t$ is the unique edge incident to $j_t$ that is processed by $\vir$.
Moreover, because $j_t \notin \cS_K$, the edge $e_t$ is added to $\bm$.
Hence, every edge added to $\tm$ is also added to $\bm$, implying $\tm \subseteq \bm$.
A symmetric argument shows $\bm \subseteq \tm$, and therefore $\tm = \bm$.

Finally, for each offline agent $i \in I$, let $\tilde{\cE}_i$ denote the set of edges incident to $i$ in $\tm$, and let $\tilde{\cJ}_i$ be the corresponding set of online neighbors.
The final matching $\cM$ returned by $\smg$ can equivalently be viewed as selecting a single edge $e=(i,j^*) \in \tilde{\cE}_i$, where $j^*$ is the first arriving agent among $\tilde{\cJ}_i$.

Since all agents in $J \setminus \cS_K$ arrive in uniformly random order, each $j \in \tilde{\cJ}_i$ has an equal probability of arriving first.
Equivalently, each edge in $\tilde{\cE}_i$ is selected with equal probability.
Thus, $\cM$ can be viewed as being obtained from $\tm$ by independently sampling one edge uniformly at random from $\tilde{\cE}_i$ for each $i \in I$.

Because $\tm = \bm$ and $\tilde{\cE}_i = \cE_i^b$, this sampling procedure is exactly how $\bam$ is constructed in $\vir$.
We conclude that $\cM$ and $\bam$ are statistically identical.
\end{proof}


\section{Modular Instances Attaining Tightness for Inequalities~(\ref{ineq:bm}) and~(\ref{ineq:smg-1})}
\label{app:tight}

\subsection{Z-Shaped Modular Instances Attaining Tightness of Inequality~\eqref{ineq:bm}}
\label{app:z-tight}

Consider an input graph $G = (I,J,E)$ consisting of $m/2 = |I|/2$ identical modules.
Each module, illustrated in Figure~\ref{fig:motiv-a}, is an edge-weighted $Z$-shaped bipartite graph.

\begin{figure}[ht!]
\centering
\begin{minipage}{0.45\linewidth}
\centering
\begin{tikzpicture}[scale=0.9]
    \node[circle, draw, thick, minimum size=7mm] (i1) at (0,0) {$i_1$};
    \node[circle, draw, thick, minimum size=7mm] (i2) at (0,-2) {$i_2$};
    
    \node[circle, draw, thick, minimum size=7mm] (j1) at (4,0) {$j_1$};
    \node[circle, draw, thick, minimum size=7mm] (j2) at (4,-2) {$j_2$};

    \draw[ultra thick] (i1) -- node[above, midway, blue] {$1$} (j1);
    \draw[ultra thick] (i2) -- node[above, midway, blue] {$1$} (j2);
    \draw[ultra thick] (i2) -- node[above, midway, blue] {$1+\ep$} (j1);
\end{tikzpicture}
\end{minipage}
\hfill
\begin{minipage}{0.48\linewidth}
\begin{align*}
&|I| = m, \quad |J| = m;\\
&\OPT = 2;\\
&\ta = 1-\delta;\\
&\E[w(\cM^b)] = (1+\ep)\,\delta + \delta^2.
\end{align*}
\end{minipage}
\caption{
A modular instance demonstrating the tightness of inequality~\eqref{ineq:bm}.
Each module is a $Z$-shaped edge-weighted bipartite graph.
For this module, the auxiliary set $\cM^b$ produced by $\vir$ satisfies
$\E[w(\cM^b)] = (1+\ep)\,\delta + \delta^2$.
}
\label{fig:motiv-a}
\end{figure}

By concatenating $m/2$ independent copies of this module, the instance scales to arbitrary size without altering the tightness behavior.
For each module, the offline optimal value is $\OPT=2$.
Moreover, when $\ep = o(1)$ and $\delta = o(1)$, we have
\[
\E[w(\cM^b)] = \tfrac{\OPT}{2}\,\delta \, (1+o(1))
= \tfrac{\OPT}{2}\,(1-\ta)\,(1+o(1)),
\]
which shows that inequality~\eqref{ineq:bm} is asymptotically tight.

To see this, note that the edge $(i_2,j_1)$ is added to $\cM^b$ if and only if $j_1 \notin \cS_K$, which occurs with probability $1-\ta = \delta$.
Similarly, the edge $(i_2,j_2)$ is added to $\cM^b$ if and only if both $j_1 \notin \cS_K$ and $j_2 \notin \cS_K$, which occurs with probability $\delta^2$.
Here we omit additive $O(1/n)$ terms arising from sampling without replacement.

\subsection{Star-Like Modular Instances Attaining Tightness of Inequality~\eqref{ineq:smg-1}}
\label{app:star-tight}

Consider an input graph $G = (I,J,E)$ consisting of $m = |I|$ identical modules.
Each module, illustrated in Figure~\ref{fig:motiv}, is an unweighted star graph centered at an offline node $i^*$ with $d$ online neighbors.

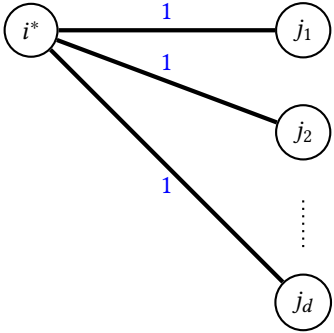
\begin{figure}[ht!]
\centering
\begin{minipage}{0.45\linewidth}
\centering
\begin{tikzpicture}[scale=0.9]
    \node[circle, draw, thick, minimum size=7mm] (i) at (0,0) {$i^*$};

    \node[circle, draw, thick, minimum size=7mm] (j1) at (4,0) {$j_1$};
    \node[circle, draw, thick, minimum size=7mm] (j2) at (4,-1.5) {$j_2$};
    \node[circle, draw, thick, minimum size=7mm] (jd) at (4,-4) {$j_d$};

    \draw[dotted, thick] (4,-2.5) -- (4,-3.2);

    \draw[ultra thick] (i) -- node[above, midway, blue] {$1$} (j1);
    \draw[ultra thick] (i) -- node[above, midway, blue] {$1$} (j2);
    \draw[ultra thick] (i) -- node[below, midway, blue] {$1$} (jd);
\end{tikzpicture}
\end{minipage}
\hfill
\begin{minipage}{0.48\linewidth}
\begin{align*}
&|I| = m, \quad |J| = n = m\,d;\\
&w_e = 1, \quad \forall e \in E;\\
&\E[w(\bm)] = \sum_{\ell=1}^d (1-\ta)^\ell;\\
&\E[w(\bam)] = 1-\ta.
\end{align*}
\end{minipage}
\caption{
A modular instance demonstrating the tightness of inequality~\eqref{ineq:smg-1}.
Each module is an unweighted star centered at an offline node $i^*$ with $d$ online neighbors.
For this module, the auxiliary sets $\bm$ and $\bam$ produced by $\vir$ satisfy
$\E[w(\bm)] = \sum_{\ell=1}^d (1-\ta)^\ell$ and $\E[w(\bam)] = 1-\ta$.
}
\label{fig:motiv}
\end{figure}

By concatenating $m$ independent copies of this module, the instance scales to arbitrary size without altering the tightness behavior.

\begin{lemma}
\label{lem:ineq-tig}
Consider a single module centered at an offline node $i^*$ as shown in Figure~\ref{fig:motiv}.
Let $\bm$ and $\bam$ denote the sets of edges incident to $i^*$ output by $\vir$.
Then
\[
\E[w(\bm)] = \sum_{\ell=1}^d (1-\ta)^\ell,
\qquad
\E[w(\bam)] = 1-\ta,
\qquad
\E[w(\cM^b)] = \E[w(\widehat{\cM})] \, \sum_{\ell=1}^d (1-\ta)^{\ell-1}.
\]
\end{lemma}

The lemma above implies that inequality~\eqref{ineq:smg-1} is tight for this family of instances.

\begin{proof}
Fix a constant $\ta \in [0,1]$ and let $[d] = \{j_1, j_2, \ldots, j_d\}$.
For any subset $S \subsetneq [d]$ and any $j \in [d] \setminus S$, we have
\begin{align*}
\Pr[j \in \cS_K \mid \wedge_{j' \in S} (j' \in \cS_K)]
&= \frac{K-|S|}{n-|S|}
= \ta - O(d/n), \\
\Pr[j \in \cS_K \mid \wedge_{j' \in S} (j' \notin \cS_K)]
&= \frac{K}{n-|S|}
= \ta + O(d/n).
\end{align*}
Thus, when $d = O(1)$ and $n \to \infty$, the events $\{j \in \cS_K : j \in [d]\}$ are asymptotically independent with marginal probability $\ta$.
For ease of exposition, we assume asymptotic independence in the calculations below and omit $O(d/n)$ terms.

Under this assumption,
\[
\E[w(\am)] = 1 - (1-\ta)^d,
\qquad
\E[w(\bm)] = \E[w(\am)] \, \frac{1-\ta}{\ta}
= \sum_{\ell=1}^d (1-\ta)^\ell.
\]

We now compute $\E[w(\bam)]$.
Let $j_1$ denote the first online agent incident to $i^*$ that is processed by $\vir$.
If $j_1 \in \cS_K$, which occurs with probability $\ta$, then no edge is added to $\bm$ and hence $w(\bam)=0$.
If instead $j_1 \notin \cS_K$, which occurs with probability $1-\ta$, then $\bm \neq \emptyset$ and exactly one unit-weight edge incident to $i^*$ is selected into $\bam$, yielding $w(\bam)=1$.
Therefore,
\[
\E[w(\bam)] = 0 \cdot \ta + 1 \cdot (1-\ta) = 1-\ta.
\]
\end{proof}


\section{Proof of Theorem~\ref{thm:smg-cr-tig}}
\label{app:smg-cr-tight}

\subsection{Proof of Lemma~\ref{lem:bam-b}}

\begin{proof}
We apply a local-search-based analysis to identify the worst-case scenario.
Assume without loss of generality that $e^* = (i^*=1, j^*=1)$ is the first edge processed in $\vir$, and hence has the largest weight among all edges.\footnote{The same analysis can be applied recursively to the residual instance after $\vir$ processes the first edge.}

\noindent\textbf{Case 1:} $e^* \in M^*$.
Let $\tE \subseteq E_{i^*}$ be the set of edges incident to $i^*$ that may be processed in Step~\eqref{alg:v2} of $\vir$, conditioning on the event $j^* \notin \cS_K$.
Since $e^*$ has the largest weight, we have $e^* \in \tE$.
For each $e = (i^*, j) \in \tE$, let $\chi_e = 1$ if $j \notin \cS_K$.
Let $\beta$ denote the expected weight contributed to $\bam$ from the edge incident to $i^*$, and let $\gam$ denote the weight of the edge incident to $i^*$ in the offline optimal matching $M^*$. 
Here $\gam = w(e^*) =: w^*$, and our goal is to lower bound $\beta/\gam$.
Let $|\tE| = \td \le d$.
Then
\begin{align*}
\beta
&= w(e^*) \, \E[\chi_{e^*}] \, \E\bB{\frac{1}{\sum_{e \in \tE} \chi_e} ~\Big|~ \chi_{e^*} = 1} \\
&= w^* \, (1-\ta) \, \bP{1 \, \ta + \frac{1}{2}(1-\ta)\,\ta + \cdots + \frac{1}{\td-1} (1-\ta)^{\td-2}\,\ta + \frac{1}{\td} (1-\ta)^{\td-1}} \\
&= w^* \, \kap(\td, \ta),
\end{align*}
where $\kap(\cdot,\cdot)$ is defined in~\eqref{eqn:kap}.
As before, we omit $O(d/n)$ terms by assuming $d$ is constant and $n \gg 1$; see the discussion in the proof of Lemma~\ref{lem:ineq-tig}.
For any fixed $\ta \in [0,1]$, $\kap(\td,\ta)$ is non-increasing in $\td \in \{1,2,\ldots,d\}$.
Therefore,
\begin{align}
\frac{\beta}{\gam} = \frac{\beta}{w^*} \ge \kap(d,\ta).
\label{ineq:case1}
\end{align}
Thus, in this case the ratio $\beta/\gam$ can be as low as $\kap(d,\ta)$.

\medskip

\noindent\textbf{Case 2:} $e^* \notin M^*$.
Assume without loss of generality that $\te = (i^*, j=2) \in M^*$ and $\bae = (\bai, j^*) \in M^*$.
Let $w^* := w(e^*)$, $\tw := w(\te)$, and $\baw := w(\bae)$.
By optimality of $M^*$, we have $w^* \le \tw + \baw$.
Let $\beta$ denote the expected weight gained from edges incident to $i^*$ and $\bai$ in $\bam$, and let $\gam$ denote the total weight of edges incident to $i^*$ and $\bai$ in $M^*$; thus $\gam = \tw + \baw$.

We now lower bound $\beta/\gam$.
Since $e^*=(i^*,j^*)$ is processed first, the edge $\bae$ can no longer be considered because $j^*$ becomes processed.
Hence the gain from edges incident to $\bai$ can be as low as zero (e.g., if $\bai$ has only one neighbor $j^*$).
We therefore focus on the gain from edges incident to $i^*$.
At least two edges, namely $e^*=(i^*,j^*)$ and $\te=(i^*,j=2)$, can contribute to $\bam$.
Write $\beta=\beta_1+\beta_2$, where $\beta_1$ and $\beta_2$ are the expected gains contributed by $e^*$ and $\te$, respectively.

\medskip

As in \textbf{Case 1}, let $\tE \subseteq E_{i^*}$ be the set of edges incident to $i^*$ that may be processed in Step~\eqref{alg:v2} of $\vir$, conditioning on $j^* \notin \cS_K$.
We have $e^* \in \tE$ since it has the largest weight.
We consider two subcases.

\medskip

\noindent\textbf{Case 2a:} $\te \in \tE$ and $|\tE|=\td \le d$.
Applying the same reasoning as in \textbf{Case 1}, we obtain
\begin{align}
\beta_1 &\ge w(e^*) \, \kap(d,\ta) = w^* \, \kap(d,\ta), \nonumber \\
\beta_2
&= w(\te) \, (1-\ta)^2 \, \bP{\frac{1}{2}\,\ta + \cdots + \frac{1}{\td-1}(1-\ta)^{\td-3}\,\ta + \frac{1}{\td}(1-\ta)^{\td-2}} \nonumber \\
&= \tw \, \bp{\kap(\td,\ta) - \ta(1-\ta)}
\ge \tw \, \sbp{\kap(d,\ta) - \ta(1-\ta)}.
\label{ineq:5/6/a}
\end{align}
The last inequality follows since $\kap(\cdot,\ta)$ is non-increasing in its first argument for any fixed $\ta \in [0,1]$.

Therefore, the ratio of the expected gain from edges incident to $i^*$ and $\bai$ in $\bam$ (denoted by $\beta$) to the corresponding optimal weight in $M^*$ (denoted by $\gam$) satisfies
\begingroup
\allowdisplaybreaks
\begin{align}
\frac{\beta}{\gam}
&= \frac{\beta_1+\beta_2}{w(\bae)+w(\te)}
= \frac{\beta_1+\beta_2}{\baw+\tw} \nonumber \\
&\ge \frac{w^* \, \kap(d,\ta) + \tw \, \sbp{\kap(d,\ta) - \ta(1-\ta)}}{\baw+\tw} \nonumber \\
&= \frac{\kap(d,\ta) + (\tw/w^*) \, \sbp{\kap(d,\ta) - \ta(1-\ta)}}{\baw/w^* + \tw/w^*} \nonumber \\
&\ge \frac{\kap(d,\ta) + (\tw/w^*) \, \sbp{\kap(d,\ta) - \ta(1-\ta)}}{1 + \tw/w^*}
\quad \text{(since $\baw/w^* \le 1$)} \nonumber \\
&\ge \frac{2\kap(d,\ta) - \ta(1-\ta)}{2}
\quad \text{(since $\tw/w^* \le 1$)}.
\label{ineq:case2a}
\end{align}
\endgroup

\medskip

\noindent\textbf{Case 2b:} $\te \notin \tE$ and $|\tE|=\td \le d-1$.
Recall that $\tE \subseteq E_{i^*}$ is the set of edges incident to $i^*$ that may be processed in Step~\eqref{alg:v5} of $\vir$, conditioning on $j^* \notin \cS_K$.
In this case, before processing $\te=(i^*,j=2)$, some other edge $\he=(\hi,j=2)$ must be processed, with weight $w(\he)=:\hw \ge \tw$.
Let $\aw$ denote the weight contributed by $\hi$ in the optimal matching, where $\aee=(\hi,j=3) \in M^*$.

Let $\beta$ denote the total expected weight gained from edges incident to $i^*$, $\bai$, and $\hi$ in $\bam$, and let $\gam$ denote the corresponding total weight in $M^*$.
Then $\gam=\baw+\tw+\aw$, and $\beta \ge \beta_1+\beta_2+\beta_3$, where $\beta_1,\beta_2,\beta_3$ are the expected gains contributed by $e^*$, $\he$, and $\aee$, respectively.

From the previous analysis, $\beta_1 \ge w^* \, \kap(d-1,\ta)$ since $\te \notin \tE$ and hence $|\tE| \le d-1$.
Proceeding as in \textbf{Case 2a}, we have
\begingroup
\allowdisplaybreaks
\begin{align}
\frac{\beta}{\gam}
&\ge \frac{\beta_1+\beta_2+\beta_3}{\baw+\tw+\aw} \nonumber \\
&\ge \frac{w^* \, \kap(d-1,\ta) + \tw \, \kap(d,\ta) + \aw \, \sbp{\kap(d,\ta) - \ta(1-\ta)}}{\baw+\tw+\aw}
\quad \text{(since $\hw \ge \tw$)} \nonumber \\
&= \frac{\kap(d-1,\ta) + (\tw/w^*) \, \kap(d,\ta) + (\aw/w^*) \, \sbp{\kap(d,\ta) - \ta(1-\ta)}}{\baw/w^* + \tw/w^* + \aw/w^*} \nonumber \\
&\ge \frac{\kap(d-1,\ta) + \sbp{\kap(d,\ta) - \ta(1-\ta)} + (\tw/w^*) \, \kap(d,\ta)}{2 + \tw/w^*}
\quad \text{(since $\aw \le w^*$)} \nonumber \\
&\ge \min \bP{
\frac{\kap(d-1,\ta) + \sbp{\kap(d,\ta) - \ta(1-\ta)} + \kap(d,\ta)}{3},
\frac{\kap(d-1,\ta) + \sbp{\kap(d,\ta) - \ta(1-\ta)}}{2}
} \nonumber \\
&\ge \frac{\kap(d,\ta) + \sbp{\kap(d,\ta) - \ta(1-\ta)}}{2}
\quad \text{(since $\kap(d-1,\ta) \ge \kap(d,\ta)$)}.
\label{ineq:case2b}
\end{align}
\endgroup

By recursively applying the analysis above to the remaining subgraph, we establish Lemma~\ref{lem:bam-b}.
\end{proof}

\begin{proof}[\textbf{Proof of Theorem~\ref{thm:smg-cr-tig}}]
The lower bounds on $\beta/\gam$ in \eqref{ineq:case1} (\textbf{Case 1}) and \eqref{ineq:case2a} (\textbf{Case 2a}) can be tight, as illustrated in Figure~\ref{fig:1} (for $d \ge 1$) and Figure~\ref{fig:2} (for $d \ge 2$), respectively.\footnote{This is why we claim tightness only for $d \ge 2$, excluding the case $d=1$.}
In contrast, the lower bound in \eqref{ineq:case2b} (\textbf{Case 2b}) is tight only at the corner cases $\ta \in \{0,1\}$.
Therefore, for any fixed $\ta \in (0,1)$, the worst case occurs in \textbf{Case 2a}, where
\[
\frac{\beta}{\gam} = \kap(d,\ta) - \frac{\ta(1-\ta)}{2}.
\]
Summarizing the above analysis together with Lemma~\ref{lem:bam-b} completes the proof of Theorem~\ref{thm:smg-cr-tig}.
\end{proof}

\begin{figure}[ht!]
\begin{subfigure}[b]{0.45\textwidth}
\begin{tikzpicture}
 \draw (0,0) node[minimum size=0.2mm,draw,circle, thick] {$i^*$};
 \draw (4,0) node[minimum size=0.2mm,draw,circle, thick] {$j^*$};
 \draw (4,-1.5) node[minimum size=0.2mm,draw,circle, thick] {$j_2$};
 \draw [dotted, ultra thick](3.7,-3) -- (4.3,-3);
 \draw (4,-4) node[minimum size=0.2mm,draw,circle, thick] {$j_d$};

 \draw[-, ultra thick ] (0.4,0)--(3.6,0) node [blue, above, midway, sloped] {$w^*=1$};
 \draw[-, ultra thick] (0.4,0)--(3.6,-1.5) node[blue, above, midway, sloped] {$w_2=\ep$};
 \draw[-, ultra thick] (0.4,0)--(3.6,-4) node[blue, below, midway, sloped] {$w_d=\ep$};
 \draw[-] (0.4,0)--(3.7,-3);
\end{tikzpicture}
\caption{Worst-case structure for \tbf{Case 1}.}
\label{fig:1}
\end{subfigure}
\hfill
\begin{subfigure}[b]{0.45\textwidth}
\begin{tikzpicture}
 \draw (0,1.5) node[minimum size=0.2mm,draw,circle, thick] {$\bai$};
 \draw[-, ultra thick ] (0.4,1.5)--(3.6,0) node [blue, above, midway, sloped] {$\baw=1$};

 \draw (0,0) node[minimum size=0.2mm,draw,circle, thick] {$i^*$};
 \draw (4,0) node[minimum size=0.2mm,draw,circle, thick] {$j^*$};
 \draw (4,-1.5) node[minimum size=0.2mm,draw,circle, thick] {$j_2$};
 \draw [dotted, ultra thick](3.7,-3) -- (4.3,-3);
 \draw (4,-4) node[minimum size=0.2mm,draw,circle, thick] {$j_d$};

 \draw[-, ultra thick ] (0.4,0)--(3.6,0) node [blue, above, midway, sloped] {$w^*=1+\ep$};
 \draw[-, ultra thick] (0.4,0)--(3.6,-1.5) node[blue, above, midway, sloped] {$\tw=1$};
 \draw[-, ultra thick] (0.4,0)--(3.6,-4) node[blue, below, midway, sloped] {$w_d=\ep$};
 \draw[-, thick] (0.4,0)--(3.7,-3) node[blue, above, midway, sloped] {$w_3=\ep$};
\end{tikzpicture}
\caption{Worst-case structure for \tbf{Case 2a}.}
\label{fig:2}
\end{subfigure}

\caption{
Figure~\ref{fig:1} illustrates the worst-case scenario for \textbf{Case 1}, where $w^* = w(i^*,j^*) = 1$ and $w_\ell = w(i^*,j_\ell) = \ep > 0$ for all $1 < \ell \le d$.
In this instance, $e^*=(i^*,j^*)$ contributes $\gam = 1$ to $\OPT$, while contributing $\beta = \kap(d,\ta) + O(\ep)$ in expectation to $\E[w(\bam)]$.
Figure~\ref{fig:2} illustrates the worst-case scenario for \textbf{Case 2a}, where $w^* = w(i^*,j^*) = 1+\ep$, $\baw = w(\bai,j^*) = 1$, $\tw = w(i^*,j_2) = 1$, and $w_\ell = w(i^*,j_\ell) = \ep > 0$ for all $2 < \ell \le d$.
Here the edges incident to $\bai$ and $i^*$ contribute $\gam = 2$ to $\OPT$, while contributing $\beta = 2\kap(d,\ta) - \ta(1-\ta) + O(\ep)$ in expectation to $\E[w(\bam)]$.
}
\label{fig:vir}
\end{figure}
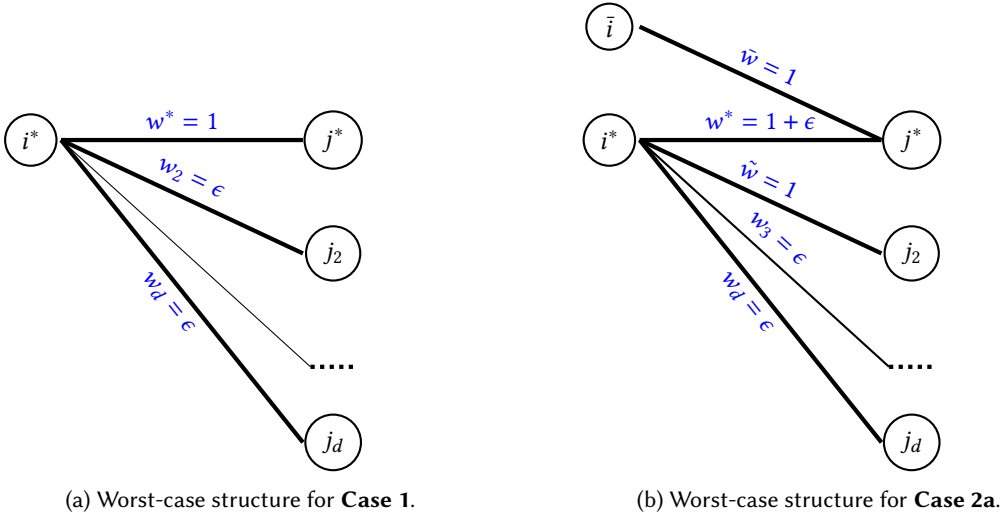


\section{Examples Showing the Failure of Negative Dependence in $\smg$}
\label{app:neg-dep}

\begin{figure}[ht!]
\centering
\begin{minipage}[b]{0.45\linewidth}
\centering
\begin{tikzpicture}[ultra thick]
    \node[circle, draw] (j1) at (1,0) {};
    \node[above] at (1,0.2) {$j_1$};
    \node[circle, draw] (j2) at (2,0) {};
    \node[above] at (2,0.2) {$j_2$};

    \draw[dotted, ultra thick] (2.8,0) -- (3.2,0);

    \node[circle, draw] (jd1) at (5,0) {};
    \node[above] at (5,0.2) {$j_{d-1}$};
    \node[circle, draw] (jd) at (7,0) {};
    \node[above] at (7,0.2) {$j_d$};

    \node[circle, draw] (ti) at (3,-2) {};
    \node[below] at (3,-2.2) {$\ti$};

    \node[circle, draw] (bi) at (7,-2) {};
    \node[below] at (7,-2.2) {$\bi$};

    \draw (j1) -- (ti);
    \draw (j2) -- (ti);
    \draw (jd1) -- (ti);
    \draw (jd) -- (ti);

    \draw (jd) -- (bi);

    \node[blue, below] at (2,-1) {$w_1$};
    \node[blue, below] at (3,-1) {$w_2$};
    \node[blue, above] at (3.8,-1) {$w_{d-1}$};
    \node[blue, below] at (5,-1.1) {$w_d$};
    \node[blue, left]  at (7,-1.2) {$w_{d+1}$};

    \node[blue, below] at (5,-2.7) {$w_1 > w_2 > \cdots > w_{d+1}$};
\end{tikzpicture}

\vspace{1mm}
{\small (a) Correlated matching events}
\end{minipage}
\hfill
\begin{minipage}[b]{0.45\linewidth}
\centering
\begin{tikzpicture}[ultra thick]
    \node[circle, draw] (i1) at (1,0) {};
    \node[above] at (1,0.2) {$i_1$};
    \node[circle, draw] (j1) at (4,0) {};
    \node[above] at (4,0.2) {$j_1$};
    \draw (i1) -- node[blue, above, midway] {$w_1$} (j1);
    \draw (i1) -- node[blue, above, near end] {$w_2$} (3.9,-1.4);
    \draw (i1) -- node[blue, above, near end] {$w_3$} (3.9,-2.9);

    \node[circle, draw] (i2) at (1,-1.5) {};
    \node[above] at (1,-1.3) {$i_2$};
    \node[circle, draw] (j2) at (4,-1.5) {};
    \node[above] at (4,-1.3) {$j_2$};
    \draw (i2) -- node[blue, above, near start] {$w_4$} (j2);

    \node[circle, draw] (i3) at (1,-3) {};
    \node[above] at (1,-2.8) {$i_3$};
    \node[circle, draw] (j3) at (4,-3) {};
    \node[above] at (4,-2.8) {$j_3$};
    \draw (i3) -- node[blue, above, near start] {$w_5$} (j3);

    \node[blue, below] at (2.5,-3.7) {$w_1 > w_2 > w_3 > \max\{w_4,w_5\}$};
\end{tikzpicture}

\vspace{1mm}
{\small (b) Positive correlation among matches}
\end{minipage}

\caption{
Examples illustrating the challenges in bounding the variance of the total number of matched offline agents in $\widehat{\cM}$ produced by $\vir$.
\textbf{(a)} Two offline agents $\ti$ and $\bi$ share overlapping online neighborhoods with strictly ordered edge weights.
For suitable choices of $(d,\ta)$, the variance $\Var[X_{\bi}]$ can attain its maximum value $1/4$, and the joint variance $\Var[X_{\ti}+X_{\bi}]$ can strictly exceed the variance obtained when $X_{\ti}$ and $X_{\bi}$ are treated as independent Bernoulli random variables with mean $\bta$; see Lemma~\ref{lem:var-1}.
\textbf{(b)} Three offline agents $i_1,i_2,i_3$ exhibit positive correlation:
letting $X_\ell$ indicate whether $i_\ell$ is matched for $\ell \in \{1,2,3\}$, we have
$
\E[X_2 X_3] - \E[X_2]\E[X_3]
= \ta\,\bta^4 > 0$ for all $\ta \in (0,1)$, demonstrating that matching events need not be negatively correlated; see Lemma~\ref{lem:var-2}.
}
\label{fig:var}
\end{figure}
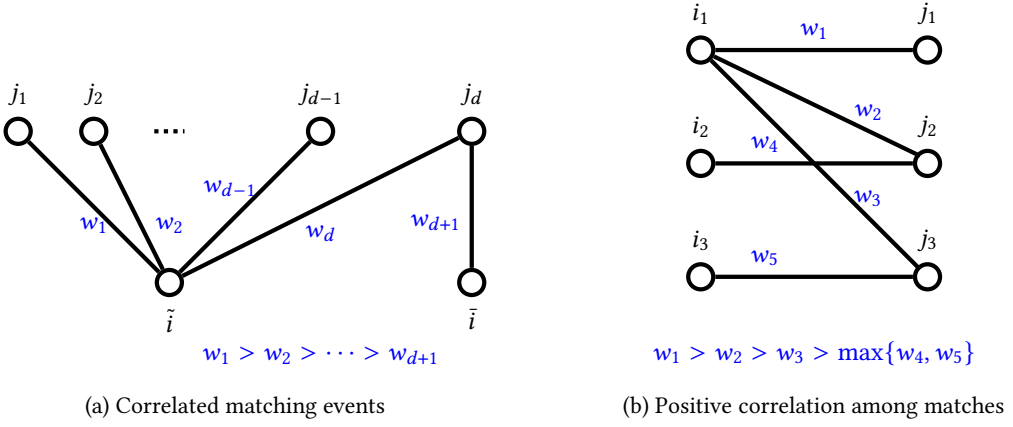

\begin{lemma}[Appendix~\ref{app:var-1}]
\label{lem:var-1}
Consider the graph shown in Figure~\ref{fig:var}(a).
Let $X_{\bi}=1$ and $X_{\ti}=1$ indicate that $\bi$ and $\ti$ are matched in $\widehat{\cM}$, respectively.
Then the following statements hold.
\begin{enumerate}
\item $\Var[X_{\bi}] = 1/4$ when
\[
d = \log_{\bta}\!\Bigl(1-\frac{1}{2\bta}\Bigr) + 1,
\qquad \text{with } \bta = 1-\ta > \tfrac12 .
\]
\item $\Var[X_{\bi}+X_{\ti}] > 2\,\bta(1-\bta)$ when
\[
d > \log_{\bta}(4\bta-3) - 1,
\qquad \text{with } \bta = 1-\ta > \tfrac34 .
\]
\end{enumerate}
\end{lemma}

\begin{lemma}[Appendix~\ref{app:var-2}]
\label{lem:var-2}
Consider the graph shown in Figure~\ref{fig:var}(b).
For $\ell \in \{1,2,3\}$, let $X_\ell=1$ indicate that the offline agent $i_\ell$ is matched by $\vir$.
Then, for any $\ta \in (0,1)$,
\[
\E[X_2 X_3] - \E[X_2]\E[X_3] = \ta\,\bta^4 > 0,
\]
where $\bta=1-\ta$.
\end{lemma}

Taken together, Lemmas~\ref{lem:var-1} and~\ref{lem:var-2} illustrate two distinct obstacles in variance analysis, as visualized in Figure~\ref{fig:var}.
Lemma~\ref{lem:var-1} shows that even the variance of a \emph{single} matching indicator can reach its maximum value under suitable structural conditions, and that the joint variance of two such indicators may exceed what would be expected under independence.
Lemma~\ref{lem:var-2} further demonstrates that matching events for different offline agents can be \emph{positively correlated}, ruling out standard concentration arguments based on negative dependence.
These examples highlight that variance control for $\smg$ cannot rely on naive independence assumptions and must explicitly account for structural correlations induced by the matching process.

\subsection{Proof of Lemma~\ref{lem:var-1}} \label{app:var-1}
We split the proof into two parts, corresponding to the first and second claims.
\begin{proof}[Proof of the first claim]
Let $e_j=(\ti,j)$ for $1 \le j \le d$, and let $e_{d+1}=(\bi,d)$.
Under $\vir$, the edges $\{e_j\}_{j=1}^d$ are processed sequentially in decreasing order of weight.
We focus on the event that $e_{d+1}$ is processed.

\smallskip
\noindent\textbf{Case 1.}
$j \notin \cS_K$ for all $1 \le j \le d-1$.
Then all edges $e_j$ with $1 \le j \le d-1$ are added to $\cM^b$, and $e_d$ is processed next.
Consequently, $e_{d+1}$ is never considered and $X_{\bi}=0$.

\smallskip
\noindent\textbf{Case 2.}
There exists some $1 \le j \le d-1$ such that $j \in \cS_K$.
In this case, $e_{d+1}$ is eventually processed, and $\bi$ is matched with probability $\bta=1-\ta$.
Thus $X_{\bi}\sim \mathrm{Ber}(\bta)$.

\smallskip
Combining the two cases, we obtain
\[
\E[X_{\bi}] = (1-\bta^{\,d-1})\,\bta.~\footnote{As before, we assume $d=o(n)$ and may pad the instance with $n-d$ dummy online nodes so that the equality holds exactly as $n\to\infty$.}
\]
To maximize the variance of $X_{\bi}$, we set $\E[X_{\bi}]=1/2$, which yields
\[
(1-\bta^{\,d-1})\,\bta = \tfrac12 .
\]
Solving this equation gives
$d = \log_{\bta}(1-1/(2\bta)) + 1$, and hence $\Var[X_{\bi}]=1/4$.
\end{proof}


\begin{proof}[Proof of the second claim]
Let $e_j=(\ti,j)$ for $1 \le j \le d$, and let $e_{d+1}=(\bi,d)$.
We distinguish the following cases.

\smallskip
\noindent\textbf{Case 1.}
$j=1 \in \cS_K$.
Then $e_1$ is added to $\cM^a$, $e_{d+1}$ is processed, $X_{\ti}=0$, and
$X_{\bi}\sim \mathrm{Ber}(\bta)$.

\smallskip
\noindent\textbf{Case 2.}
$j=1 \notin \cS_K$, and there exists some $2 \le j \le d-1$ with $j \in \cS_K$.
Then $X_{\ti}=1$ and $X_{\bi}\sim \mathrm{Ber}(\bta)$.

\smallskip
\noindent\textbf{Case 3.}
$j \notin \cS_K$ for all $1 \le j \le d-1$.
Then $X_{\ti}=1$ and $X_{\bi}=0$.

\smallskip
Let $Y := X_{\ti}+X_{\bi}$.
From the above cases, $Y$ takes values $0,1,2$ with respective probabilities
\[
(1-\bta)^2,\qquad
\bta(1-\bta)(2-\bta^{\,d-1})+\bta^{\,d},\qquad
\bta^2(1-\bta^{\,d-1}).
\]
A direct computation yields
\[
\Var[Y] - 2\,\bta(1-\bta)
= \bta^{\,d+1}\bigl(4\bta-3-\bta^{\,d+1}\bigr).
\]
Solving $\Var[Y] > 2\,\bta(1-\bta)$ gives
\[
d > \log_{\bta}(4\bta-3) - 1.
\]
For any $\bta \in (3/4,1)$, such a value of $d \ge 1$ exists, completing the proof.
\end{proof}


\subsection{Proof of Lemma~\ref{lem:var-2}} \label{app:var-2}
\begin{proof}
For $\ell \in \{1,2,3\}$, let $\chi_\ell=1$ indicate that the online agent $j_\ell$ does not belong to the sample set $\cS_K$.
As in the previous analysis, we treat $\{\chi_\ell\}_{\ell=1}^3$ as independent Bernoulli random variables with $\E[\chi_\ell]=\bta=1-\ta$.

The matching events for the three offline agents can be expressed as follows.
We have $X_1=\chi_1$, since $i_1$ is matched if and only if $j_1 \notin \cS_K$.
Moreover,
\[
X_2 = (1-\chi_1)\,\chi_2,
\]
because $i_2$ is matched only if $j_1 \in \cS_K$ and $j_2 \notin \cS_K$.
Similarly,
\[
X_3 = (1-\chi_1\chi_2)\,\chi_3,
\]
since $i_3$ is matched unless both $j_1$ and $j_2$ are excluded from $\cS_K$, and provided that $j_3 \notin \cS_K$.

It follows that
\[
\E[X_2 X_3]
= \Pr[\chi_1=0,\chi_2=1,\chi_3=1]
= \ta\,\bta^2,
\]
while
\[
\E[X_2]\E[X_3]
= (\ta\,\bta)\,\bigl(1-\bta^2\bigr)\,\bta.
\]
Therefore,
\[
\E[X_2 X_3] - \E[X_2]\E[X_3]
= \ta\,\bta^4,
\]
which is strictly positive for all $\ta \in (0,1)$.
\end{proof}

\section{Proof of Lemma~\ref{lem:var-5}}
\label{app:lem-var5}

For each offline agent $i \in I$, define $X_i := Y_i \cdot Z_i$, where $Y_i=1$ indicates that there exists at least one edge incident to $i$ that passes the condition in Step~\eqref{alg:v2} of $\vir$, and $Z_i \sim \mathrm{Ber}(\bta)$ indicates that $j_i \notin \cS_K$, with $j_i$ denoting the first online neighbor of $i$ processed in Step~\eqref{alg:v5}.\footnote{For notational convenience, we assume that $Y_i$ and $Z_i$ are independent. In principle, these two random variables may be weakly correlated—for instance, when $Y_i=1$ implies the existence of a constant-size subset $S \subseteq J$ such that $j \in \cS_K$ for all $j \in S$. However, such correlations vanish asymptotically as $n \to \infty$, and can be safely ignored for our purposes.}

We call an offline agent $i$ \emph{free} if $\E[Y_i]=1$, and \emph{non-free} otherwise.
Equivalently, a free offline agent is guaranteed to have at least one incident edge processed in Step~\eqref{alg:v5}, and therefore satisfies $X_i \sim \mathrm{Ber}(\bta)$ with $\bta=1-\ta$.

\smallskip

For two distinct offline agents $\ti \neq \bi$, we say that $\ti$ \emph{negatively impacts} $\bi$ if the event $Z_{\ti}=1$ strictly decreases the value of $\E[Y_{\bi}]$, and \emph{positively impacts} $\bi$ if it strictly increases $\E[Y_{\bi}]$.
Intuitively, a \emph{negative} impact means that the outcome $j_{\ti} \notin \cS_K$—where $j_{\ti}$ is the first online neighbor of $\ti$ processed in Step~\eqref{alg:v5}—can cause $\bi$ to lose one or more incident edges from being considered by $\vir$, whereas a \emph{positive} impact has the opposite effect.
Note that any offline agent that can be either negatively or positively impacted by another must be non-free by definition.

\begin{lemma}
\label{lem:var-4}
Any free or non-free offline agent can only negatively impact another non-free agent.
\end{lemma}

Lemma~\ref{lem:var-4} implies that the matching events of all free offline agents can be treated as independent Bernoulli random variables with mean $\bta$.

\begin{figure}[ht!]
\centering
\begin{tikzpicture}
    \draw (0,0) node[minimum size=0.2mm,draw,circle, thick] {\redd{$\ti$}};
    \draw (0,-1.5) node[minimum size=0.2mm,draw,circle, thick] {\bluee{$\hi$}};
    \draw (0,-3) node[minimum size=0.2mm,draw,circle, thick] {\bluee{$\bi$}};

    \draw (4,0) node[minimum size=0.2mm,draw,circle, thick] {$j_1$};
    \draw (4,-1.5) node[minimum size=0.2mm,draw,circle, thick] {$j_2$};
    \draw (4,-4) node[minimum size=0.2mm,draw,circle, thick] {$j_3$};

    \draw[-, ultra thick] (0.4,0)--(3.6,0)
        node [blue, above, near end, sloped] {$e_1,\; w_1$};
    \draw[-, ultra thick] (0.4,0)--(3.6,-1.5)
        node[blue, above, near end, sloped] {$e_2,\; w_2$};

    \draw[-, ultra thick] (0.4,-1.5)--(3.6,-1.5)
        node [blue, above, midway, sloped] {$e_5,\; w_5$};

    \draw[-, ultra thick] (0.4,-3)--(3.6,-1.5)
        node [blue, above, midway, sloped] {$e_3,\; w_3$};
    \draw[-, ultra thick] (0.4,-3)--(3.6,-4)
        node [blue, above, midway, sloped] {$e_4,\; w_4$};

    \draw[-, very thick] (0.4,0)--(2,0.5);
    \draw[-, very thick] (0.4,0)--(2,0.7);
    \draw[-, very thick] (0.4,0)--(2,0.9);
\end{tikzpicture}
\caption{
An illustrative example for the proof of Lemma~\ref{lem:var-4}.
A free offline agent $\ti$ negatively impacts a non-free offline agent $\hi$ under the weight ordering
$w_1 > w_2 > w_5$ and $w_1 > w_2 > w_3 > w_4$.
In this configuration, the offline agent $\bi$ becomes non-free only when either the edge $e_4$ does not exist or $e_4$ itself is negatively impacted by another agent.
In either case, $\bi$ cannot exert any influence on any other offline agent.
}
\label{fig:var-3}
\end{figure}
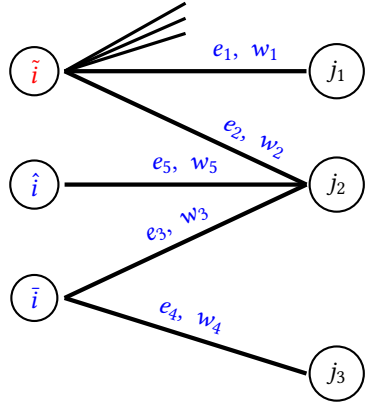

\begin{proof}[\tbf{Proof of Lemma~\ref{lem:var-4}}]
By definition, a free offline agent is one that is guaranteed to reach a stage in $\vir$ at which it has the largest incident edge (by weight), and hence this edge will be processed before any competing edges.
Consider a representative free offline agent $\ti$ as illustrated in Figure~\ref{fig:var-3}, and let $e_1=(\ti,j_1)$ denote the largest-weight edge incident to $\ti$.
Assume that $e_2$ is not the largest edge overall; otherwise, $\ti$ would have no influence on any other offline agent.

\xhdr{Impact of $\ti$ on an offline agent $\bi$ that shares at least one online neighbor with $\ti$.}
We distinguish the following cases.

\tbf{Case 1.}
The second-largest edge is $e=(\hi,j_2)$ with $\hi \neq \ti$ and $\hi \neq \bi$.
In this case, processing $e$ removes all other edges incident to $j_2$, and therefore the event $Z_{\ti}$ has no effect on $Y_{\bi}$.

\tbf{Case 2.}
The second-largest edge is $e=(\bi,j)$ for some $j \in \{j_2,j_3,\ldots\}$.
As in Case~1, processing $e$ removes all other edges incident to $j$, and hence $Z_{\ti}$ again has no influence on $Y_{\bi}$.

\tbf{Case 3.}
The second-largest edge is incident to $\ti$.
Without loss of generality, assume that $e_2=(\ti,j_2)$ is the second-largest edge (otherwise, the above analysis applies recursively).
If $Z_{\ti}=1$, then $e_2$ is processed next, which removes all other edges incident to $j_2$, including $e_3=(\bi,j_2)$.
If $Z_{\ti}=0$, then all edges incident to $\ti$ are removed, leaving $e_3$ intact.

Thus, the event $Z_{\ti}=0$ preserves one additional edge for $\bi$ relative to $Z_{\ti}=1$, and can only weakly increase the value of $Y_{\bi}$.
In particular, $Z_{\ti}=1$ cannot increase $Y_{\bi}$.

\xhdr{Impact of $\ti$ on an offline agent that does not share any online neighbor with $\ti$.}
From the discussion above, Case~3 is the only scenario in which $Z_{\ti}$ can affect $Y_{\bi}$.
We now examine whether such an effect can propagate to a third offline agent $i \notin \{\ti,\bi\}$ through $\bi$.

If $Z_{\ti}$ affects $Y_{\bi}$ only when $e_3$ is the sole remaining edge incident to $\bi$, then both outcomes $Z_{\ti}=1$ (which removes $e_3$) and $Z_{\ti}=0$ (which leaves $e_3$ as the only remaining edge) ultimately result in all edges incident to $j_2$ being removed.
Hence, no influence can propagate further to any other offline agent.

If $Z_{\ti}$ does not change the value of $Y_{\bi}$, then $Z_{\ti}=1$ can only weakly reduce the value of $Y_i$ for $i \neq \ti,\bi$.
Let $e^*_{\bi}$ denote the first edge incident to $\bi$ that passes Step~\eqref{alg:v2} (or $\emptyset$ if no such edge exists).
Any change in $e^*_{\bi}$ caused by $Z_{\ti}$ can only correspond to switching from one non-empty edge to another (for example, from $e_3$ to $e_4$), which again can only reduce the set of edges available to other offline agents.

Therefore, $Z_{\ti}=1$ can only decrease, and never increase, the value of $Y_i$ for any $i \notin \{\ti,\bi\}$.
This establishes that any free or non-free offline agent can only negatively impact another non-free agent, completing the proof.
\end{proof}

\begin{figure}[ht!]
\centering
\begin{tikzpicture}
    \draw (0,0) node[minimum size=0.2mm,draw,circle, thick] {\redd{$i_0$}};
    \draw (0,-1.5) node[minimum size=0.2mm,draw,circle, thick] {\bluee{$i_1$}};
    \draw (0,-3) node[minimum size=0.2mm,draw,circle, thick] {\bluee{$i_2$}};
    \draw (0,-5) node[minimum size=0.2mm,draw,circle, thick] {\bluee{$i_k$}};

    \draw (4,0) node[minimum size=0.2mm,draw,circle, thick] {$j_0$};
    \draw (4,-1.5) node[minimum size=0.2mm,draw,circle, thick] {$j_1$};
    \draw (4,-3) node[minimum size=0.2mm,draw,circle, thick] {$j_2$};
    \draw (4,-5) node[minimum size=0.2mm,draw,circle, thick] {$j_k$};

    \draw[-, ultra thick] (0.4,0)--(3.6,0) node [blue, above, midway, sloped] {$w_0$};
    \draw[-, ultra thick] (0.4,0)--(3.6,-1.5) node [blue, above, midway, sloped] {$w_1$};
    \draw[-, ultra thick] (0.4,0)--(3.6,-3) node [blue, above, near end, sloped] {$w_2$};
    \draw[-, ultra thick] (0.4,0)--(3.6,-5) node [blue, above, near end, sloped] {$w_k$};

    \draw(1,-1.5) node [blue, above] {$v_1$};
    \draw[-, ultra thick] (0.4,-1.5)--(3.6,-1.5);
    \draw[-, ultra thick] (0.4,-3)--(3.6,-3) node [blue, above, near start, sloped] {$v_2$};
    \draw[-, ultra thick] (0.4,-5)--(3.6,-5) node [blue, above, near start, sloped] {$v_k$};

    \draw[dotted, ultra thick] (4,-3.5)--(4,-4.5);
    \draw[dotted, ultra thick] (0,-3.5)--(0,-4.5);

    \draw (2,-6) node[below]
        {$w_0>w_1>\cdots>w_k>\max\{v_1,v_2,\ldots,v_k\}$};
\end{tikzpicture}
\caption{
An example illustrating a configuration in which the matches of non-free offline agents
$(i_1,\ldots,i_k)$ in $\vir$ can be pairwise \emph{positively} correlated.
Such correlations arise only when all non-free agents are negatively impacted by a single free agent $i_0$.
The covariance of the total number of matches within this cluster is upper bounded by
$\min\{\bta^4/\ta^2,\,(k^2+2k)\,\bta^{7/2}\}$.
}
\label{fig:var-7}
\end{figure}
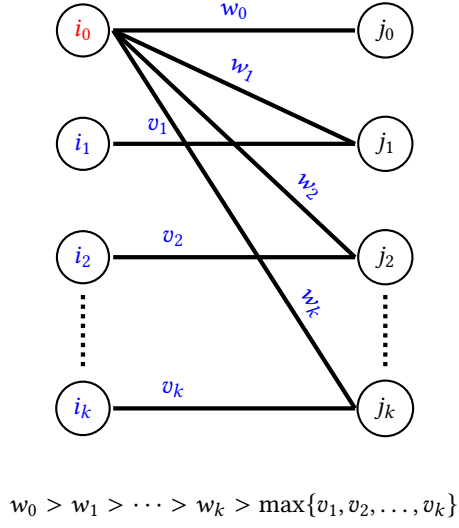

\begin{proof}[\tbf{Proof of Lemma~\ref{lem:var-5}}]
From the proof of Lemma~\ref{lem:var-4}, pairwise positive correlation can arise only between non-free offline agents that are all negatively impacted by the same free offline agent.
Consider such a cluster as illustrated in Figure~\ref{fig:var-7}.
Because $w_0>w_1>\cdots>w_k>\max\{v_1,v_2,\ldots,v_k\}$, the offline agent $i_0$ is free, while $i_1,\ldots,i_k$ are all non-free.

For each $0\le \ell\le k$, let $X_{i_\ell}=1$ indicate that $i_\ell$ is matched in $\vir$, and define
$\mu_\ell:=\E[X_{i_\ell}]$.
A direct computation shows that
\begingroup
\allowdisplaybreaks
\begin{align*}
\mu_0 &= \bta, \\
\mu_\ell &= \bta\,(1-\bta^\ell), && 1\le \ell\le k, \\
\mu_{0,\ell} &= \E[X_{i_0}X_{i_\ell}]
              = \bta^2\,(1-\bta^{\ell-1}), && 1\le \ell\le k, \\
\mu_{\ell,\ell'} &= \E[X_{i_\ell}X_{i_{\ell'}}]
              = \bta^2\,(1-\bta^{\ell}), && 1\le \ell<\ell'\le k .
\end{align*}
\endgroup

Define the total covariance within this cluster by
\[
\sg(k)
:= \sum_{0\le \ell<\ell'\le k}
   \bigl(\E[X_{i_\ell}X_{i_{\ell'}}]-\E[X_{i_\ell}]\E[X_{i_{\ell'}}]\bigr).
\]
Substituting the expressions above yields
\begin{align*}
\sg(k)
&= \Bigl(\sum_{\ell=1}^k \mu_{0,\ell}
    + \sum_{1\le \ell<\ell'\le k} \mu_{\ell,\ell'}\Bigr)
  - \Bigl(\mu_0\sum_{\ell=1}^k \mu_\ell
    + \sum_{1\le \ell<\ell'\le k} \mu_\ell\mu_{\ell'}\Bigr) \\
&= \bta^2(\bta^k-1)
  + \frac{\bta^4}{\ta}
    \Bigl(\frac{1-\bta^{2k}}{1-\bta^2}-k\,\bta^{k-1}\Bigr) \\
&\le \frac{\bta^4}{\ta}\cdot\frac{1}{1-\bta^2}
 = \frac{\bta^4}{\ta^2(1+\bta)}
 \le \frac{\bta^4}{\ta^2}.
\end{align*}

We next derive an alternative bound.
Using $\ln(1-x)\ge -x/\sqrt{1-x}$ for $x\in[0,1)$, we obtain
\begingroup
\allowdisplaybreaks
\begin{align*} \sg(k) & = \bta^2(\bta^k-1)
  + \frac{\bta^4}{\ta}
    \Bigl(\frac{1-\bta^{2k}}{1-\bta^2}-k\,\bta^{k-1}\Bigr) \\ & \le \frac{\bta^4}{\ta} \, \bp{\frac{1 - \bta^{2k}}{1 - \bta^2} - k \, \bta^{k-1}} \\ & = \frac{\bta^4}{\ta} \, \bp{\frac{1 - \sfe^{2k \, \ln (1-\ta)}}{\ta \, (2 - \ta)} - k \, \sfe^{(k-1) \, \ln(1-\ta)}} \\ & \le \frac{\bta^4}{\ta} \, \bp{\frac{2k \ta / \sqrt{1 - \ta}}{\ta (2 - \ta)} - k \, \bp{1 - \frac{(k-1) \ta}{\sqrt{1 - \ta}}}} \quad \text{\Big(since $\ln(1-x) \ge -x / \sqrt{1-x}$ for $x \in [0,1)$\Big)} \\ & = \frac{\bta^4}{\ta} \, \bp{\frac{k / \sqrt{1 - \ta}}{(1 - \ta/2)} - k + \frac{k(k-1) \ta}{\sqrt{1 - \ta}}} \\ & = \frac{\bta^4}{\ta} \, \bp{k \, \frac{\ta / (1 + \sqrt{1 - \ta}) + (\ta / 2) \, \sqrt{1 - \ta}}{(1 - \ta/2) \, \sqrt{1 - \ta}} + \frac{k(k-1) \ta}{\sqrt{1 - \ta}}} \\ & \le (1 - \ta)^4 \, \bp{k \, \frac{3}{\sqrt{1 - \ta}} + \frac{k(k-1)}{\sqrt{1 - \ta}}} = (1 - \ta)^4 \, \frac{k^2 + 2k}{\sqrt{1 - \ta}} = (k^2 + 2k) \, \bta^{7/2}. \end{align*}
\endgroup

Therefore, for any cluster of $k$ non-free agents negatively impacted by a single free agent, the total covariance is at most
\[
\min\{\bta^4/\ta^2,\,(k^2+2k)\,\bta^{7/2}\}.
\]

Since there are at most $m/(k+1)$ such disjoint clusters with $k+1\le d$, and different clusters are independent, the total covariance across all offline agents is upper bounded by
\[
m\, \min\{\bta^4/\ta^2,\,(d+2)\,\bta^{7/2}\}.
\]
This completes the proof.
\end{proof}


\section{Proof of  the Lower Bound on $\psi$ in Theorem~\ref{thm:smg-var-new}}
\label{app:smg-var-tig}

\begin{proof}
Consider an instance $G$ consisting of $m$ identical star graphs, where each star comprises one offline agent and $d$ online agents.
Thus $|I|=m$ and $|J|=n=m\,d$.
For each $i \in I$, let $X_i=1$ indicate that $i$ is matched in $\vir$, and define
\[
X := \sum_{i \in I} X_i .
\]

For each offline agent $i$, let $\ell_i \in J$ denote the online agent incident to $i$ with the largest edge weight.
Recall that $\cS_K$ is the random sample used by $\vir$.
By construction of $\vir$, we have
\[
X_i = 1 \quad \text{if and only if} \quad \ell_i \notin \cS_K .
\]
Consequently, $X_i \sim \mathrm{Ber}(\bta)$ with $\bta := 1-\ta$, and hence
\[
\Var[X_i] = \ta\,\bta .
\]

Next, for any two distinct offline agents $i,i' \in I$, we have
\[
\E[X_i X_{i'}]
= \Pr\bigl(\ell_i \notin \cS_K \;\wedge\; \ell_{i'} \notin \cS_K\bigr)
= \frac{n-K}{n} \, \frac{n-K-1}{n-1}
= \bta \, \frac{\bta - 1/n}{1 - 1/n}.
\]

We now compute the variance of $X$:
\begingroup
\allowdisplaybreaks
\begin{align*}
\Var[X]
&= \sum_{i \in I} \Var[X_i]
   + 2 \sum_{1 \le i < i' \le m}
     \bigl(\E[X_i X_{i'}] - \bta^2\bigr) \\
&= m\,\ta\,\bta
   + m(m-1)
     \Bigl(\bta \, \frac{\bta - 1/n}{1 - 1/n} - \bta^2\Bigr) = m\,\ta\,\bta
   - m(m-1)\,\bta \, \frac{\ta}{n-1} \\
&= m\,\ta\,\bta
   \Bigl(1 - \frac{m-1}{n-1}\Bigr) = m\,\ta\,\bta
   \Bigl(1 - \frac{1}{d} + O(1/n)\Bigr),
\end{align*}
\endgroup
which completes the proof.
\end{proof}

\section{Structural Properties of $\psi_{m,d}(\theta)$}
\label{app:ws-va}

\begin{lemma}
\label{lem:ws-va}
\begin{enumerate}
\item \textbf{Asymptotic tightness near full sampling.}
For every fixed $d\ge2$, in the fixed-$d$ asymptotic regime $n\to\infty$,
\[
\psi_{m,d}(\theta)
=
\Omega\bigl(\overline{\psi}_{m,d}(\theta)\bigr)
\qquad
\text{as }\theta\to1_{-}.
\]
Together with Theorem~\ref{thm:smg-var-new}, this implies
\[
\psi_{m,d}(\theta)
=
\Theta\bigl(\overline{\psi}_{m,d}(\theta)\bigr)
\qquad
\text{as }\theta\to1_{-}.
\]

\item \textbf{One-sided continuity at both endpoints.}
For every fixed $m$ and $d$,
\[
\lim_{\theta\to0_{+}}\psi_{m,d}(\theta)
=
\psi_{m,d}(0)
=
0,
\]
and
\[
\lim_{\theta\to1_{-}}\psi_{m,d}(\theta)
=
\psi_{m,d}(1)
=
0.
\]
More specifically, for $\theta\in[0,1]$,
\[
\psi_{m,d}(\theta)
\le
m^3d\,\theta.
\]
\end{enumerate}
\end{lemma}

\begin{proof}
We first establish the asymptotic tightness claim near $\theta=1$.

\smallskip
\noindent
\textbf{Case $\theta=1-\epsilon$.}
By Theorem~\ref{thm:smg-var-new}, in the fixed-$d$ asymptotic regime
$n\to\infty$,
\[
\psi_{m,d}(1-\epsilon)
\ge
m\,\epsilon(1-\epsilon)
\left(
1-\frac{1}{d}+o_n(1)
\right).
\]
For every fixed $d\ge2$, the coefficient $1-1/d$ is bounded below by
$1/2$. Hence, for sufficiently large $n$,
\[
\psi_{m,d}(1-\epsilon)
=
\Omega\bigl(m\epsilon(1-\epsilon)\bigr),
\]
where the implicit constant can be chosen independently of $m$, $d$,
and $\epsilon$ for $d\ge2$.

On the other hand, recalling that
$\widetilde{\theta}=\max\{1/2,\theta\}$ and
$\overline{\theta}=1-\theta$, for sufficiently small $\epsilon>0$ we
have
\begin{align*}
\overline{\psi}_{m,d}(1-\epsilon)
&=
m\left(
\epsilon(1-\epsilon)
+
2\min\left\{
\frac{\epsilon^4}{(1-\epsilon)^2},
(d+2)\epsilon^{7/2}
\right\}
\right)\\
&\le
m\left(
\epsilon(1-\epsilon)
+
\frac{2\epsilon^4}{(1-\epsilon)^2}
\right)=
m\bigl(\epsilon(1-\epsilon)+O(\epsilon^4)\bigr).
\end{align*}
Therefore,
\[
\psi_{m,d}(\theta)
=
\Omega\bigl(\overline{\psi}_{m,d}(\theta)\bigr)
\qquad
\text{as }\theta\to1_{-}.
\]
Since Theorem~\ref{thm:smg-var-new} already gives
\[
\psi_{m,d}(\theta)
\le
\overline{\psi}_{m,d}(\theta),
\]
the upper bound is order-tight in this regime.

\smallskip
We next establish the one-sided continuity statements.

\smallskip
\noindent
\textbf{Continuity at $\theta=1$.}
Let $\theta=1-\epsilon$. From the upper bound above,
\[
0
\le
\psi_{m,d}(1-\epsilon)
\le
m\left(
\epsilon(1-\epsilon)
+
\frac{2\epsilon^4}{(1-\epsilon)^2}
\right),
\]
which converges to zero as $\epsilon\to0_{+}$.
Moreover, $\smg(1)$ samples all online agents and therefore always
returns an empty matching. Hence
\[
\psi_{m,d}(1)=0,
\]
and consequently
\[
\lim_{\theta\to1_{-}}\psi_{m,d}(\theta)
=
\psi_{m,d}(1)
=
0.
\]

\smallskip
\noindent
\textbf{Continuity at $\theta=0$.}
Fix an arbitrary instance
$G=(I,J,E)\in\mathcal{I}(m,d)$, where
$\mathcal{I}(m,d)$ denotes the class of bipartite instances with
$m$ offline agents and maximum offline degree at most $d$. 
Let
\[
J^{+}
:=
\{j\in J:E_j\neq\emptyset\}
\]
denote the set of online agents incident to at least one edge.
Since every offline agent has degree at most $d$,
\[
|J^{+}|
\le
|E|
\le
md.
\]

Let $\mathcal{E}$ be the event that none of the agents in $J^{+}$
belongs to the sample $\cS_K$.
For every $j\in J$,
\[
\Pr[j\in\cS_K]
=
\frac{K}{n}
\le
\theta.
\]
Therefore, by the union bound,
\[
\Pr[\mathcal{E}^{c}]
\le
|J^{+}|\,\frac{K}{n}
\le
md\,\theta.
\]

Conditioned on $\mathcal{E}$, the sampled graph contains no edge, so
all offline prices produced by $\smg(\theta)$ are zero.
Under the fixed tie-breaking convention of $\smg$, each
$j\in J^{+}$ therefore selects a fixed maximum-weight incident edge,
independently of the arrival order.
Consequently, the set of offline agents that are eventually matched is
deterministic on $\mathcal{E}$: an offline agent is matched if and only
if it is the selected endpoint of at least one agent in $J^{+}$.

Let this deterministic match count be $x_G$, and let
\[
X_G:=|\smg(\theta,G)|.
\]
Since $0\le X_G\le m$ and $X_G=x_G$ on $\mathcal{E}$,
\begin{align*}
\Var[X_G]
&\le
\E\bigl[(X_G-x_G)^2\bigr]\le
m^2\,\Pr[\mathcal{E}^{c}]\le
m^3d\,\theta.
\end{align*}
Because this bound holds for every
$G\in\mathcal{I}(m,d)$, taking the supremum over all such instances
gives
\[
\psi_{m,d}(\theta)
\le
m^3d\,\theta.
\]
Therefore,
\[
\lim_{\theta\to0_{+}}\psi_{m,d}(\theta)=0.
\]

Finally, when $\theta=0$, the sample is empty and all prices are zero.
By the same argument, the set of matched offline agents is deterministic,
although the identities of their matched online partners may depend on
the arrival order. Hence
\[
\psi_{m,d}(0)=0.
\]
We conclude that
\[
\lim_{\theta\to0_{+}}\psi_{m,d}(\theta)
=
\psi_{m,d}(0)
=
0,
\]
which completes the proof.
\end{proof}

\section{Arrival-Order Independence of the Black-Box Solver}
\label{sec:amend}

This appendix illustrates why the black-box convention in
Section~\ref{sec:sda} is substantive. Even if $\BB$ returns an exact
optimal matching on every prefix, the online competitiveness can approach
zero when its choice among optimal matchings depends on the arrival
order. The example below applies both to the original-weight
prefix-optimization framework of~\citet{kess13} and to our reweighted
framework. These frameworks are not identical in general; on the instance
below, however, every matching-phase reweighting multiplies all edge
weights by the same positive factor.

\begin{property}[Conditional arrival-order independence]
\label{pro:bb}
For every $t>K$ and every set $S\subseteq J$ with $|S|=t$, conditioned on
$\cS_t=S$, the matching $\cM_t$ returned by $\BB$ on the prescribed
weighted prefix graph is independent of the ordered prefix
$(j_1,\ldots,j_t)$.
\end{property}

For a deterministic solver, Property~\ref{pro:bb} means that permuting
the arrival order within a fixed labeled prefix does not change the
returned matching. For a randomized solver, the conditional distribution
of the returned matching must be the same for every such permutation.
Section~\ref{sec:sda} enforces this property through its black-box
convention, which also specifies the independence of the random seeds
across calls.

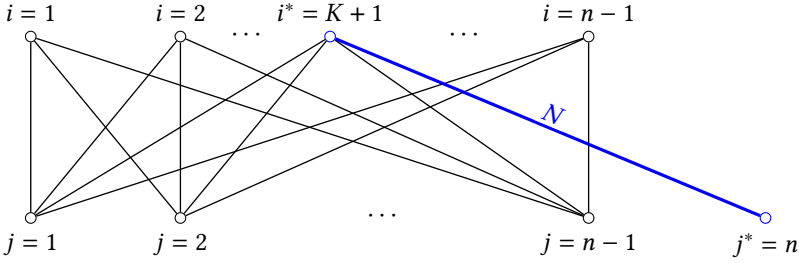
\begin{figure}[ht!]
\centering
\begin{tikzpicture}[
    x=0.9cm, y=1cm,
    vertex/.style={
        circle,draw,fill=white,minimum size=4pt,inner sep=0pt
    },
    ordinary edge/.style={line width=0.5pt}
]
    \node[vertex,label=above:{$i=1$}]
        (i1) at (0,0) {};
    \node[vertex,label=above:{$i=2$}]
        (i2) at (2.2,0) {};
    \node[vertex,draw=blue,label=above:{$i^*=K+1$}]
        (is) at (4.4,0) {};
    \node[vertex,label=above:{$i=n-1$}]
        (in) at (8.2,0) {};
    \node at (3.2,0) {$\cdots$};
    \node at (6.4,0) {$\cdots$};

    \node[vertex,label=below:{$j=1$}]
        (j1) at (0,-2.4) {};
    \node[vertex,label=below:{$j=2$}]
        (j2) at (2.2,-2.4) {};
    \node[vertex,label=below:{$j=n-1$}]
        (jn) at (8.2,-2.4) {};
    \node[vertex,draw=blue,label=below:{$j^*=n$}]
        (js) at (10.8,-2.4) {};
    \node at (5.2,-2.4) {$\cdots$};

    \foreach \ii in {i1,i2,is,in}{
        \foreach \jj in {j1,j2,jn}{
            \draw[ordinary edge] (\ii)--(\jj);
        }
    }
    \draw[blue,line width=1.2pt] (is)--
        node[midway,above,sloped,fill=white,inner sep=1.5pt]
        {$N$} (js);
\end{tikzpicture}
\caption{An instance illustrating the role of
Property~\ref{pro:bb}. There are $n-1$ offline agents and $n$ online
agents. The subgraph induced by $I=[n-1]$ and $J_0=[n-1]$ is complete
bipartite, with every edge having weight $1$. The remaining online agent
$j^*=n$ has the single neighbor $i^*=K+1$, joined by an edge of weight
$N>1$. Vertex labels are fixed independently of arrival order, and
$K=\lfloor n/\sfe\rfloor$.}
\label{fig:comp}
\end{figure}

\begin{example}[An exact solver with vanishing online competitiveness]
\label{exam:ind}
Fix $n\ge6$ and let $K:=\lfloor n/\sfe\rfloor$, so that $K\ge2$.
Consider the graph in Figure~\ref{fig:comp}, with
\[
I=[n-1],
\qquad
J_0=[n-1],
\qquad
J=J_0\cup\{j^*\},
\qquad
j^*=n.
\]
Every pair in $I\times J_0$ is an edge of weight $1$. The only additional
edge is
\[
e^*:=(i^*,j^*),
\qquad
i^*:=K+1,
\qquad
w(e^*)=N>1.
\]
In particular, $n=|J|$, consistently with our notation elsewhere.
The offline optimum includes $e^*$ and $n-2$ unit-weight edges, giving
\[
\OPT=N+n-2.
\]

\xhdr{An arrival-order-dependent exact solver.}
Let $j_\ell$ denote the online agent arriving at time $\ell$.
Define a deterministic prefix solver $\BB_{\mathrm{ord}}$ as follows.
If $j^*\notin\cS_t$, return
\[
\cM_t=\{(\ell,j_\ell):1\le\ell\le t\}.
\]
This is feasible because $t\le n-1$ whenever $j^*\notin\cS_t$.
If $j^*\in\cS_t$, include $e^*$ and match as many agents of
$\cS_t\setminus\{j^*\}$ as possible to distinct vertices of
$I\setminus\{i^*\}$. Thus, in this case,
\[
\cM_t=\{e^*\}\cup\mathcal{Q}_t,
\qquad
|\mathcal{Q}_t|=\min\{t-1,n-2\},
\]
where $\mathcal{Q}_t$ can be chosen deterministically using fixed vertex
labels. These rules specify the solver on the prefixes used in this
example; on other inputs it may return any exact optimum.

If $j^*\notin\cS_t$, the returned matching has weight $t$, which is
optimal. If $j^*\in\cS_t$, its weight is
$N+\min\{t-1,n-2\}$, which is also optimal since $N>1$.
Nevertheless, $\BB_{\mathrm{ord}}$ violates Property~\ref{pro:bb}:
before $j^*$ arrives, it assigns each observed online agent to the
offline vertex indexed by that agent's arrival position.

The same solver is exact for the reweighted prefixes in
Algorithm~\ref{alg:sda}. Indeed, at every matching-phase round
$t>K\ge2$, the prefix contains at least two ordinary online agents,
whether or not $j^*$ has appeared. Hence $d_t(i)\ge2$ for every
$i\in I$, and
\[
a_t(i)=b_t=\frac{K}{t-1}>0
\qquad\text{for all }i\in I.
\]
Consequently, $\widetilde w_t(e)=b_t w(e)$ for every prefix edge,
and the original and reweighted objectives have exactly the same
optimal matchings.

\xhdr{Online performance.}
Let $T$ be the arrival time of $j^*$, and let $\cM$ be the final matching
produced by Algorithm~\ref{alg:sda} using $\BB_{\mathrm{ord}}$ and
$\ta=1/\sfe$.

If $T\le K$, then $j^*$ is discarded during sampling.
If $T=K+1$, then the online matching is still empty, so $e^*$ is
selected and accepted.
If $T>K+1$, then at the first matching-phase round the solver proposes
\[
(K+1,j_{K+1})=(i^*,j_{K+1}),
\]
which is accepted. Thus $i^*$ is already matched when $j^*$ arrives,
and $e^*$ is rejected.
Therefore,
\[
\{e^*\in\cM\}=\{T=K+1\},
\qquad
\Pr[e^*\in\cM]=\frac1n.
\]
There are at most $n-1$ unit-weight edges in any matching, so
\[
\E[w(\cM)]\le\frac{N}{n}+n-1.
\]
It follows that
\[
\frac{\E[w(\cM)]}{\OPT}
\le
\frac{N/n+n-1}{N+n-2}
\le
\frac1n+\frac{n-1}{N}.
\]
For example, taking $N=n^3$ gives
\[
\frac{\E[w(\cM)]}{\OPT}
\le\frac1n+\frac1{n^2}\longrightarrow0
\qquad(n\to\infty).
\]
Thus exact prefix optimization alone does not guarantee constant
online competitiveness when arrival-order-dependent tie-breaking is
allowed. The same argument applies to the original-weight
prefix-optimization framework.
\end{example}

\xhdr{Interpretation in terms of random seeds.}
Although $\BB_{\mathrm{ord}}$ is deterministic, the same construction
can be viewed as coupling a solver's tie-breaking randomness to the
arrival process. Before $j^*$ appears, a solver could select an optimal
matching by ordering the vertices of $\cS_t$ according to a random
permutation and assigning the $\ell$th vertex to offline vertex
$i=\ell$. If that permutation is taken to be the actual arrival order,
then it is uniform conditional on $\cS_t$, but is perfectly coupled to
the ordered prefix. This reproduces $\BB_{\mathrm{ord}}$.
Having the correct marginal distribution for the tie-breaking
randomness is therefore not enough; its dependence on the arrival
order matters. The example concerns independence from arrivals and
does not establish that independence between seeds across calls is
necessary.

\xhdr{Enforcing the black-box convention.}
Randomization is not necessary to enforce Property~\ref{pro:bb}.
For a deterministic solver, supply a canonical representation of the
unordered prefix based on fixed, arrival-independent vertex labels,
and use a fixed deterministic procedure without arrival-order-dependent
state. In particular, an exact solver may select the lexicographically
first maximum-weight matching under a fixed ordering of the labeled
edges, using lexicographic order only to resolve ties in total matching
weight.

For a randomized solver, use the same canonical prefix representation
and supply fresh random seeds that are mutually independent and
independent of the entire arrival order, as specified in
Section~\ref{sec:sda}. Merely adding external randomness does not
suffice if the solver still uses arrival positions, the distinguished
current arrival, or the current online matching to select its output.
The essential requirement is arrival-order independence of the
prefix computation, not randomization itself.


\section{Comparison between $\sm(\BB,\ta)$ with $\BB=\gre$ and $\smg(\ta)$}
\label{app:comp}

Throughout this section, we consider the Black-Box Sampling--Matching
framework $\sm(\BB,\ta)$ instantiated with the Greedy algorithm, i.e.,
$\BB=\gre$. Under this instantiation, $\sm(\gre,\ta)$ runs \gre on the
evolving prefix graph $\cG_t:=\cG(I,\cS_t)$ at every round
$t=K+1,\ldots,n$, using the reweighted edge values
$\widetilde w_t$ defined in~\eqref{eq:sda-reweighting}.
In contrast, $\smg(\ta)$ runs \gre only once on the sampled subgraph
$\cG(I,\cS_K)$ to determine offline prices, and then applies a
threshold-based rule to each subsequent arrival.

The examples below demonstrate that these two frameworks can behave
fundamentally differently, even when both use \gre and process the same
sequence of arrivals. Importantly, the prefix reweighting in
$\sm(\gre,\ta)$ does not alter the Greedy matching in either example,
as verified below.

\begin{figure}[ht!]
\begin{minipage}{.5\linewidth}
 \begin{tikzpicture}[ultra thick]
  
    \draw (1,0) node[minimum size=0.2mm,draw,circle] {};
    \draw (1,0.2) node[above] {$i=1$};
               \draw (1,-2) node[minimum size=0.2mm,draw,circle] {};
    \draw (1,-2.2) node[left] {$j_1$};
   
       \draw (3,0) node[minimum size=0.2mm,draw,circle] {};
    \draw (3,0.2) node[above] {$i=2$};
    
           \draw (3,-2) node[minimum size=0.2mm,draw,circle] {};
    \draw (3,-2.2) node[left] {$j_2$};
    
           \draw (5,-2) node[minimum size=0.2mm,draw,circle] {};
    \draw (5,-2.2) node[left] {$j_3$};
    
           \draw (7,-2) node[minimum size=0.2mm,draw,circle] {};
    \draw (7,-2.2) node[left] {$j_4$};
    
       \draw[-] (1,-0.2)--(1,-1.8);     
    \draw (1,-1.5) node[left]{\bluee{$2$}};
    
          \draw[-] (3,-0.15)--(1,-1.85);      
    \draw (1.8,-1.2) node[below]{\bluee{$1.8$}};
    
      \draw[-] (3,-1.85)--(3,-0.2);  
        \draw (3,-1.5) node[left]{\bluee{$1.5$}};
        
            \draw[-] (4.85,-1.9)--(1,-0.18);  
        \draw (4.6,-1.85) node[above]{\bluee{$3$}};
 
          \draw[-] (6.85,-1.9)--(3,-0.18);  
        \draw (6.6,-1.85) node[above]{\bluee{$1.6$}};
        
           \draw[|<->|] (1,-2.7) -- (3.3,-2.7);
    \draw (2,-3) node[below] {{Sampling Phase}};
               \draw[|->] (4.7,-2.7) -- (7.3,-2.7);
    \draw (6,-3) node[below] {{Matching Phase}}; 
         \end{tikzpicture}
         \end{minipage}\hfill
\begin{minipage}{.5\linewidth}
 \begin{tikzpicture}[ultra thick]
  \draw (1,0) node[minimum size=0.2mm,draw,circle] {};
    \draw (1,0.2) node[above] {$i=1$};
               \draw (1,-2) node[minimum size=0.2mm,draw,circle] {};
    \draw (1,-2.2) node[right] {$j_1$};
   
       \draw (3,0) node[minimum size=0.2mm,draw,circle] {};
    \draw (3,0.2) node[above] {$i=2$};
    
           \draw (3,-2) node[minimum size=0.2mm,draw,circle] {};
    \draw (3,-2.2) node[right] {$j_2$};
    
           \draw (5,-2) node[minimum size=0.2mm,draw,circle] {};
    \draw (5,-2.2) node[right] {$j_3$};
    
           \draw (7,-2) node[minimum size=0.2mm,draw,circle] {};
    \draw (7,-2.2) node[right] {$j_4$};
    
       \draw[-] (1,-0.2)--(1,-1.8);     
    \draw (1,-1.5) node[left]{\bluee{$2$}};

      \draw[-] (3,-1.85)--(3,-0.2);  
        \draw (3,-1.5) node[left]{\bluee{$1.5$}};
        
            \draw[-] (4.85,-1.9)--(1,-0.18);  
        \draw (4.6,-1.85) node[above]{\bluee{$3$}};
 
          \draw[-] (6.85,-1.9)--(3,-0.18);  
        \draw (6.6,-1.8) node[above]{\bluee{$1.6$}};
        
                  \draw[-] (6.85,-1.9)--(1,-0.18);  
        \draw (6.1,-1.65) node[below]{\bluee{$2.5$}};
        
                   \draw[|<->|] (1,-2.7) -- (3.3,-2.7);
    \draw (2,-3) node[below] {{Sampling Phase}};
               \draw[|->] (4.7,-2.7) -- (7.3,-2.7);
    \draw (6,-3) node[below] {{Matching Phase}}; 
         \end{tikzpicture}
\end{minipage}

\caption{
Two examples illustrating behavioral differences between
$\sm(\BB,\ta)$ instantiated with $\BB=\gre$ and $\smg(\ta)$.
The numbers next to the edges denote their original weights, and
$j_t$ denotes the online agent arriving at time $t=1,2,3,4$.
Assume that the sample size is $K=2$.
In the left panel, $\sm(\gre,\ta)$ rejects $j_4$, whereas
$\smg(\ta)$ matches $j_4$ with $i=2$.
In the right panel, $\sm(\gre,\ta)$ matches $j_4$ with $i=2$,
whereas $\smg(\ta)$ rejects $j_4$.
}
\label{fig:diff}
\end{figure}
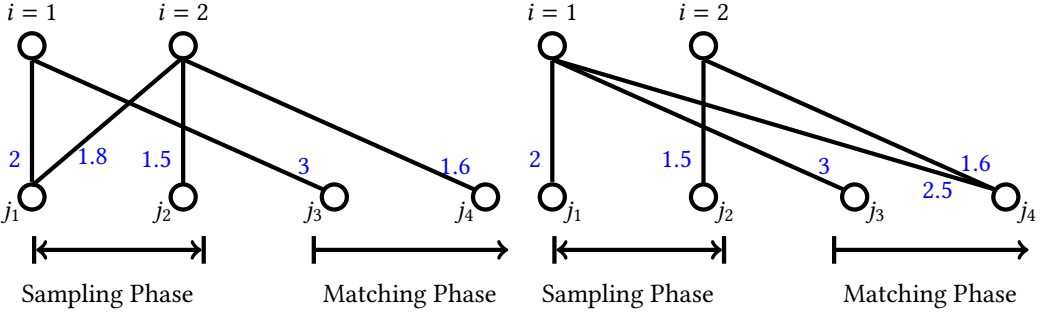

\begin{example}
Consider the two instances shown in Figure~\ref{fig:diff}.
In both cases, there are two offline agents $i=1,2$ and four online
agents arriving in the order $j_1,j_2,j_3,j_4$.
The sampling size is $K=2$, so both $\sm(\gre,\ta)$ and $\smg(\ta)$
discard the first two arrivals $j_1$ and $j_2$.

We first observe that the reweighting used by $\sm(\gre,\ta)$ does not
change any Greedy decision in these examples.
At time $t=3$,
\[
b_3=\frac{K}{t-1}=1,
\]
and hence $a_3(i)=1$ for every offline agent $i$.
Thus,
\[
\widetilde w_3(e)=w(e)
\]
for every edge in the prefix graph $\cG_3$.
At time $t=4$, every offline agent has prefix degree at least two in
both instances. Therefore,
\[
a_4(1)=a_4(2)=b_4=\frac{2}{3},
\]
so that
\[
\widetilde w_4(e)=\frac{2}{3}w(e)
\]
for every edge in $\cG_4$.
Hence the reweighting uniformly rescales all edge weights at time $4$
and leaves the Greedy matching unchanged.

\smallskip
\noindent
\textbf{Left instance.}
At time $t=3$, running \gre on the prefix graph $\cG_3$ selects
$(i=1,j_3)$, so $\sm(\gre,\ta)$ matches $j_3$ with $i=1$.
At time $t=4$, running \gre on $\cG_4$ yields the matching
\[
\{(i=1,j_3),(i=2,j_1)\}.
\]
Thus $j_4$ is unmatched in the prefix matching returned by \gre, and
$\sm(\gre,\ta)$ rejects $j_4$.

In contrast, $\smg(\ta)$ computes its prices by running \gre on the
sampled graph $\cG(I,\cS_2)$.
The resulting matching is
\[
\{(i=1,j_1),(i=2,j_2)\},
\]
and hence
\[
p(i=1)=2,
\qquad
p(i=2)=1.5.
\]
At time $t=3$, $\smg(\ta)$ matches $j_3$ with $i=1$.
Upon the arrival of $j_4$, the edge $(i=2,j_4)$ has weight
\[
1.6>p(i=2)=1.5,
\]
and $i=2$ remains unmatched.
Therefore, $\smg(\ta)$ matches $j_4$ with $i=2$.

\smallskip
\noindent
\textbf{Right instance.}
As in the left instance, at time $t=3$ the Greedy prefix matching
contains $(i=1,j_3)$, so $\sm(\gre,\ta)$ matches $j_3$ with $i=1$.
At time $t=4$, running \gre on $\cG_4$ yields a matching containing
\[
(i=1,j_3)
\qquad\text{and}\qquad
(i=2,j_4).
\]
Since $i=2$ is still unmatched in the online matching, $\sm(\gre,\ta)$
matches $j_4$ with $i=2$.

Under $\smg(\ta)$, the sampled graph yields the same prices
\[
p(i=1)=2,
\qquad
p(i=2)=1.5.
\]
At time $t=3$, $\smg(\ta)$ matches $j_3$ with $i=1$.
When $j_4$ arrives, both edges $(i=1,j_4)$ and $(i=2,j_4)$ satisfy
their respective price thresholds, but
\[
w(i=1,j_4)=2.5
>
1.6=w(i=2,j_4).
\]
Hence the thresholded max-weight rule first selects $(i=1,j_4)$.
Since $i=1$ is already matched, $\smg(\ta)$ rejects $j_4$ rather than
falling back to the lower-weight qualifying edge $(i=2,j_4)$.

These examples therefore show that $\sm(\gre,\ta)$ and $\smg(\ta)$
can make qualitatively different online decisions even though both use
the Greedy algorithm as an offline subroutine. The distinction arises
from their fundamentally different use of Greedy: $\sm(\gre,\ta)$
recomputes a reweighted prefix matching at every round, whereas
$\smg(\ta)$ computes prices once from the initial sample and subsequently
uses a local thresholded max-weight rule.
\hfill$\blacksquare$
\end{example}

\end{document}